\RequirePackage{fix-cm}
\documentclass[smallextended]{svjour3}       % onecolumn (second format)
\smartqed  % flush right qed marks, e.g. at end of proof
\usepackage{acronym}
\usepackage{algpseudocode}
\usepackage{amsfonts}
\usepackage{amsmath}
\usepackage{amssymb}
\usepackage[american]{babel}
\usepackage[format=plain,position=bottom,labelfont=bf,small,textfont=small,labelsep=quad]{caption}
\usepackage{dsfont}
\usepackage{fixltx2e}
\usepackage{float}
\usepackage[T1]{fontenc}
\usepackage{gensymb}
\usepackage{graphicx}
\usepackage[utf8x]{inputenc}
\usepackage{siunitx}
\usepackage{subfigure}
\usepackage{wrapfig}
\usepackage{xspace}
\usepackage[all]{xy}
\graphicspath{{./graphics/}}

\newcommand{\old}[1]{{}}
\newcommand{\todo}[1]
{
    \marginpar
        {\bf {[!!!]}}
        {\bf {#1}}
}
\newcommand{\wo}{\setminus}
\newcommand{\agp}{\textnormal{AGP}}
\newcommand{\agr}{\textnormal{AGR}}

\newcommand{\bigO}{\operatorname{O}}
\newcommand{\bigOmega}{\operatorname{\Omega}}
\newcommand{\conv}{\textnormal{conv}}
\newcommand{\dcup}{\mathbin{\dot{\cup}}}
\newcommand{\kernel}{\textnormal{kernel}}
\newcommand{\N}{\mathds{N}}
\newcommand{\one}{\mathds{1}}

\newcommand{\R}{\mathds{R}}

\newcommand{\V}{\mathcal{V}}
\newcommand{\zero}{{\bf 0}}

\newcommand{\NP}{NP\xspace}

\newcommand{\ie}{i.\,e.\xspace}

\newcommand{\OBdA}{Without loss of generality\xspace}

\def\vis#1{\mathcal{V}(#1)}

\renewenvironment{proof}{\noindent{\it Proof.}}{\qed\bigskip}

\floatstyle{boxed}
\newfloat{algorithm}{tbp}{loa}
\floatname{algorithm}{Algorithm}
\algnewcommand\Input{\item[\textbf{Input:}]}

\acrodef{AG}{Art Gallery}\acused{AG}
\acrodef{AGP}{Art Gallery Problem}\acused{AGP}
\acrodef{AGR}{Art Gallery Relaxation}\acused{AGR}
\acrodef{AGR2}[AGR\textsubscript{2}]{2-Guard Art Gallery Relaxation}\acused{AGR2}
\acrodef{CLIQUE}{Clique Decision Problem}\acused{CLIQUE}
\acrodef{EC}{Edge Cover}\acused{EC}
\acrodef{FC}{Full Circulant Decision Problem}\acused{FC}
\acrodef{FEC}{Fractional Edge Cover}\acused{FEC}
\acrodef{FECG}{Fractional Edge Cover Graph}\acused{FECG}
\acrodef{FECSG}{Fractional Edge Cover Support Graph}\acused{FECSG}
\acrodef{IP}{Integer Program}\acused{IP}
\acrodef{LP}{Linear Program}\acused{LP}
\acrodef{Q0}[Q\textsubscript{0}]{Minimum}\acused{Q0}
\acrodef{Q1}[Q\textsubscript{1}]{First Quartile}\acused{Q1}
\acrodef{Q2}[Q\textsubscript{2}]{Second Quartile}\acused{Q2}
\acrodef{Q3}[Q\textsubscript{3}]{Third Quartile}\acused{Q3}
\acrodef{Q4}[Q\textsubscript{4}]{Maximum}\acused{Q4}
\acrodef{SC}{Set Cover}\acused{SC}
\acrodef{SCP}{Set Covering Problem}\acused{SCP}

\begin{document}

\title{Facets for Art Gallery Problems\thanks{%
	Alexander Kr\"oller was partially supported by DFG project Kunst!, KR 3133/1-1.
	Research was conducted while Stephan Friedrichs and Christiane Schmidt were affiliated with TU Braunschweig.
	Christiane Schmidt is supported by the Israeli Centers of Research Excellence (I-CORE) program (Center No. 4/11).}}

\author{S{\'a}ndor P.~Fekete \and Stephan Friedrichs \and Alexander Kr\"oller \and Christiane Schmidt}

\authorrunning{Fekete, Friedrichs, Kr\"oller, and Schmidt} % if too long for running head

\institute{S{\'a}ndor P.~Fekete / Alexander Kr\"oller \at
Technische Universit\"at Braunschweig \\
Institut f\"ur Betriebssysteme und Rechnerverbund \\
%M\"uhlenpfordtstraße 23 \\
%D-38106 Braunschweig \\
%Tel.: +49 531 3913110 \\
%Fax: +49 531 3913109 \\
\email{\{s.fekete, a.kroeller\}@tu-bs.de}
\and
Stephan Friedrichs \at
Max Planck Institute for Informatics, Saarbr\"ucken, Germany \\
\email{sfriedri@mpi-inf.mpg.de}
\and
Christiane Schmidt \at
The Rachel and Selim Benin School of Computer Science and Engineering\\
The Hebrew University of Jerusalem\\
%9190401 Jerusalem\\
%Supported by the Israeli Centers of Research Excellence (I-CORE) program (Center No. 4/11).\\
\email{cschmidt@cs.huji.ac.il}
}

\date{Received: date / Accepted: date}

\maketitle

\begin{abstract}
The {\sc Art Gallery Problem} (AGP) asks for placing a minimum number of stationary
guards in a polygonal region $P$, such that all points in $P$ are guarded.
The problem is known to be NP-hard, and its inherent continuous structure
(with both the set of points that need to be guarded and the set of
points that can be used for guarding being uncountably infinite)
makes it difficult to apply a straightforward formulation as an Integer
Linear Program.  We use an iterative primal-dual relaxation approach
for solving AGP instances to optimality.
At each stage, a pair of LP relaxations for a finite candidate subset of
primal covering and dual packing constraints and variables is
considered; these correspond to possible guard positions and
points that are to be guarded.

Particularly useful are cutting planes for eliminating fractional solutions.
We identify two classes of facets, based on {\sc Edge Cover} and {\sc Set Cover} (SC) inequalities.
Solving the separation problem for the latter is NP-complete, but exploiting
the underlying geometric structure, we show that large subclasses of
fractional SC solutions cannot occur for the AGP.
This allows us to separate the relevant subset of facets in polynomial time.
We also characterize all facets for finite AGP relaxations with coefficients
in $\{0, 1, 2\}$.

Finally, we demonstrate the practical usefulness of our approach.
Our cutting plane technique yields a significant improvement in terms of speed and solution quality due to considerably reduced integrality gaps as compared to the approach by Kröller et~al.~\cite{kbfs-esbgagp-12}.

\end{abstract}
 % keywords are in abstract.tex
\section{Introduction}\label{sec:intro}

The {\sc Art Gallery Problem} (AGP) is one of the classical problems of geometric optimization:
given a polygonal region $P$ with $n$ vertices, find as few stationary guards as possible, such that any
point of the region is visible by at least one of the guards. As first proven by Chv\'atal~\cite{c-actpg-75}
and then shown by Fisk~\cite{f-spcwt-78}
in a beautiful and concise proof (which is highlighted in the shortest chapter in ``Proofs from THE BOOK''~\cite{thebook}),
$\left\lfloor\frac{n}{3}\right\rfloor$ guards are sometimes necessary and always sufficient
when $P$ is a simple polygon. Worst-case bounds of this type are summarized under the name ``Art-Gallery-type
theorems'', and used as a metaphor even for unrelated problems; see O'Rourke~\cite{r-agta-87} for
an early overview, and Urrutia~\cite{u-agip-00} for a more recent survey.

Algorithmically, the AGP is closely related to the {\sc Set Cover} (SC) problem:
All points in $P$ have to be covered by star-shaped subregions of $P$.
The AGP is NP-hard, even for a simply connected polygonal region $P$~\cite{ll-ccagp-86}.
However, the SC problem has no underlying geometry, and it is well known that
geometric variants of problems may be easier to solve or approximate than their
discrete, graph-theoretic counterparts, so it is natural to explore ways to
exploit the geometric nature of the AGP.
But the AGP is far from being easily discretized, as both the set to be covered
(all points in $P$) as well as the covering family (all star-shaped subregions
around some point of $P$) usually are uncountably infinite.

%There are, however, two differences to a discrete SC problem.
%On the one hand, it is well known that geometric variants of problems may be easier to solve or approximate than their
%discrete, graph-theoretic counterparts, so it is natural to explore ways to exploit the geometric nature of the AGP;
%on the other hand, the AGP is far from being discrete,
%as both the set to be covered (all points in $P$) as well as the covering family
%(all star-shaped subregions around some point of $P$) usually are uncountably infinite.

It is natural to consider more discrete versions of the AGP. Ghosh~\cite{g-aaagp-10} showed that restricting
possible guard positions to the $n$ vertices, \ie, the AGP with vertex guards, allows an
$\bigO(\log n)$-approximation algorithm of complexity $\bigO(n^5)$; conversely, Eidenbenz et al.~\cite{esw-irgpt-01}
showed that for a region with holes, finding an optimal set of vertex guards is at least as hard as SC, so there is little
hope of achieving a better approximation guarantee than $\bigOmega(\log n)$.
While these results provide tight bounds in terms of approximation, they do by no means close the
book on the arguably most important aspect of mathematical optimization: combining structural insights with powerful mathematical
tools in order to achieve provably optimal solutions for instances of interesting size.
Moreover, even a star-shaped polygon may require a large number of vertex
guards, so general AGP instances may have significantly
better solutions than the considerably simpler discretized version with vertex guards.

\subsection{Solving AGP Instances}

Computing optimal solutions for general AGP instances is
not only relevant from a theoretical point of view, but has also gained in practical importance
in the context of modeling, mapping and surveying complex environments, such as in the fields of architecture
or robotics and even medicine, which are seeking to exploit the ever-improving capabilities of computer vision and
laser scanning. Amit, Mitchell and Packer~\cite{amp-lgvcp-10}
have considered purely combinatorial primal and dual heuristics for general AGP instances.
Only very recently have researchers begun to combine methods from
integer linear programming with non-discrete geometry in order to obtain optimal solutions.
As we have shown in \cite{kbfs-esbgagp-12}, it is possible to
combine an iterative primal-dual relaxation approach with structures from computational geometry in order to
solve AGP instances with unrestricted guard positions; this approach is based on considering a sequence
of primal and dual subproblems, each with a finite number of primal variables (corresponding to guard positions)
and a finite number of dual variables (corresponding to ``witness'' positions).

Couto et al.~\cite{csr-eeaoagp-07,csr-eeeaoagp-08,crs-exmvg-11} used a similar approach for the AGP with vertex guards.
Tozoni et al.~\cite{DaviPedroCid-OO2013}
proposed and algorithm computes
lower and upper bounds for the AGP, based on
computing finite set-cover instances with the help of a
state-of-the-art IP solver.
To generate a lower bound, a finite set of witness candidates is chosen
and a restricted AGP is solved, in which only the witnesses
have to be covered. For this, it suffices
to extract a finite set of potential guard positions
from the visibility arrangement of the witness set in order to
ensure optimality.
Similarly, finite sets of potential witness positions
for a given finite guard set can be
extracted from the visibility arrangement of the guards.
This allows it to compute upper and lower bounds for the optimal
AGP value by solving discrete set cover instances.
The algorithm  of \cite{DaviPedroCid-OO2013} iterates between  generating tighter
lower and  upper bounds  by refining the  witness and  guard candidate
sets along the iterations.
It stops when lower and upper bounds coincide.
Although no theoretical convergence has been established, in tests, the
approach is able to yield optimal solutions
for a  large variety of  instance classes,  even for
polygons      with      up       to      a      thousand      vertices.

An approach presented in \cite{kbfs-esbgagp-12}
considers a similar primal-dual scheme, but focuses on the
linear relaxation of the primal guard cover, whose dual is the witness packing problem.
This forms the basis of integer solutions and the approach presented in this paper; more details are described
in Section~\ref{sec:model}. Furthermore, we have collaborated with the authors of
\cite{csr-eeaoagp-07,csr-eeeaoagp-08,crs-exmvg-11,DaviPedroCid-OO2013}
and produced a video \cite{bdd-pgpcs-13} that highlights and illustrates the approaches to the AGP,
and also demonstrates its relevance for practical applications.

\subsection{Set Cover}

Also important for the work on the AGP is the discrete and finite problem of covering a
given set of objects by an inexpensive collection of subsets. This is known as the
{\sc Set Cover Problem} (SC), which has enjoyed a considerable amount of attention.
Highly relevant for the purposes of this paper
is the work by Balas and Ng~\cite{bn-scp-89} on the discrete SC polytope,
which describes all its facets with coefficients in $\{0, 1, 2\}$.

\subsection{Our Results}

\begin{figure}
        \centering
        %\subfigure[Solution without Cuts]{
                \def\svgwidth{.48\linewidth}
                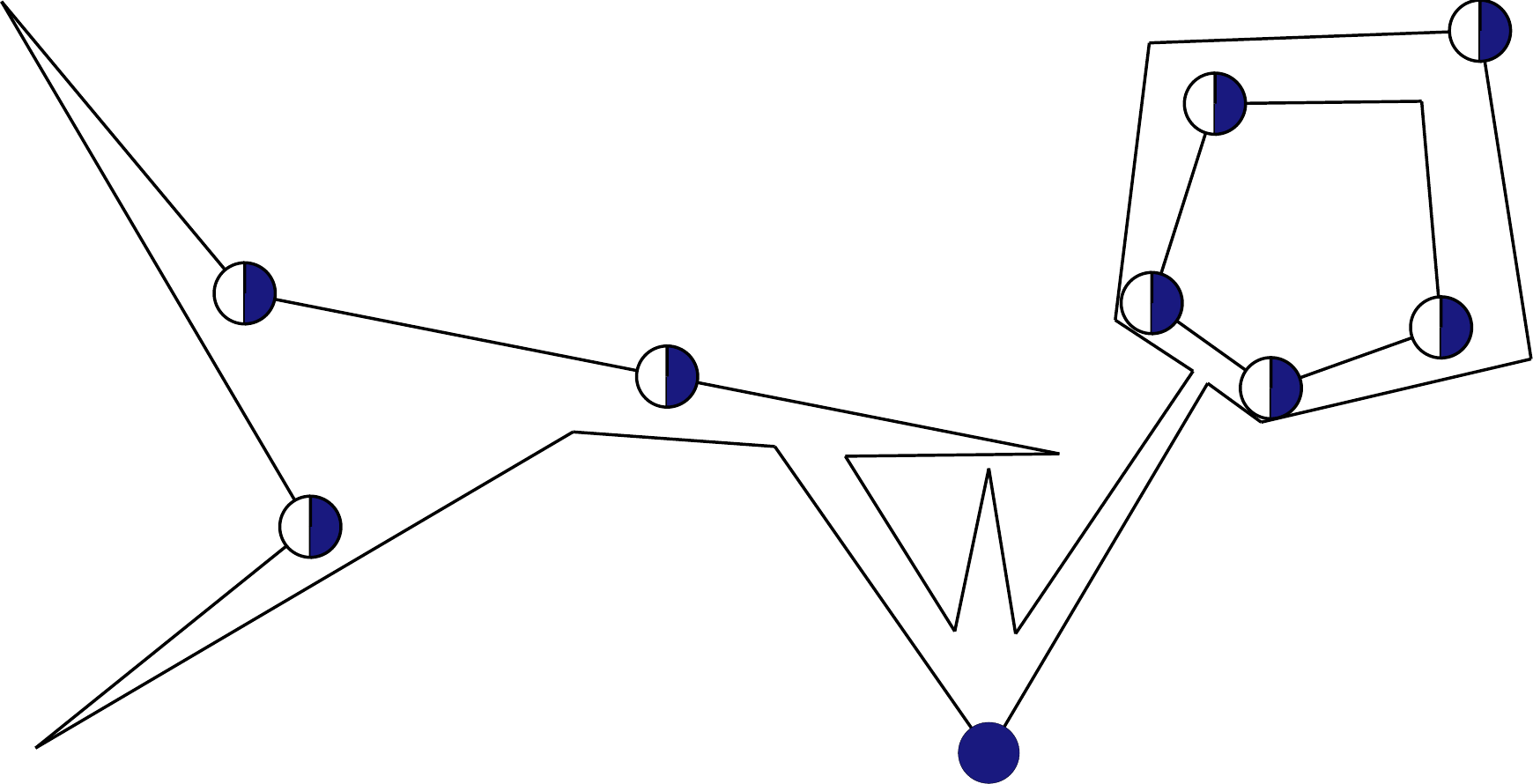
                \label{fig:testcase-fractional}
        %}
        %\subfigure[Solution with Cuts]{
                \def\svgwidth{.48\linewidth}
                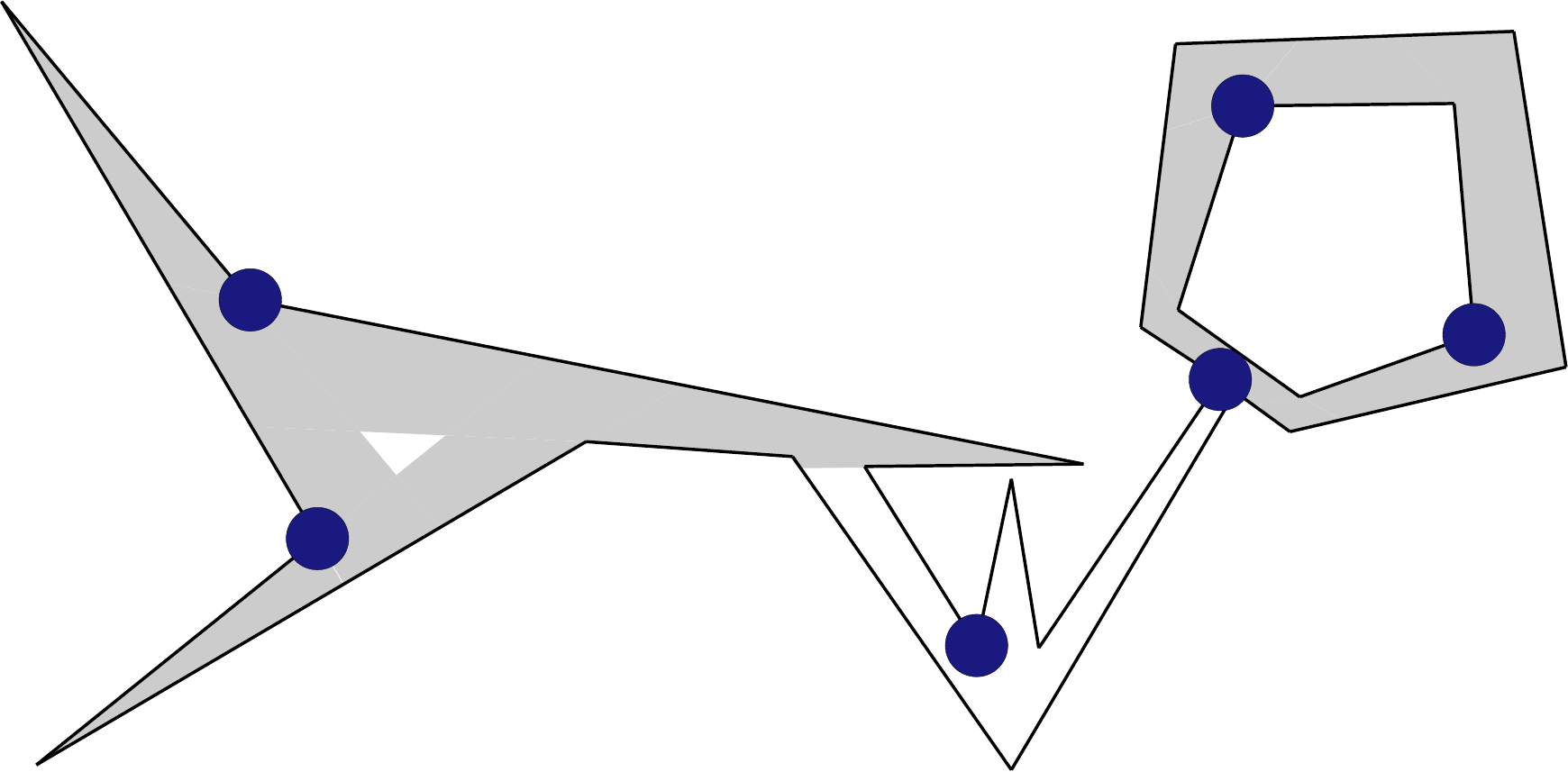
                \label{fig:testcase-cuts}
        %}
        \caption[Solution without and with Cuts]{%
        An optimal fractional solution of value 5 without (left) and an optimal integer solution of value 6
        with cutting planes (right).
        Circles show guards, fill-in indicates fractional amount.
        %The fractional solution shown in the left figure  uses five guards.
        Cutting planes enforce at least two guards in the left
        and three in the right area, both marked in gray.
        %This cuts off the fractional solution and enforces one with six guards,
        %\ie, one more than above, see Figure~\subref{fig:testcase-cuts}.
        %It is both integral and optimal.
        }
        \label{fig:testcase}
%\vspace*{-5mm}
\end{figure}

In this paper, we extend and deepen our recent work~\cite{kbfs-esbgagp-12} on iterative primal-dual
relaxations, by proving a number of polyhedral properties of the resulting AGP
polytopes and integrating them into modified versions of the algorithm presented in~\cite{kbfs-esbgagp-12}.
We provide the first study of the art gallery polytope and give a full characterization of all
its facets with coefficients in $\{0, 1, 2\}$.

Remarkably, we are able to exploit geometry to prove that only a very restricted
family of facets of the general SC polytope will typically have to be used as
cutting planes for removing fractional variables.
Instead, we are able to prove that many fractional solutions only occur in
intermittent SC subproblems; thus, they simply vanish when new guards or
witnesses are introduced.
This saves us the trouble of solving an NP-complete separation problem.
%We also demonstrate that this has real-world consequences by putting the
%resulting cutting planes to practical use:
Computational results illustrate greatly reduced integrality gaps for a wide
variety of benchmark instances, as well as reduced solution times.
Details are as follows.
%Appendix~\ref{app:proofs}.
Related SC results are described by Balas et~al.~\cite{bn-scp-89}.
%Further details on SC by~\cite{bn-scp-89} and our experiments are in
%Appendices~\ref{app:balas} and~\ref{app:diagrams}, respectively.

\begin{itemize}
\item We provide two variants of our primal-dual framework for solving the AGP.
	Both aim at producing binary solutions, one integrates an IP in the primal phase and both greatly benefit from our cutting planes.
	Our algorithms also serve as benchmark for the cutting plane approach in our experiments.
\item We show how to employ cutting planes for an iterative primal-dual framework for solving the
AGP. This is interesting in itself, as it provides an approach to tackling
optimization problems with infinitely many constraints and variables.
The particular challenge is to identify constraints that remain valid for any
choice of infinitely many possible primal and dual variables, as we are not solving one particular
IP, but an iteratively refined sequence.
\item Based on a geometric study of the involved SC constraints, we characterize
all facets of involved AGP polytopes that have coefficients in $\{0, 1, 2\}$.
In the SC setting, these facets are capable of cutting off fractional solutions,
but the separation problem is NP-complete.
We use geometry to prove that only some of these facets are able to cut off
fractional solutions in an AGP setting under reasonable assumptions, allowing us
to solve the separation problem in polynomial time.
\item
	We provide a class of facets based on {\sc Edge Cover} (EC)
	constraints.
%\item We employ geometry to prove that a considerable number of facets can be used for replacing sets of fractional
%variables by new, integer variables, instead of just cutting off the fractional solution and re-solving the LP.
%This differs from the SC, in which all those facets are needed.
%\item
%\todo{Just draft. Comments?} Gut! SF
%One class of discussed facets originates from the SC polytope.

\item We demonstrate the practical usefulness of our results by showing greatly improved solution speed and quality for
a wide array of large benchmarks.% instances.
\end{itemize}

%Because of limited space, technical proofs as well as deeper aspects of previous work and experimental results
%are given in Appendices A, B, C.

%The rest of this paper is organized as follows.
%We describe the basic primal-dual relaxation framework in Section~\ref{sec:model}.
%Section~\ref{sec:sc} describes our main results that are based on facets of the
%SC polytope; Section~\ref{sec:ec} does the same for edge cover.
%Section~\ref{sec:exp} presents our computational results, and
%Section~\ref{sec:concl} provides some final conclusions.
%%Due to space limitations, a number of technical details have been moved to
%appendices.

\section{Preliminaries}
\label{sec:prelim}

We consider a polygonal region $P$ with $n$ vertices that may have holes, \ie, that does not have to be simply connected.
For a point $p\in P$, we denote by $\V(p)$ the {\em visibility polygon} of $p$
in $P$, \ie, the set of all $q\in P$, such that the straight-line connection
$\overline{pq}$ lies completely in $P$.
$P$ is \emph{star-shaped} if $P = \V(p)$ for some $p \in P$.
The set of all such points is the \emph{kernel} of $P$, denoted by $\kernel(P)$.
For a set $S\subseteq P$, $\V(S):=\cup_{p\in S} \V(p)$.

A set $C\subseteq P$ is a {\em guard cover} of $P$, if $\V(C)=P$.
The AGP asks for a guard cover
of minimum cardinality $c$; this is the same as covering $P$ by a minimum number
of star-shaped sub-regions of $P$.
Note that Chvátal's Watchman Theorem~\cite{c-actpg-75} guarantees
$c \leq \left\lfloor\frac{n}{3}\right\rfloor$.
%A set of points on the plain is in \emph{general position}, if no three points
%are on a straight line.
For simplicity, we abbreviate $x(G) := \sum_{g \in G} x_g$, for any vector $x$.

\section{Mathematical-Programming Formulation and LP-Based Solution Procedure}
\label{sec:model}%
\def\AG{\agp}
\def\AGP{\agp}
\def\AGR{\agr}

In order to keep this work self-contained, we briefly recapitulate our previously published~\cite{kbfs-esbgagp-12}
LP formulations of the AGP as well as how to use them to obtain fractional
optimal Art Gallery solutions.
Then we motivate the necessity to integrate cutting planes to cut off those
fractional solutions in order to obtain binary ones.
Furthermore, we specify requirements for cutting planes, allowing us to
seamlessly integrate them in our framework.

Let $P$ be a polygon and $G,W \subseteq P$ sets of points for possible guard
locations and \emph{witnesses}, \ie, points to be guarded, respectively.
We assume $W \subseteq \V(G)$, which is easily guaranteed by initially including all vertices of $P$ in $G$.
The AGP that only requires covering $W$ exclusively using guards in $G$ can be formulated as an IP
denoted by $\agp(G,W)$:
\begin{alignat}{3}
  \text{\makebox[3em][l]{min}} & \sum_{g\in G} x_g \label{eq.ipf.obj}\\
  \text{\makebox[3em][l]{s.\,t.}} & \sum_{g\in G\cap\vis{w}} x_g \geq 1 &\;\;& \forall w\in W\label{eq.ipf.cover}\\
           &  x_g \in \{0,1\}       &    & \forall g\in G,\label{eq.ipf.bin}
\end{alignat}
where the original AGP is $\agp(P,P)$.
Chvátal's Watchman Theorem~\cite{c-actpg-75} guarantees that only a finite
number of variables in $\agp(P,P)$ are non-zero, but it still
has uncountably many variables and constraints,
so it cannot be solved directly.
Thus we consider finite $G,W \subset P$ and iteratively solve $\agp(G,W)$ while adding points to $G$ and $W$.
For dual separation and to generate lower bounds,
we require the LP relaxation $\agr(G,W)$ obtained by relaxing the
integrality constraint~\eqref{eq.ipf.bin} to:
\begin{equation}
0\leq x_g\leq 1 \quad\forall g\in G.\label{eq.lpf.var}
\end{equation}
The dual of $\agr(G,W)$ is
\begin{alignat}{3}
  \text{\makebox[3em][l]{max}}    & \sum_{w \in W} y_w \\
  \text{\makebox[3em][l]{s.\,t.}} & \sum_{w \in W \cap \vis{g}} y_w \leq 1 &\;\;& \forall g \in G \\
                                  &  0 \leq y_w \leq 1                     &    & \forall w \in W.
\end{alignat}
The algorithms based on this formulation and the following argumentation are presented in pseudocode in Section~\ref{sec:algorithms} (Algorithms~\ref{alg:lp} and~\ref{alg:ip}).

The relation between a solution of $\agr(G,W)$ and $\agr(P,P)$ is not obvious,
see Figure~\ref{fig:agp-relaxations} for the following argumentation.
In~\cite{kbfs-esbgagp-12}, we show that $\agr(P,P)$ can be
solved optimally for many problem instances by using finite $G$ and $W$.
The procedure uses primal/dual separation (\ie, cutting planes and column
generation) to connect $\agr(G,W)$ to $\agr(P,P)$:

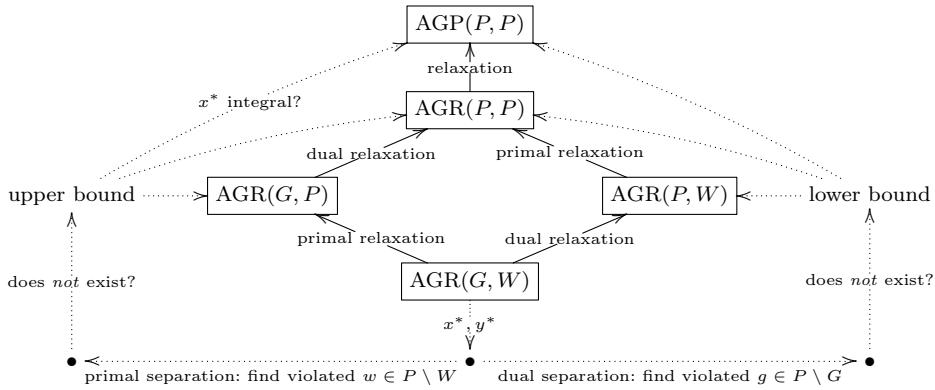
\begin{figure}
	\xymatrix@R-2mm{
	&
		& *+[F]{\agp(P,P)}
		&
		&
		\\
	&
		& *+[F]{\agr(P,P)} \ar[u]|-{\textnormal{relaxation}}
		&
		&
		\\
	\textnormal{upper bound}
			\ar@{.>}[r]
			\ar@{.>}@/^/[urr]
			\ar@{.>}@/^/[uurr]|-{\textnormal{$x^*$ integral?}}
		& *+[F]{\agr(G,P)} \ar[ur]|-{\textnormal{dual relaxation}}
		&
		& *+[F]{\agr(P,W)} \ar[ul]|-{\textnormal{primal relaxation}}
		& \textnormal{lower bound}
				\ar@{.>}[l]
				\ar@{.>}@/_/[ull]
				\ar@{.>}@/_/[uull]
		\\
	&
		& *+[F]{\agr(G,W)}
				\ar[ur]|-{\textnormal{dual relaxation}}
				\ar[ul]|-{\textnormal{primal relaxation}}
				\ar@{.>}[d]|-{x^*,\,y^*}
		&
		&
		\\
	\bullet \ar@{.>}[uu]|-{\textnormal{does \emph{not} exist?}}
		&
		& \bullet
				\ar@{.>}[ll]^{\textnormal{primal separation: find violated $w \in P \setminus W$}}
				\ar@{.>}[rr]_{\textnormal{dual separation: find violated $g \in P \setminus G$}}
		&
		& \bullet \ar@{.>}[uu]|-{\textnormal{does \emph{not} exist?}}
	}
	\caption[The \acl{AGP} and its Relaxations]{%
	The \ac{AGP} and its relaxations for $G, W \subseteq P$.
%	Continuous arrows denote relaxations:
%	Primal relaxations leave out constraints, dual relaxations omit
%	variables (\ie, dual constraints).
	Dotted arrows represent which %the circumstances under which certain
	conclusions may be drawn from the primal and dual solutions $x^*$ and $y^*$.
	}
	\label{fig:agp-relaxations}
%\vspace*{-5mm}
\end{figure}

For some finite sets $G$ and $W$, we solve $\agr(G,W)$ using the
simplex method. This produces an optimal primal solution $x^*$ and
dual solution $y^*$ with objective value $z^*$.
The primal is a minimum covering of $W$ by the guards in $G$, the dual a maximum
packing of witnesses in $W$, such that each guard in $G$ sees at most one of
them.
%Now there are three cases:
We analyze $x^*$ and $y^*$ as follows:
\begin{enumerate}
\item If there exists a point $w\in P\wo W$ with $x^*(G \cap \vis{w})
  < 1$, then $w$ corresponds to an inequality of $\AGR(P,P)$ that is
  violated by $x^*$. The new witness $w$ is added to $W$, and the LP is
  re-solved. If such a point $w$ cannot be found, $x^*$ is optimal for
  $\AGR(G,P)$, and $z^*$ is an upper bound for $\AGR(P,P)$.
\item If there exists a point $g\in P\wo G$ with $y^*(W \cap
  \vis{g})>1$, then it corresponds to a violated dual inequality of
  $\AGR(P,P)$. We create the LP column for $g$ and re-solve the LP. If
  such a $g$ does not exist, $y^*$ is an optimal dual solution for
  $\AGR(P,W)$ and $z^*$ is a lower bound for $\AGR(P,P)$.
%\item If 1. is not the case and there is a lower bound of $z^*$, then
%  $x^*$ is optimal for $\AGR(P,P)$. The algorithm terminates.
\end{enumerate}

%The two cases reflect the primal and dual separation problems, asking
%for the existence of a violated $w$ and $g$.
\old{
Both separation problems can be solved efficiently, using the following geometric interpretation.
The overlay of the visibility polygons of all points $g\in G$ with
$x^*_g>0$  decomposes $P$ into a planar arrangement of bounded complexity. The
coverage function $c(p):=\sum_{g\in G \cap \vis{p}}x^*(p)$ is constant over
every face and edge of the arrangement. Thus,
primal separation simply requires checking all regions of the arrangement for
coverage less than $1$,
%Consequently, an algorithm for
%xxxthe primal separation problem
%xxxsimply has to iterate over all faces, edges, and vertices of the
%arrangement to identify one point in which the coverage value is less
%than $1$.
while
dual separation
means looking for coverage larger than $1$.
%problem can
%be solved in the same manner as the primal (looking for coverage
%higher than $1$).
}

Both separation problems can be solved efficiently using the overlay
of the visibility polygons of all points $g\in G$ with
$x^*_g>0$ (for the primal case) and all $w \in W$ with $y^*_w > 0$ (for the dual
case), which decomposes $P$ into a planar arrangement of bounded complexity.
%Primal separation is answered by finding an element with coverage of less than
%$1$ in that overlay, dual by identifying one with a coverage of more than $1$.

Should the upper and the lower bound meet, we have an optimal solution of
$\agr(P,P)$, but $\agr(P,P)$ is the LP relaxation of $\agp(P,P)$, so its optimal
solution may contain fractional guard values~\cite{kbfs-esbgagp-12},
compare Figure~\ref{fig:testcase}.
At this point, it is possible to solve $\agp(G,P)$ using primal separation only,
which produces binary upper bounds; but they do not necessarily match the lower
bounds, which are still obtained using the relaxation.
This scenario can prevent our procedure from terminating, even if it found an
optimal Art Gallery solution, because it might be unable to prove its optimality.
Algorithm~\ref{alg:ip} in Section~\ref{sec:algorithms} explores that approach.

In the remainder of this paper, we explore the use of cutting planes to cut
off large classes of fractional solutions obtained by a procedure like the one
described above, increasing lower bounds and enhancing integrality.
Let $\alpha$ be such a cutting plane.
Recall that $\agp(P,P)$ has an infinite number of both variables and constraints.
That means that it is not enough for $\alpha$ to be feasible for $\agp(G,W)$
for the current iteration's finite sets $G$ and $W$;
$\alpha$ must remain feasible in all future iterations of our algorithm.
Formally, feasibility for $\agp(G,W)$ is insufficient; instead, we require $\alpha$
not to cut off any $x \in \{0,1\}^{G'}$ for an arbitrary $P \supseteq G' \supseteq G$,
such that $x$ is
feasible for $\agp(G',P)$.
%
%for \emph{any}
%choice of $G$ with $x \in \{0,1\}^G$ that is feasible for $\agp(G,P)$,
%$\alpha$ must not cut off $x$.
%As $\alpha$ has to assign new coefficients to new variables in new iterations,
%we treat it as a function $\alpha:P \to \mathbb{R}$, rather than a vector.
%\todo{Ist das so OK? Ali?}
An LP with a set $A$ of such additional constraints is denoted by $\agr(G,W,A)$,
its IP counterpart by $\agp(G,W,A)$.
Note that $\agp(G,P)$ and $\agp(G,P,A)$ have the same set of feasible solutions.
By $\agp(G,W)$, we sometimes denote the set of its feasible solutions rather
than the IP itself, as in $\conv(\agp(G,W))$.
See Section~\ref{sec:algorithms} for LP- and IP-based algorithms using the framework presented in this section.

%%% Local Variables:
%%% mode: latex
%%% TeX-master: "main"
%%% End:

\section{Algorithms}
\label{sec:algorithms}

The algorithm of Kröller et~al.~\cite{kbfs-esbgagp-12} produces fractional solutions of the AGP.
We present two modifications, Algorithms~\ref{alg:lp} and~\ref{alg:ip}, focused on obtaining binary solutions.

Our first modification, used in both algorithms, is that we do not run primal and dual separation, compare Section~\ref{sec:model}, in every iteration.
Instead, we repeatedly run primal (dual) separation until a primally (dually) feasible solution has been obtained and then switch to running dual (primal) separation until a feasible dual (primal) solution has been found, and so on.
We call these phases \emph{primal (dual) phases} and repeat an alternating sequence of them, until primally and dually feasible solutions with matching bounds have been found.

\subsection{LP Mode}

Algorithm~\ref{alg:lp} relies on cutting planes to cut off fractional solutions that are feasible for $\agr(G,P)$, but not $\agp(G,P)$.
Those cutting planes are constraints in the primal LP, and variables in the dual.
This means that they have two effects:
They enhance the integrality of acquired solutions and they increase the lower bound.

The issue with this approach is that we are not guaranteed to find a binary solution, because we might not have a cutting plane available which is able to cut off the current primal solution.

\begin{algorithm}
	\begin{algorithmic}[1]
	\Input Polygon $P$
	\State $G \gets W \gets \textnormal{all vertices of $P$}$
	\State $A \gets \emptyset$
	\State $(\textnormal{lowerBound}, \textnormal{upperBound})
			\gets (1, \infty)$
	\Statex
	\Repeat
		\Repeat
			\State $(x^*, y^*) \gets$ optimize $\agr(G,W,A)$
			\State $W \gets W \cup \textnormal{run primal separation}$
			\State $A \gets A \cup \textnormal{separate cuts}$\label{alg:lp-cuts-1}
			\If{separation failed \textbf{and} $x^*$ is integral}
				\State $\textnormal{upperBound} \gets \min(\textnormal{upperBound},
						\textnormal{objective value of $x^*$})$
			\EndIf
		\Until{separation failed \textbf{or} $\textnormal{lowerBound} = \textnormal{upperBound}$}

		\Repeat
			\State $(x^*, y^*) \gets$ optimize $\agr(G,W,A)$
			\State $G \gets G \cup \textnormal{run dual separation}$
			\State $A \gets A \cup \textnormal{separate cuts}$\label{alg:lp-cuts-2}
			\If{separation failed}
				\State $\textnormal{lowerBound} \gets \max(\textnormal{lowerBound},
						\lceil\textnormal{objective value of $y^*$}\rceil)$
			\EndIf
		\Until{separation failed \textbf{or} $\textnormal{lowerBound} = \textnormal{upperBound}$}
	\Until{$\textnormal{lowerBound} = \textnormal{upperBound}$ \textbf{or} time limit reached}
	\end{algorithmic}
	\caption[\acs{LP} Mode Algorithm]{%
	The \acs{LP} mode algorithm only solves \acp{LP}.
	}
	\label{alg:lp}
\end{algorithm}

\subsection{IP Mode}

The typical approach of eliminating fractional solutions in linear optimization is to employ an integer program (IP).
In Algorithm~\ref{alg:ip}, we solve $\agp(G,W,A)$ for finite $G, W \subset P$ and iteratively apply primal separation to the result, which produces feasible binary solutions.

Unfortunately, this procedure does not necessarily find optimal solutions of $\agp(P,P)$, because it does not generate new guard positions:
For generating guards we need a dual solution, which an IP cannot provide.
To counter that, we use the dual phase of Algorithm~\ref{alg:lp} where we solve the LP $\agr(G,W,A)$.
This step is supported by cutting planes, which help increase the lower bound and thus reducing the integrality gap.

Note that Algorithm~\ref{alg:ip} is not guaranteed to terminate, because an optimal fractional and an optimal binary solution may require different guard locations~\cite{kbfs-esbgagp-12}.
This effect is weakened, but not completely suppressed by the use of cutting planes.
The impact is that there is an integrality gap between the upper and the lower bounds, which can be large.

\begin{algorithm}
	\begin{algorithmic}[1]
	\Input Polygon $P$
	\State $G \gets W \gets \textnormal{all vertices of $P$}$
	\State $A \gets \emptyset$
	\State $(\textnormal{lowerBound}, \textnormal{upperBound})
			\gets (1, \infty)$
	\Statex
	\Repeat
		\Repeat
			\State $x^* \gets$ optimize $\agp(G,W,A)$
			\State $W \gets W \cup \textnormal{run primal separation}$
			\If{separation failed}
				\State $\textnormal{upperBound} \gets \min(\textnormal{upperBound},
						\textnormal{objective value of $x^*$})$
			\EndIf
		\Until{separation failed \textbf{or} $\textnormal{lowerBound} = \textnormal{upperBound}$}

		\Repeat
			\State $(x^*, y^*) \gets$ optimize $\agr(G,W,A)$
			\State $G \gets G \cup \textnormal{run dual separation}$
			\State $A \gets A \cup \textnormal{separate cuts}$\label{alg:ip-cuts}
			\If{separation failed}
				\State $\textnormal{lowerBound} \gets \max(\textnormal{lowerBound},
						\lceil\textnormal{objective value of $y^*$}\rceil)$
			\EndIf
		\Until{separation failed \textbf{or} $\textnormal{lowerBound} = \textnormal{upperBound}$}
	\Until{$\textnormal{lowerBound} = \textnormal{upperBound}$ \textbf{or} time limit reached}
	\end{algorithmic}
	\caption[\acs{IP} Mode Algorithm]{%
	The \acs{IP} mode algorithm has one difference to
	Algorithm~\ref{alg:lp}:
	It solves \acp{IP} in the primal separation phase, thus only producing
	binary upper bounds.
	}
	\label{alg:ip}
\end{algorithm}

\section{Set Cover Facets}\label{sec:sc}

For finite sets of guards and witnesses $G, W \subset P$, $\agp(G,W)$ is an SC
polytope.
This motivates the investigation of SC-based facets.
In this section, we discuss a family of facets inspired by
Balas et~al.~\cite{bn-scp-89} and show that their separation,
while \NP-complete in the SC setting, can, under reasonable assumptions, be
solved in polynomial time when
exploiting
the underlying geometry of the AGP.
Additionally, we present a complete list of all AGP facets only using
coefficients in $\{0,1,2\}$.

\subsection{A Family of Facets}
\label{sec:sc-facets}

Let $P$ be a polygon and $G,W \subset P$ finite sets of guard and witness
positions.
Consider a finite non-empty subset
$\emptyset \subset S \subseteq W$ of witness positions;
the overlay of visibility regions of $S$
is called $\alpha_S$.
It implies the partition $P = J_0 \dcup J_1 \dcup J_2$, see~Figure~\ref{fig:j012}.
This is the geometry that is analogous to what Balas and Ng~\cite{bn-scp-89} did for the SC polytope.

\begin{figure}
	\centering
	\def\svgwidth{.6\linewidth}
	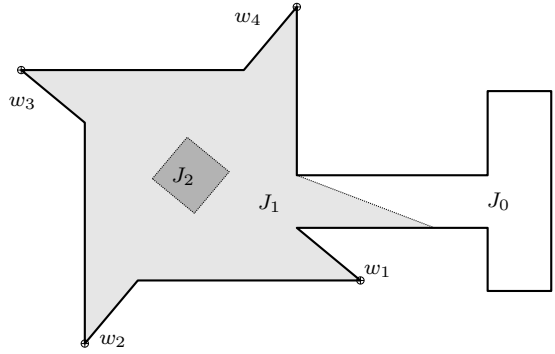
	\caption[Geometric Interpretation of $J_0$, $J_1$ and $J_2$]{%
	Polygon and witness selection $S = \{w_1, w_2, w_3, w_4\}$.
	Guards located in $J_2$ can cover all of $S$, and those in $J_1$
	some part of it, while those in $J_0$ cover none of
	$S$.
	}
	\label{fig:j012}
%\vspace*{-5mm}
\end{figure}

\begin{enumerate}
\item
	$J_2:= \{g \in P\mid S \subseteq \V(g)\}$, the set of points in $P$ covering
	all of $S$.
\item
	$J_0:= \{g \in P\mid \V(g) \cap S = \emptyset\}$, the set of positions in $P$
	that see none of $S$.
\item
	$J_1:= P \setminus (J_2 \cup J_0)$ %$\{g \in P\mid \emptyset \subset \V(g) \cap S \subset S\}$,
	the set of positions in $P$ that cover a non-trivial subset of $S$.
\end{enumerate}

Every feasible solution of the \ac{AGP} has to cover $S$.
Thus, it takes one guard in $J_2$, or at least two guards in $J_1$ to cover $S$.
For any $G$, this induces the following constraint \eqref{eq:j012-constraint};
for the sake of simplicity, we will also refer to this by $\alpha_S$.

\begin{equation}\label{eq:j012-constraint}
	\sum_{g \in J_2 \cap G} 2 x_g + \sum_{g \in J_1 \cap G} x_g \geq 2
\end{equation}

In the context of our iterative algorithm, it is important to represent
$\alpha_S$ independently from $G$.
This is achieved by storing the visibility overlay of the witnesses in $S$,
which implicitly makes the regions $J_0$, $J_1$ and $J_2$ available.
Any guard $g\in J_i$ in current or future iterations simply gets the coefficient
$i$.

Sufficient coverage of $S$ is necessary for sufficient coverage of
$P$, so \eqref{eq:j012-constraint} is valid for any $x\in \{0,1\}^G$ that is
feasible for $\agp(G,P)$, thus fulfilling our requirement of remaining feasible
in future iterations.
However, covering $S$ may require more than two guards in $J_1$,
so \eqref{eq:j012-constraint} does not always provide a
supporting hyperplane of $\conv(\agp(G,W))$.

When choosing a single witness $S = \{w\}$, we obtain $J_2 = \V(w)$,
$J_1 = \emptyset$ and $J_0 = P \setminus \V(w)$.
The resulting constraint is Inequality~\eqref{eq.ipf.cover}, the
witness-induced constraint of $w$,
multiplied by two.
For a choice of $S$ with two witnesses, $S = \{w_1, w_2\}$,
constraint~\eqref{eq:j012-constraint} yields the sum of the witness-induced
constraints of $w_1$ and $w_2$.
%Both examples are linear combinations of constraints of the \ac{LP}
%$\agr(G,W)$.
%It is easy to see that for $|S| \leq 2$, \eqref{eq:j012-constraint} only yields
%constraints that are fulfilled by all feasible solutions of
%$\agr(G,W)$.
%In order to find constraints that cut off fractional solutions, \ie,
%constraints that are feasible for $\agp(G,W)$, but not for $\agr(G,W)$, we
Thus, we consider $|S| \geq 3$ in the remainder of this section.

In order to show when~\eqref{eq:j012-constraint} defines a facet of
$\conv(\agp(G,W))$, we first need to apply a result
of~\cite{bn-scp-89} to the AGP setting.

\begin{lemma}\label{lem:agp-fulldim}
Let $P$ be a polygon and $G,W \subset P$ finite sets of guard and witness
positions.
Then $\conv(\agp(G,W))$ is full-dimensional, if and only if
\begin{equation}
	\forall w \in W: \quad \left| \V(w) \cap G \right| \geq 2
\end{equation}
\end{lemma}

\begin{proof}
We start by proving necessity. If every witness is seen by at least two guards, the $|G|$
vectors $x^i = \one - e_i$ are linearly independent and feasible
solutions of $\agp(G,W)$, so $\conv(\agp(G,W))$ is full-dimensional.

Now we consider sufficiency. If $\V(w) \cap G = \emptyset$ for some $w \in W$,
there is no feasible solution at all; if $\V(w) \cap G = \{g\}$, there is
none with $x_g = 0$, so there cannot be more than $|G| - 1$ linearly
independent solutions, and $\conv(\agp(G,W))$ is not full-dimensional.
\end{proof}

We require some terminology adapted from~\cite{bn-scp-89}.
Two guards $g_1,g_2 \in J_1$ are a \emph{2-cover} of $\alpha_S$, if
$S \subseteq \V(g_1) \cup \V(g_2)$.
The \emph{2-cover graph} of $G$ and $\alpha_S$ is the graph with nodes in
$J_1 \cap G$ and an edge between $g_1$ and $g_2$ if and only if $g_1,g_2$ are a 2-cover
of $\alpha_S$.
In addition, we have $T(g)=\{w \in \V(g) \cap W \mid \V(w) \cap G \cap (J_0 \wo \{g\}) = \emptyset\}$.

\begin{theorem}\label{thm:agp012-facet}
Given a polygon $P$ and finite $G,W \subset P$, let $\conv(\agp(G,W))$ be
full-dimensional and let $\alpha_S$ be as defined in
\eqref{eq:j012-constraint}, such that $S$ is maximal, \ie, there is no
$w \in W \setminus S$ with $\V(w) \subseteq \V(S)$.
Then the constraint induced by $\alpha_S$ defines a facet of
$\conv(\agp(G,W))$, if and only if:
\begin{enumerate}
\item\label{itm:ag-odd-cycle}
        Every component of the 2-cover graph of $\alpha_S$ and $G$ has an odd
        cycle.
\item\label{itm:ag-tk}
        For every $g \in J_0 \cap G$ such that $T(g) \neq \emptyset$ there
        exists either
        \begin{enumerate}
        \item\label{itm:ag-tk-1}
                some $g' \in J_2 \cap G$ such that $T(g) \subseteq \V(g')$;
        \item\label{itm:ag-tk-2}
                some pair $g', g'' \in J_1 \cap G$ such that
                $T(g) \cup S \subseteq \V(g') \cup \V(g'')$.
        \end{enumerate}
\end{enumerate}
\end{theorem}

\begin{proof}
$G$ and $W$ are finite, so $\agp(G,W)$ is an instance of \ac{SC} with
universe $W$ and subsets $G$, while $\conv(\agp(G,W))$ describes a
full-dimensional \ac{SC} polytope.

Our claim follows from Theorem~2.6 of Balas et~al.~\cite{bn-scp-89}, because the
conditions as well as the notion of 2-cover graphs and $T$ are equivalent.
The only difference is that we need to intersect $J_i$ with $G$ in order to
obtain finite sets, as our $J_i$ is a region in $P$, while that of Balas et~al.
naturally is a finite set of variables.
%The inequality $\alpha^S x \geq 2$ of~\cite{bn-scp-89} is the
%same inequality as the one induced by $\alpha_S$.
%Because our set $J_i$ is not a set of variables, but a region of $P$
%(in order to cope with all possible choices of guard positions),
%we intersect it with $G$ to obtain our finite set of variables.
%
%The set $S = M(J_0)$ of~\cite{bn-scp-89} corresponds to our
%definition of $S$, combined with the requirement that no $w \in W \setminus S$
%has $\V(w) \subseteq \V(S)$.
%If such an additional witness $w$ existed, the coefficients of $\alpha_S$
%could only decrease, but never increase;
%this ensures that there is no $S' \supset S$ with $\alpha_{S'}(g) < \alpha_S(g)$
%for some $g \in G$, \ie, there is no dominating hyperplane.
%In addition, our notion of a 2-cover graph and of $T(g)$ is equivalent to that
%of Balas and Ng~\cite{bn-scp-89}.
\end{proof}

\subsection{Geometric Properties of \texorpdfstring{$\alpha_S$}{Alpha S}}
\label{sec:sc-star}

It is easy to construct SC instances for any choice of $|S|\geq 3$, such that the SC
version of $\alpha_S$ cuts off a fractional solution~\cite{bn-scp-89}.
Finding $\alpha_S$ in the SC setting is NP-complete, see below.
But in the following, we show that in an AGP setting, only
$\alpha_S$ with $|S| = 3$ actually plays a role in cutting off fractional
solutions under reasonable assumptions, allowing us to separate it in polynomial
time.
%a drastically reduced number
%of these fractional solutions can actually occur. As we discuss after the following
%lemma and definition, this is beneficial in the context of our
%iterative framework; it also matches and explains the practical experience presented
%in Section~\ref{sec:exp}.
%In order to do that, we need the following lemma.
%For the corresponding highlights from \cite{bn-scp-89},
%cf.~Appendix~\ref{app:balas}.
%The crucial overall idea is illustrated in Figure~\ref{fig:sc012}.

\begin{lemma}\label{lem:ag012-ck}
Let $P$ be a polygon, $G, W \subset P$ finite sets of guard and witness
positions and $\emptyset \subset S \subseteq W$.
If every guard in $J_1 \cap G$ belongs to some 2-cover of $\alpha_S$ and
$S$ is minimal for $G$, \ie, there is no
proper subset $T \subset S$ such that $\alpha_T$ and $\alpha_S$ induce the same
constraint for $G$, the matrix of
$\agp(G,S)$ contains a permutation of the full circulant of order
$k = |S|$, which is
\begin{equation}
        C_k^{k-1} = \left(\begin{matrix}
                0      & 1      & \cdots & 1 \\
                1      & 0      & \ddots & \vdots \\
                \vdots & \ddots & \ddots & 1 \\
                1      & \cdots & 1      & 0
        \end{matrix}\right) \in \{0,1\}^{k \times k}.
\end{equation}
\end{lemma}

\begin{proof}
As $G$ and $W$ are finite, $\conv(\agp(G,W))$ is an \ac{SC} polytope with
universe $W$ and subsets $G$.
Then, as above, the claim follows from Balas et~al. or, or more accurately,
from Theorem~3.1 of~\cite{bn-scp-89}:
Our definition of $S$ being minimal for $G$ complies with
the definition of a minimal $C$-equivalent subset in~\cite{bn-scp-89};
the notion of matrix $A_S^{J_1}$ corresponds to our
$\agp(J_1 \cap G,S)$, \ie, a submatrix of $\agp(G,S)$.
\end{proof}

Lemma~\ref{lem:ag012-ck} holds, because the 2-cover property holds if and only if
no guard's coefficient in $\alpha_S$ can be reduced without turning
Inequality~\eqref{eq:j012-constraint} invalid~\cite{bn-scp-89}.
As $S$ is minimal, removing $w$ from $S$ must increase coefficients, \ie,
reclassify a guard $g \in J_1 \cap G$ to $J_2$.
So $\V(g) \cap S = S \setminus \{w\}$.
Such a guard exists for every $w \in S$.

Lemma~\ref{lem:ag012-ck} also states that separating $\alpha_S$ is equivalent
to finding permutations of $C_k^{k-1}$ in the LP matrix of $\agr(G,W)$.
It is possible to reduce a simple graph's adjacency matrix to a polygon with
guards $G$ and witnesses $W$, such that $\agr(G,W)$ contains a permutation of
$C_k^{k-1}$ if and only if the graph contains a clique of size $k$ or higher:
Introduce a guard and a witness for each of the graph's vertices, place all of them into a
convex polygon, and add a hole between a guard and a witness if they represent
the same vertex or if the two vertices are not connected.
Hence, the separation problem is NP-complete.

In the following, we examine when the separation of $\alpha_S$ is useful for our
iterative algorithm.
As $\alpha_S$ is represented by one or several permutations of $C_k^{k-1}$, we
need to introduce the notion of a polygon corresponding to $C_k^{k-1}$.
This allows us to examine the underlying geometry of $\alpha_S$ in the AGP.

\begin{figure}
%        \subfigure[\acs{AG} model of $C_3^2$]{
        \subfigure{
                \def\svgwidth{.25\linewidth}
                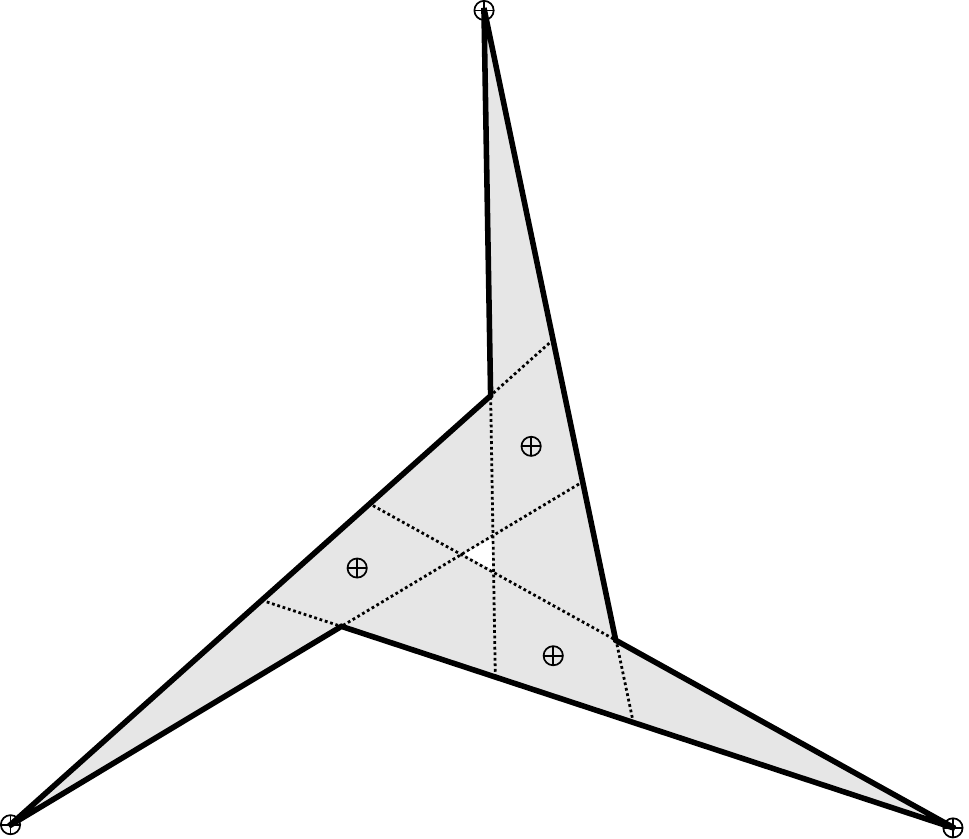
                \label{fig:sc012-3-godfried}
        }
%        \hfill\subfigure[\acs{AG} model of $C_4^3$]{
        \hfill\subfigure{
                \def\svgwidth{.25\linewidth}
                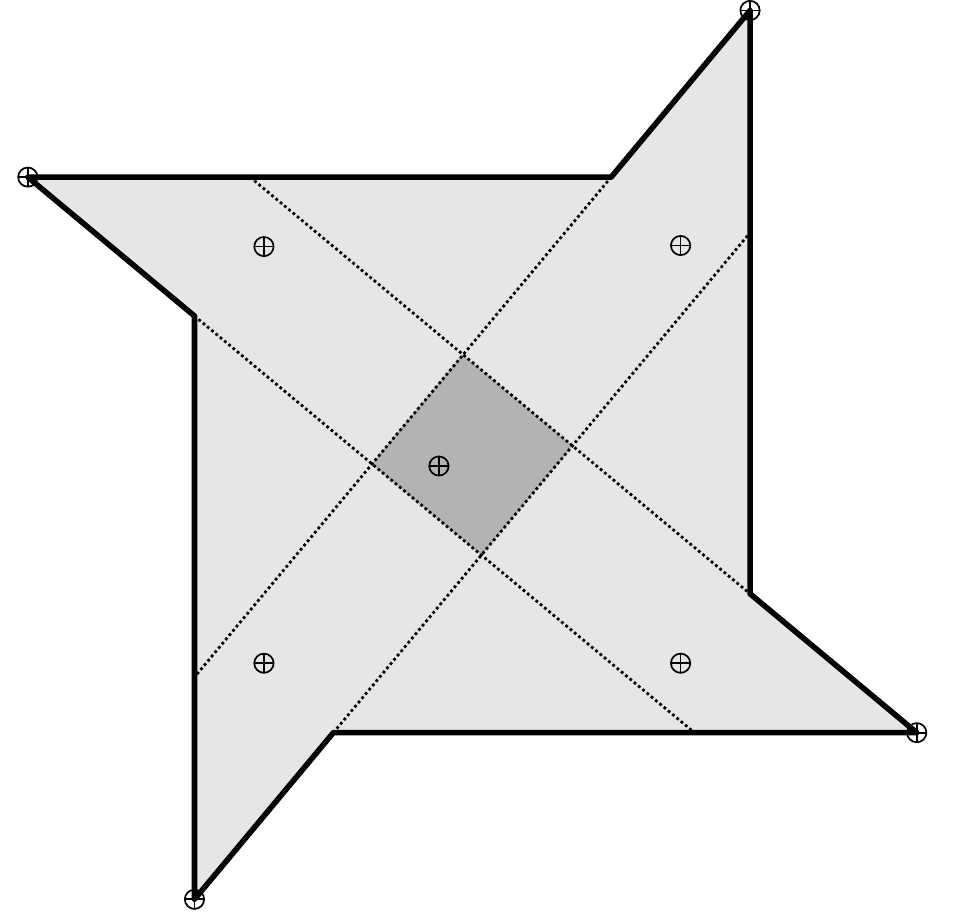
                \label{fig:sc012-4-godfried}
        }
%        \hfill\subfigure[Invalid AG model of $C_4^3$]{%
        \hfill\subfigure{%
		\def\svgwidth{.4\linewidth}
        	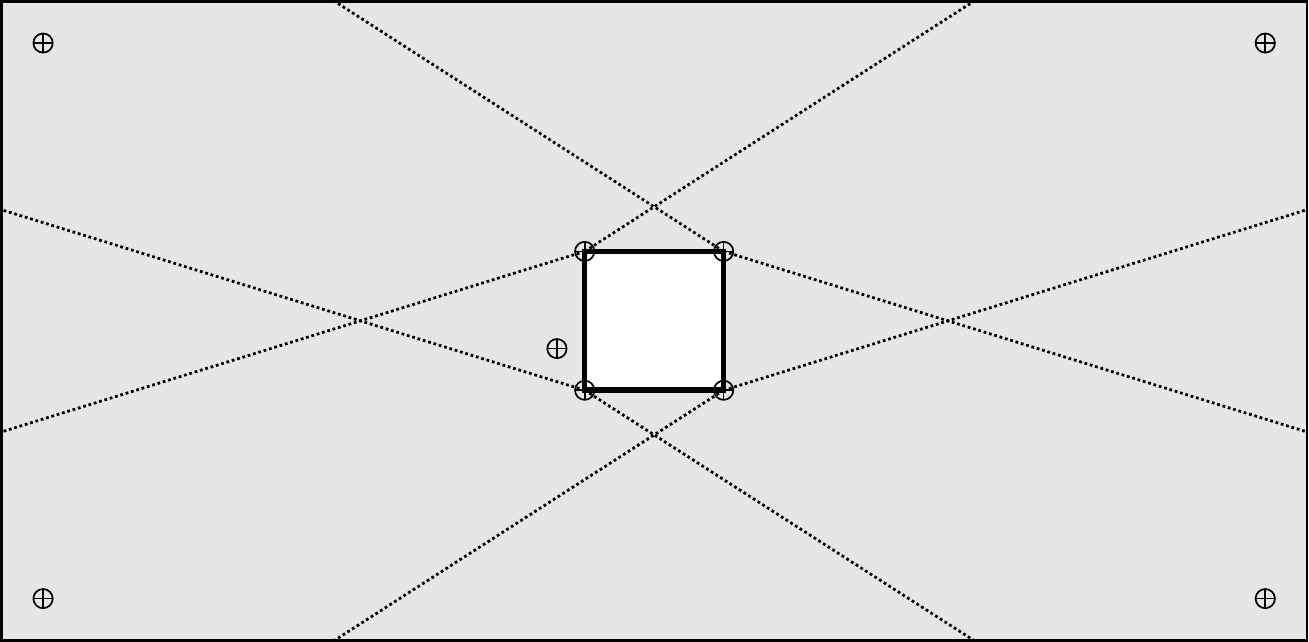
        	\label{fig:sc012-4-hole}
        }

        \caption[\acl{AG} Interpretation of $C_k^{k-1}$]{%
        $P_3^2$ (left) and two attempts for $P_4^3$ (middle and right).
        In the left case, Inequality~\eqref{eq:j012-constraint} enforces using two
        guards instead of three $\frac12$-guards.

        The first attempt for $P_4^3$ (middle) is star-shaped; here a cutting
        plane would cut off the intermediate fractional solution of four
        $\frac13$-guards, but as soon as $g^*$ is found, the fractional solution
        is replaced by a binary one with just one guard, with or without cutting
        plane.

        Finally, the second attempt for $P_4^3$ (right) is not star-shaped, but
        again, there is no need for a cutting plane to cut off the fractional
        solution of four $\frac13$-guards:
        $w^*$ is only covered by $\frac23$, so $w^*$ is separated by our
        algorithm and then enforces the use of at least two guards in the next
        iteration; again, with or without cutting plane.
%        {\bf \subref{fig:sc012-3-godfried}} A polygon
%        of type $P_3^2$.
%        $\alpha_{\{w_1,w_2,w_3\}}$ would cut off the corresponding fractional solution and enforce the
%        placement of two guards.
%        {\bf \subref{fig:sc012-4-godfried}}
%        $P_4^3$ is star-shaped.
%        As soon as $g^* \in J_2$ is separated, the optimum uses
%        one guard located at $g^*$.
%        This implies that $\alpha_{\{w_1,w_2,w_3,w_4\}}$ has no effect as soon
%        as $g^*$ is added to $\agr\left( G, W, \{ \alpha_W \} \right)$.
%        %$g^* \in J_2$ a coefficient of 2.
%        {\bf \subref{fig:sc012-4-hole}} Another attempted \ac{AG} model of $C_4^3$.
%        %As opposed to Figure~\ref{fig:sc012-4-godfried},
%        $J_2$ is empty, but
%        the optimal solution of $\agr(G,W)$ is infeasible for
%        $\agr\left( G, W \cup \{w^*\} \right)$ and thus for $\agr(G,P)$.
        }
        \label{fig:sc012}
%\vspace*{-5mm}
\end{figure}

\begin{definition}[Full Circulant Polygon]\label{def:fcp}
A polygon $P$ along with $G(P) = \{g_1, \dots, g_k\} \subset P$ and
$W(P) = \{w_1, \dots, w_k\} \subset P$ for $3 \leq k \in \N$ is called
\emph{Full Circulant Polygon}, or $P_k^{k-1}$, if
\begin{eqnarray}
        \forall\ 1 \leq i \leq k&:& \quad \V(g_i) \cap W(P) = W(P) \setminus \{w_i\}
\label{eq:fcp-i}\\
        \forall w \in P&:& \quad \left| \V(w) \cap G(P) \right| \geq k - 1
\label{eq:fcp-coverage}
\end{eqnarray}
We may refer to $G(P)$ and $W(P)$ by just $G$ and $W$, respectively.
%, whenever
%$P$ is clear from context.
\end{definition}

Note that in $P_k^{k-1}$ the full circulant $C_k^{k-1}$ completely describes the
visibility relations between $G$ and $W$.
This implies that the optimal solution of $\agr(G,W)$ is
$\frac{1}{k-1} \cdot \one$, with cost $\frac{k}{k-1}$. It
is feasible for $\agr(G,P_k^{k-1})$ by Property~\eqref{eq:fcp-coverage}, as
any point $w \in P_k^{k-1}$ is covered by at least
$(k-1) \cdot \frac{1}{k-1} = 1$.

Figure~\ref{fig:sc012} captures construction attempts for models of $C_k^{k-1}$.
$P_3^2$ exists; however, as we prove in Theorem~\ref{thm:fcp-star},
the polygons for $k \geq 4$ are either star-shaped
or not full circulant.
If they are star-shaped, the optimal solution is to place one guard within the kernel.
If they are not full circulant polygons, the optimal solution of $\agr(G,W)$ is
infeasible for $\agr(G,P)$ and the current fractional solution is intermittent, \ie,
cut off in the next iteration.
Both cases eliminate the need for a cutting plane,
and we may avoid the NP-complete separation problem
by restricting separation to $k=3$.

In the following we prove that $P_k^{k-1}$
is star-shaped for $k \geq 4$. %so there cannot be a $P_4^3$
%with an optimal solution greater than one, and thus no fractional solution
%that would require separation by some facet $\alpha_S$,
%compare Figure~\ref{fig:sc012}.
We start with Lemma~\ref{lem:2-composition}, which
%provides the first
%insight into the structure of $P_k^{k-1}$.
%It
shows that any pair of guards in $G$ is sufficient to cover $P_k^{k-1}$.

%\begin{lemma}[2-Composition]\label{lem:2-composition}
\begin{lemma}\label{lem:2-composition}
Let $P_k^{k-1}$ be a full circulant polygon.
Then $P_k^{k-1}$ is the union of the visibility polygons of any pair of
guards in $G\left( P_k^{k-1} \right) = \{g_1,\dots,g_k\}$:
\begin{equation}
	\forall 1 \leq i < j \leq k: \quad P_k^{k-1} = \V(g_i) \cup \V(g_j)
%\begin{aligned}
%        P_k^{k-1}
%                & = \V(g_1) \cup \V(g_2) = \V(g_1) \cup \V(g_3) = \dots = \V(g_1) \cup \V(g_k) \\
%                & = \V(g_2) \cup \V(g_3) = \dots = \V(g_2) \cup \V(g_k) \\
%                & \quad \vdots \\
%                & = \V(g_{k-1}) \cup \V(g_k)
%\end{aligned}
\end{equation}
\end{lemma}

\begin{proof}
Suppose $P_k^{k-1}$ is a full circulant polygon, but
$P_k^{k-1} \neq \V(g_i) \cup \V(g_j)$ for $1 \leq i < j \leq k$.
%We always have $\V(g_i), \V(g_j) \subseteq P_k^{k-1}$, so
%$P_k^{k-1} \supset \V(g_i) \cup \V(g_j)$ is the only way to fulfill the
%inequality.
Then there exists some $w \in P_k^{k-1}$ with $g_i \notin \V(w)$, as well
as $g_j \notin \V(w)$,
implying that $\left| \V(w) \cap G \right| \leq k-2$, a contradiction to
Property~\eqref{eq:fcp-coverage} of Definition~\ref{def:fcp}.
\end{proof}

The next step is Lemma~\ref{lem:no-holes}, which drastically restricts
the possible structure of $P_k^{k-1}$.

\begin{lemma}\label{lem:no-holes}
Let $P_k^{k-1}$ be a full circulant polygon with
$G\left( P_k^{k-1} \right) = \{g_1,\dots,g_k\}$.
Suppose $k \geq 4$.
Then $P_k^{k-1}$ has no holes.
\end{lemma}

\begin{figure}
        \def\svgwidth{.65\linewidth}
        \centering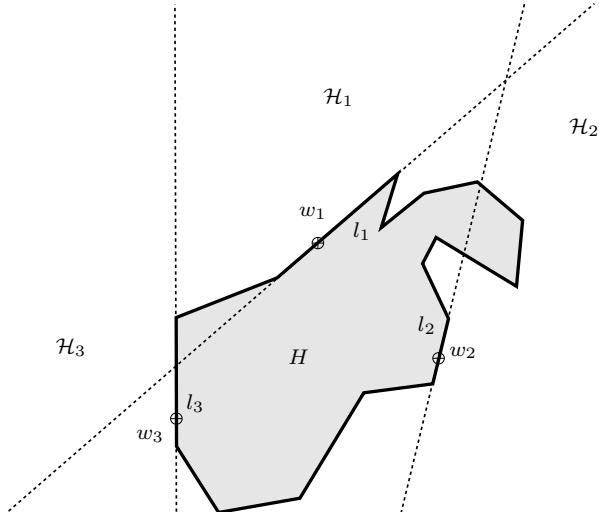
        \caption[Sketch for a Hole in $P_k^{k-1}$]{%
        A hole $H$ in $P_k^{k-1}$ with
        $\mathcal{H}_1 \cap \mathcal{H}_2 \cap \mathcal{H}_3 = \emptyset$.
        There are $k-1$ guards in $\mathcal{H}_1$ and in $\mathcal{H}_2$, so
        there must be $k-2$ in their intersection.
        This only leaves $2$ guards for $\mathcal{H}_3$.
        }
        \label{fig:hole-proof}
\end{figure}

\begin{proof}
Refer to Figure~\ref{fig:hole-proof}.
Suppose $P_k^{k-1}$ has a hole $H$.
Each edge $l_i$ of $H$ induces a half-space $\mathcal{H}_i$.
There are three such edges $l_1,l_2,l_3$, such that
$\mathcal{H}_1 \cap \mathcal{H}_2 \cap \mathcal{H}_3 = \emptyset$, for otherwise
the outside of $H$ would be convex by Helly's Theorem.
Let $w_i$ denote a point in the interior of $l_i$.

In order for $w_1$ to fulfill~\eqref{eq:fcp-coverage}, at least
$k-1$ of the guards in $G$ must be located in
$\V(w_1) \subseteq \mathcal{H}_1$.
Analogously, there must be $k-1$ guards in $\mathcal{H}_2$.
Covering $w_1$ and $w_2$ with a total of $k$ guards is only possible if at least
$k-2$ guards of $G$ are located in the intersection of the two half-spaces:
$\left| \mathcal{H}_1 \cap \mathcal{H}_2 \cap G \right| \geq k-2$.
If there are only $k' < k-2$ guards in $\mathcal{H}_1 \cap \mathcal{H}_2$,
it takes $(k-1) - k'$ additional guards in
$\mathcal{H}_1 \setminus \mathcal{H}_2$ to cover $w_1$ and in
$\mathcal{H}_2 \setminus \mathcal{H}_1$ to cover $w_2$, resulting in a total
of $k' + 2(k-1-k') = 2k - (k'+2) > k$ guards, a contradiction.

As $\mathcal{H}_1 \cap \mathcal{H}_2 \cap \mathcal{H}_3 = \emptyset$, there can
be at most $2$ guards in $\V(w_3) \subseteq \mathcal{H}_3$, which violates
Property~\eqref{eq:fcp-coverage} for $k \geq 4$, a contradiction.
\end{proof}

As shown in Figure~\ref{fig:sc012-3-hole}, $k \geq 4$ is tight:
a triangle with a concentric triangular hole is an example of $P_3^2$,
with guards in the outside corners and witnesses on the inside edges.

\begin{figure}
        \def\svgwidth{.4\linewidth}
        \centering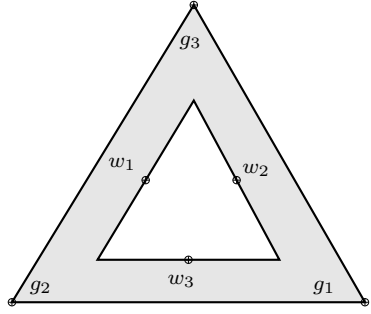
        \caption[Alternative \acl{AG} Interpretation of $C_3^2$]{%
        $P_3^2$, a possible \ac{AG} interpretation of $C_3^2$ with a hole.
        It proves that the bound of $k \geq 4$ in Lemma~\ref{lem:no-holes} is
        tight.
        }
        \label{fig:sc012-3-hole}
\end{figure}

We require one final technical lemma before proceeding to the main theorem,
Theorem~\ref{thm:fcp-star}.

\begin{lemma}\label{lem:helly-poly}
Consider two disjoint non-empty convex polygons, described
as the intersection of half-spaces:
$P_1 = \bigcap_{i=1,\dots,n} \mathcal{H}_i$ and
$P_2 = \bigcap_{i=n+1,\dots,n+m} \mathcal{H}_i$.
Then some $\mathcal{H}_i$, $1 \leq i \leq n+m$ separates $P_1$ and
$P_2$.
\end{lemma}

\begin{proof}
%If $n+m \leq 1$, $P_1 \cap P_2 = \emptyset$ is impossible.
%If there only are two different half-spaces, they are parallel, which proves
%the claim.
$n+m \leq 2$ is trivial, so consider $n+m \geq 3$.
%But if there are at least three different half-spaces, recall that Helly's
%Theorem states for the two-dimensional case that for any family
%$\mathcal F = \{M_1,\dots,M_n\}$ of $n \geq 3$ convex sets, the following
%holds:
%\begin{equation}
%        \bigforall 1 \leq i < j < k \leq n:
%                \quad M_i \cap M_j \cap M_k \neq \emptyset
%                \quad \Rightarrow \quad
%                \bigcap_{i=1,\dots,n} M_i \neq \emptyset
%\end{equation}
%Inverting it provides
Because of $P_1 \cap P_2 = \bigcap_{i=1,\dots,n+m} \mathcal{H}_i = \emptyset$,
Helly's Theorem applied to the two-dimensional convex half-spaces $\mathcal{H}_i$
implies
the existence of three half-spaces
$\mathcal{H}_i$, $\mathcal{H}_j$, and $\mathcal{H}_k$, $i < j < k$,
with $\mathcal{H}_i \cap \mathcal{H}_j \cap \mathcal{H}_k = \emptyset$.
%\begin{equation}
%        \bigcap_{i=1,\dots,n} M_i = \emptyset
%        \quad \Rightarrow \quad
%        \bigexists 1 \leq i < j < k \leq n: \quad M_i \cap M_j \cap M_k = \emptyset
%\end{equation}
%We choose $\mathcal{F} = \{\mathcal{H}_1,\dots,\mathcal{H}_{n+m}\}$ and
%obtain
%\begin{equation}
%        \bigcap_{i = 1,\dots,n+m} \mathcal{H}_i
%                = \underbrace{\mathcal{H}_1 \cap \dots \cap \mathcal{H}_n}_{P_1} \cap
%                        \underbrace{\mathcal{H}_{n+1} \cap \dots \cap \mathcal{H}_{n+m}}_{P_2}
%                = \emptyset
%\end{equation}
%So Helly's Theorem asserts there is a collection of three half-spaces
%$\mathcal{H}_i$, $\mathcal{H}_j$ and $\mathcal{H}_k$ with an empty
%intersection.
%They cannot all belong to the same polyhedron, because then their
%intersection would not be empty.
\OBdA, we assume $i=1$, $j=2$ and $k=n+1$,
%, so two half-spaces from $P_1$ and
%one from $P_2$.
%This combination provides
which provides
\begin{equation}
        \underbrace{\mathcal{H}_1 \cap \mathcal{H}_2}_{\supseteq P_1}
                \cap \underbrace{\mathcal{H}_{n+1}}_{\supseteq P_2} = \emptyset,
\end{equation}
so it follows that
$P_1 \cap \mathcal{H}_{n+1} = \emptyset$
with $P_2 \subseteq \mathcal{H}_{n+1}$ by construction,
and $\mathcal{H}_{n+1}$ is the half-space whose existence is the claim.
\end{proof}

Now all preliminaries for the main theorem of the section are met and we can
proceed to show Theorem~\ref{thm:fcp-star}, which claims that full circulant
polygons are star-shaped for $k \geq 4$.

\begin{figure}
	\def\svgwidth{\linewidth}
	\centering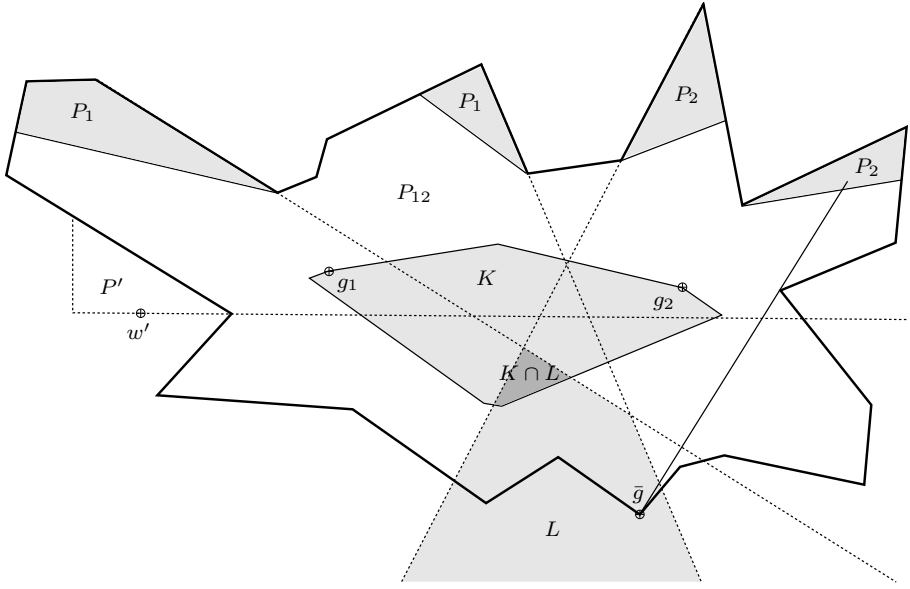
	\caption[Sketch of $P_k^{k-1}$]{%
	$P_k^{k-1}$ with guards $g_1$ and $g_2$.
	$P_1$ and $P_2$ are the gray areas at the top.
	$P_1$ is seen by $g_1$ but not by $g_{2}$; an analogous property holds for $P_2$.
	The rest of $P_k^{k-1}$ is $P_{12}$, a star-shaped polygon entirely seen by both
	$g_1$ and $g_2$, $K$ is its kernel.
	$L$ is the area whose view into $P_1 \cup P_2$ is not blocked by any
	edge of $P_k^{k-1}$ that coincides with $P_1 \cup P_2$.
	It contains all of $g_3,\dots,g_k$.

	If $P'$ would be added to $P_k^{k-1}$, $K$ would be cut off below the
	dashed line containing $w'$ and $K \cap L = \emptyset$.
	But then no point in $L$, including $g_3,\dots,g_k$, could see $w'$, a
	contradiction to the property of $P_k^{k-1}$ requiring $k-1$ guards to
	see any of its points.
	}
	\label{fig:fcp-star}
\end{figure}

\begin{theorem}\label{thm:fcp-star}
A full circulant polygon $P_k^{k-1}$ with $k \geq 4$ is star-shaped.
\end{theorem}

\begin{proof}
Refer to Figure~\ref{fig:fcp-star}.
Let $P_k^{k-1}$ with $k \geq 4$ be a full circulant polygon.
The guards in $G = G\left( P_k^{k-1} \right)$ must be covered by a total of
$k-1$ guards, \ie, they must also fulfill \eqref{eq:fcp-coverage},
so each guard can see at least $k-2$ others.
\OBdA, let $g_1$ and $g_2$ denote two guards in each other's field of view.

Now consider $P_{12} = \V(g_1) \cap \V(g_2) \subseteq P_k^{k-1}$, the subset
of $P_k^{k-1}$ seen by both $g_1$ and $g_2$.
It is star-shaped, because
$g_1$ and $g_2$ are in its kernel, which we denote by $K$.
The rest of $P_k^{k-1}$, \ie, $P_k^{k-1} \setminus P_{12}$, consists of two
types of areas:
\begin{enumerate}
\item
        $P_1 = P_k^{k-1} \setminus \left( P_{12} \cup \V(g_2) \right)$, points
        visible from $g_1$, but not from $g_2$.
\item
        $P_2 = P_k^{k-1} \setminus \left( P_{12} \cup \V(g_1) \right)$, points
        visible from $g_2$, but not from $g_1$.
\end{enumerate}
Together, $g_1$ and $g_2$ cover every $w \in P_k^{k-1}$, because
$P_k^{k-1} = \V(g_1) \cup \V(g_2)$ due to %the 2-composition
Lemma~\ref{lem:2-composition}. Thus, $P_k^{k-1} = P_{12} \dcup P_1 \dcup P_2$.

We now examine which guards can see what part of $P_{12}$, $P_1$ and
$P_2$.
For that, we classify three types of edges.
\emph{Gray} edges are those edges of $P_k^{k-1}$ that coincide with $P_1$ or
$P_2$, \emph{white} edges denote the other edges of $P_k^{k-1}$.
Finally, edges of $P_{12}$ not part of $P_k^{k-1}$ that separate $P_1 \cup P_2$
from $P_{12}$ are referred to as \emph{white-gray} edges.
Note that white-gray edges do not block the view of any guard in
$P_k^{k-1}$, because they are merely edges of the auxiliary polygon
$P_{12}$.
$K$ is the intersection of all half-spaces induced by white or
white-gray edges, because a star's kernel is the intersection of all
half-spaces induced by its edges.%~\cite{lp-oafkp-79}.

All points able to cover all of $P_1 \cup P_2$ must be contained in
\begin{equation}
\begin{aligned}
        L & = \left\{ g \in \R^2 \mid
                \textnormal{no gray edge blocks $g$'s view into $P_1$ or $P_2$} \right\} \\
          & = \bigcap_{\textnormal{$e$ is gray edge}}
                \textnormal{Half-space induced by $e$.}
\end{aligned}
\end{equation}
Some points in $L$ may be located outside of $P_k^{k-1}$, and even those
points of $L$ inside of $P_k^{k-1}$ may not be able to see all of
$P_1 \cup P_2$, due to white edges blocking their view.
$\bar g$ in Figure~\ref{fig:fcp-star} is an example for this case: it
cannot see the rightmost part of $P_2$.
However, $L \neq \emptyset$, because $g_3,\dots,g_k \in L \cap P_k^{k-1}$.
Every $g_i \in G$ with $3 \leq i \leq k$ is able to see all of
$P_1 \cup P_2$: $g_i$ can see all of $P_1$, because $g_2$, $g_i$ are a
2-cover of $P_k^{k-1}$; $g_i$ can entirely see $P_2$, because $g_1$, $g_i$
are also a 2-cover by Lemma~\ref{lem:2-composition}.

The remaining part of the proof involves two steps.
First, we argue that any point in $K \cap L$ can see all of $P_k^{k-1}$;
then we show that $K \cap L \neq \emptyset$, \ie, that $P_k^{k-1}$ is
star-shaped.
For the first step, assume there exists some $g \in K \cap L$.
Now $g$ has the following properties:
\begin{enumerate}
\item
        $g \in P_k^{k-1}$, because $K \subseteq P_k^{k-1}$.

\item
        $g$ can see all of $P_{12}$ by definition of $K$, which includes the
        interior of $P_{12}$, all white and all white-gray edges.

\item
        Because $P_k^{k-1}$ has no holes by Lemma~\ref{lem:no-holes}, $g$'s view
        on the white-gray edges is not blocked; and due to $g \in L$, there is
        nothing left that can block $g$'s view into $P_1 \cup P_2$.
\end{enumerate}
So $g \in \kernel\left(P_k^{k-1}\right)$, provided that $K \cap L \neq \emptyset$;
we show the latter as follows.

$K$ and $L$ are two-dimensional polyhedra, each of their edges is a facet.
Suppose their intersection is empty;
then by Lemma~\ref{lem:helly-poly},
there must be a facet of one of them that separates them from each other.
We consider three cases, because $K$ has two
types of edges and $L$ has one:
\begin{enumerate}
\item
        The facet is a facet of $K$, induced by a white edge $e$.
        Now consider a point $w$ in the interior of $e$.
        $w$ is seen by $g_1$ and $g_2$, but not by any of the points
        $g_3,\dots,g_k$, because $g_3,\dots,g_k \in L$, and $e$ induces a facet
        separating $K$ from $L$.
        Due to $k \geq 4$, this makes $\left| \V(w) \cap G \right| \geq k-1$,
        \ie, \eqref{eq:fcp-coverage}, impossible and thus
        contradicts the requirement of $P_k^{k-1}$ being a full circulant
        polygon.

        This would be the case in Figure~\ref{fig:fcp-star}, if $P_k^{k-1}$ had
        the extension $P'$,
        which would cut off the lower part of $K$ and thus separate $K$ from $L$.
        However, then only $g_1$ and $g_2$ but none of $g_3,\dots,g_k$ could see $w'$,
        which violates \eqref{eq:fcp-coverage}.

\item
        The facet is a facet of $K$, induced by a white-gray edge $e$.
        As $e$ is a white-gray edge, there is a part of $P_1$ or $P_2$
        adjacent to $e$.
        This part cannot be seen from any point in $L$, because the facet
        induced by $e$ is a facet of $K$ and separates $K$ from $L$ by
        assumption, and because $P_k^{k-1}$ has no holes by
        Lemma~\ref{lem:no-holes}.
        However, $g_3,\dots,g_k \in L$ not being able to see $P_1$ or $P_2$
        contradicts \eqref{eq:fcp-coverage}.

\item
        The facet is a facet of $L$.
        This means that there is some gray edge $e$ corresponding to that facet.
        A point $w$ in $e$'s interior is not seen by $g_1$ or $g_2$, because the
        facet separates $K$ from $L$.
        So $s$ is not seen by more than $k-2$ guards, and thus violates
        \eqref{eq:fcp-coverage}.
\end{enumerate}
All cases lead to contradiction, and thus $K \cap L \neq \emptyset$.
Therefore, $P_k^{k-1}$ has a non-empty kernel, and is star-shaped
for $k \geq 4$, as claimed.
\end{proof}

\begin{figure}
	\centering
	\def\svgwidth{.8\linewidth}
	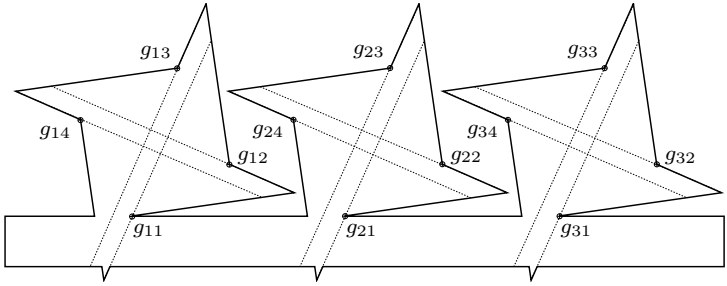
	\caption[Embedding $P_4^3$ into a larger polygon]{%
	Three instances of $P_4^3$ embedded into a larger polygon.
	Setting all guards to $\frac{1}{3}$ is
	feasible and optimal, even though no guard is placed in any of the $P_4^3$ kernels.
	}
	\label{fig:sc012-4-embedded}
%\vspace*{-5mm}
\end{figure}

Theorem~\ref{thm:fcp-star} does not rule out situations in which
$P_k^{k-1}$, for $k \geq 4$ is part of a larger polygon, as shown in
Figure~\ref{fig:sc012-4-embedded}.
%$\alpha_S$ still cuts off a feasible
%fractional solution, even though $P_4^3$ is star-shaped.
This example has no integrality gap;
placing at least five copies of $P_4^3$ around an appropriate
central subpolygon with a hole can actually create one.
However, such cases are much harder to come by, making the $k \geq 4$ facets a lot less
useful for cutting off fractional solutions.

Theorem~\ref{thm:fcp-star} does, however, provide a very useful separation
heuristic.
As the separation problem is \NP-complete for unlimited $k$, but solvable in
polynomial time for a fixed $k$, it is clear that $k$ must be limited in a
practical algorithm.
Theorem~\ref{thm:fcp-star} justifies choosing $k=3$ from a theoretical point
of view, by stating that the underlying geometry for $k>3$ is star-shaped, \ie,
allows placing one non-fractional guard in its kernel, see Figure~\ref{fig:sc012}.
As we show in Section~\ref{sec:exp}, this can also be validated in an experimental setting.

\subsection{All Art Gallery Facets with Coefficients \textbraceleft 0, 1, 2\textbraceright}

For finite $G,W\subset P$, $\agp(G,W)$ is also an SC instance.
Balas and Ng identified all SC facets with coefficients in
$\{0,1,2\}$~\cite{bn-scp-89};
so we present all AGP facets with coefficients $\{0,1,2\}$.
This includes three trivial facet classes, %see Appendix~\ref{app:balas},
%\eqref{eq:trivial-sc-1}~--~\eqref{eq:trivial-sc-3}.
\eqref{eq:trivial-ag-1}~--~\eqref{eq:trivial-ag-3},
which are unable to cut off fractional solutions of $\agr(G,W)$.
The
only non-trivial facet in this inventory is the one of type $\alpha_S$ described
above.

%For an arbitrary polygon $P$ and finite sets of guard and
%witness positions $G,W \subset P$, $\agp(G,W)$ describes an instance of
%SC,
%so any facet of $\conv(\agp(G,W))$ but be an SC facet.
%by the characterization. Thus,
%facets with coefficients in $\{0,1,2\}$ must be among those that have been
%characterized by Balas and Ng~\cite{bn-scp-89}.

\begin{equation}\label{eq:trivial-ag-1}
	x_g \geq 0
\end{equation}
is a facet of a full-dimensional $\conv(\agp(G,W))$, and only if
$\left| \V(w) \cap G \setminus \{g\} \right| \geq 2$ for all $w \in W$,
\ie, if every witness sees at least two guards other than
$g$.

A second type of \ac{AGP} facet is the upper bound of one for every guard
value.
It is a facet of every full-dimensional
$\conv(\agp(G,W))$~\cite{bn-scp-89}:
\begin{equation}\label{eq:trivial-ag-2}
	x_g \leq 1
\end{equation}

The third and last trivial \ac{AGP} facet with coefficients in $\{0,1,2\}$ is
\begin{equation}\label{eq:trivial-ag-3}
	\sum_{g \in \V(w) \cap G} x_g \geq 1
\end{equation}
This simply is the constraint induced by the witness $w \in W$, which
enforces sufficient coverage of $w$.
It is facet defining and only if two conditions hold:
First, there must not be any witness $w' \in W$ with
$\V(w') \cap G \subset \V(w) \cap G$.
Otherwise, the coverage of $w$ would be implied by that of $w'$.
Second, for any guard $g \in G \setminus \V(w)$, there exists some other
guard $g' \in \V(w) \cap G$ that can see all of
$\left\{ w' \in \V(g) \cap W \mid \bar{g} \notin \V(w'), \forall \bar{g}
\in G \setminus (\V(w) \cup \{g\}) \right\}$, compare~\cite{bn-scp-89}.

The fourth, and the only non-trivial, \ac{AGP} facet with coefficients in
$\{0,1,2\}$ is the facet of type $\alpha_S$ presented in Inequality~\eqref{eq:j012-constraint} and Theorem~\ref{thm:agp012-facet}, which is
thoroughly analyzed above.

\section{Edge Cover Facets}\label{sec:ec}

Solving $\agr(G,W)$ for finite $G,W \subset P$, such that no guard can
see more than two witnesses is equivalent to solving fractional edge cover (EC)
on the graph with nodes $W$, an edge between $v \neq w \in W$ for
each $g \in G$ with $\V(g) \cap W = \{v,w\}$, and a loop for each
$g \in G$ with $\V(g) \cap W = \{w\}$.
The fractional EC polytope is known to be half-integral~\cite{s-cope-03},
which can be exploited to show that fractional solutions always form odd-length
cycles of $\frac12$-guards.
%a useful fact when separating constraints of type~\eqref{eq:2guard-agr-sep}.

In the conclusions of~\cite{kbfs-esbgagp-12}, we proposed a class of valid inequalities
motivated by this.
The idea is to identify $k$ witnesses $\overline{W} = \{ w_1, \dots, w_k \}$,
such that no point exists that can see more than two of them.
Then at least $\left\lceil \frac{k}{2} \right\rceil$ binary guards
are needed for covering $\overline W$.
Two examples are shown in Figure~\ref{fig:2guard-agr}.
%; both are based on odd $k$-cycles.
%
%A fractional optimal solution has all guard values on the cycle at $\frac{1}{2}$.
%For an odd $k$,
\begin{equation}\label{eq:2guard-agr-sep}
	\sum_{g \in \V(\overline{W}) \cap G} x_g
		\geq \left\lceil \frac{k}{2} \right\rceil
%		= \frac{k+1}{2}
\end{equation}
%separates these fractional solutions from feasible, integral solutions.

\begin{figure}
%	\subfigure[A $5$-gonal hole with narrow corridor]{
		\def\svgwidth{.4\linewidth}
		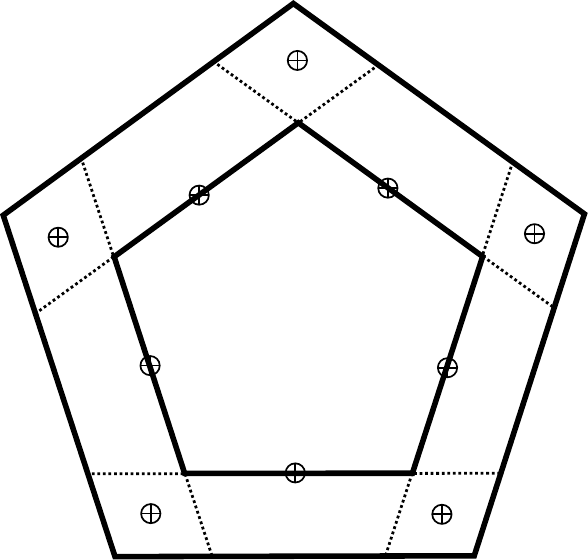
%		\label{fig:2guard-agr-5}
%	}
	\hfill%\subfigure[Godfried's Favorite Polygon]{
		\def\svgwidth{.4\linewidth}
		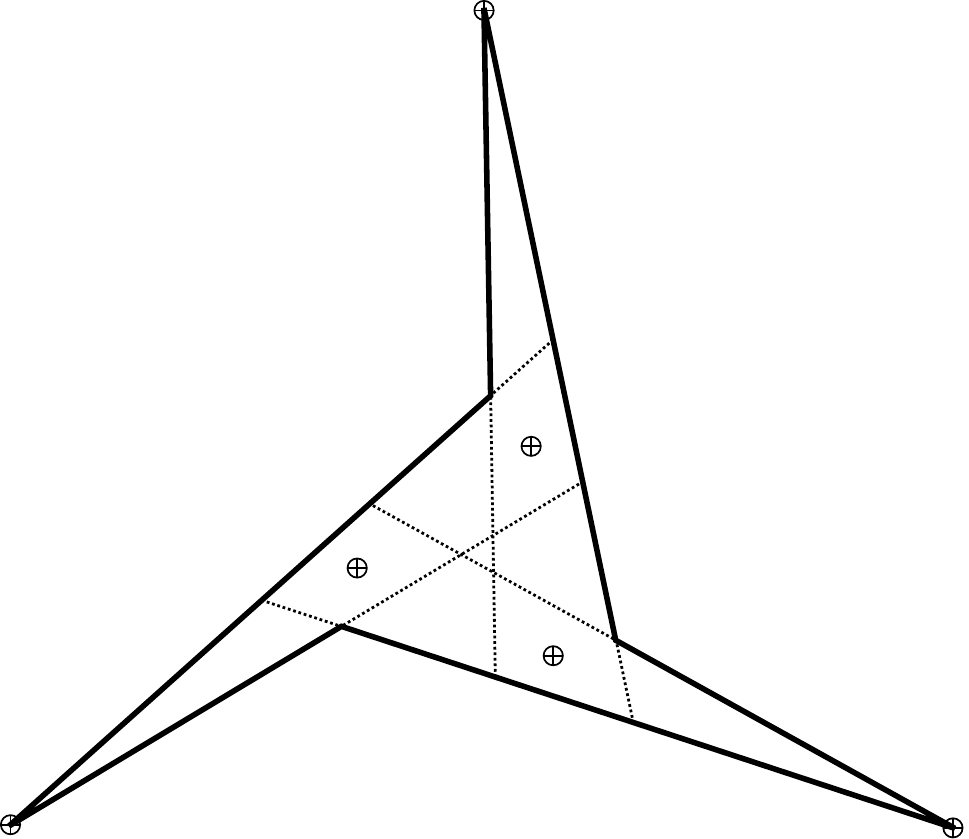
%		\label{fig:2guard-agr-godfried}
%	}
	\caption[\acl{AGR2}]{%
	Two situations, in which no guard exists that can see more than
	two witnesses.
	On the left,
	assigning $g_1=\dots=g_5=\frac{1}{2}$ results in an optimal
	fractional solution of $\frac{5}{2}$, compare~\cite{kbfs-esbgagp-12}.
	Applying \eqref{eq:2guard-agr-sep} yields
	$g_1+\dots+g_5 \geq \frac{5+1}{2} = 3$ and cuts off this
	fractional solution.

	On the right, the optimal fractional
	solution is $g_1=g_2=g_3=\frac{1}{2}$.
	\eqref{eq:2guard-agr-sep} provides the
	constraint $g_1+g_2+g_3 \geq \frac{3+1}{2} = 2$, which cuts off that
	fractional solution as well.
	}
	\label{fig:2guard-agr}
\end{figure}

%\todo{Idea: Just one of the pictures as wrapfig and text around it.
%Also: remove redundancies of caption and text.}
%Figure~\ref{fig:2guard-agr} demonstrates two situations in which
%such a cutting plane separates feasible fractional from feasible
%integral solutions.
%One is a convex $k$-gonal hole surrounded by a narrow corridor that
%rules out the existence of guards capable of covering more than two
%witnesses.
%The other is Godfried's favorite polygon, compare page~4
%of~\cite{r-agta-87}.

Obviously, for any choice of $P \supseteq G' \supseteq G$,
\eqref{eq:2guard-agr-sep} does not cut off any feasible
solution $x \in \{0,1\}^{G'}$ of $\agp(G',P)$, as long as no point in $P$ exists that
sees
more than two of these witnesses.
Hence, analogously to the SC cuts, a cut can be represented as visibility overlay
$\alpha_{\overline{W}}$ and kept in future iterations once it
has been identified.

It is not hard to show that these are facet defining under relatively mild conditions.

%\begin{theorem}\label{thm:ec-facet}
%Let $P$ be a polygon with finite sets of guard and witness positions
%$G,W \subset P$, such that $\conv(\agp(G,W))$ is full-dimensional.
%Suppose $G$ and $W$ both consist of $k=2\ell+1$ positions, such that
%the set of mutually visible guard/edge pairs forms a $2k$-cycle.
%Then \eqref{eq:2guard-agr-sep} describes a class of facets for $\conv(\agp(G,W))$.
%\end{theorem}

%The situation is more complicated when the property that no guard can see
%more than two witnesses only is fulfilled for a subset of witnesses.
%We show that even under these circumstances,
%Inequality~\eqref{eq:2guard-agr-sep} still is a facet of $\conv(\agp(G,W))$ under
%some assumptions:

\begin{theorem}\label{thm:ec-facet}
Let $P$ be a polygon with finite sets of guard and witness positions
$G,W \subset P$, such that $\conv(\agp(G,W))$ is full-dimensional.
Let $\overline{W} = \left\{ w_1, \dots, w_k \right\} \subseteq W$
be an odd subset of $k \geq 3$ witnesses, such that
\begin{enumerate}
\item
	No guard sees more than two witnesses in $\overline W$.%:
%	\begin{equation}\label{eq:ec-facet-2condition}
%		\bigforall g \in G: \quad \left|\V(g) \cap \overline{W} \right| \leq 2
%	\end{equation}

\item
	If a guard sees two witnesses $w_i \neq w_j \in \overline{W}$, they are
	a \emph{successive pair}, \ie, $i+1=j$ or $i=1$ and $j=k$.

\item\label{itm:ec-successive-2}
	Each of the $k$ successive pairs is seen by some $g \in G$.

\item
	No guard inside of $\V\left( \overline{W} \right)$ sees a witness
	outside of $\overline W$.%:
%	\begin{equation}
%		\bigforall g \in G \cap \V\left( \overline{W} \right): \quad
%			\V(g) \cap W \subseteq \overline{W}
%	\end{equation}
\end{enumerate}
Then the constraint
\begin{equation}\label{eq:ec-facet}
	\sum_{g \in \V\left(\overline{W}\right) \cap G} x_g
		\geq \left\lceil \frac{\left|\overline{W}\right|}{2} \right\rceil
\end{equation}
is a facet of $\conv(\agp(G,W))$.
\end{theorem}

\begin{proof}
As no guard sees more than two witnesses of $\overline W$, it is clear that
it takes at least $\left\lceil \frac{k}{2} \right\rceil$ guards to cover
$\overline W$.
%Constraint~\eqref{eq:ec-facet} is nothing but the multiwitness of
%Equation~\eqref{eq:genwit-agr2} from Theorem~\ref{thm:genwit-agr2} and thus
%feasible for the same reasons.

It remains to be shown how to construct $n = |G|$ affinely independent solutions
of $\agp(G,W)$.
In order to do that, we separate the guards into three groups
$G_1 \dcup G_2 \dcup G_3 = G$ with $|G_i| = n_i$:
\begin{enumerate}
\item
	$G_1$ is a set of one guard for each successive pair as in
	Condition~\ref{itm:ec-successive-2}.

\item
	$G_2$ contains all guards in $\V\left( \overline{W} \right)$ that are
	not already part of $G_1$:
	\begin{equation}
		G_2 = \left( \V\left( \overline{W} \right) \cap G \right) \setminus G_1
	\end{equation}

\item
	$G_3$ holds the rest of the guards, which are outside of $\V(\overline{W})$:
	\begin{equation}
		G_3 = G \setminus \V\left( \overline{W} \right)
	\end{equation}
\end{enumerate}

In the following, we describe a solution $x \in \{0,1\}^G$ by
$x = \left(x_{G_1}, x_{G_2}, x_{G_3}\right)$, where $x_{G_i} \in \{0,1\}^{G_i}$
denotes the vector $(x_{g_1},\dots,x_{g_{n_i}})$ with
$G_i = \{g_1,\dots,g_{n_i}\}$.
The first set of $n_1$ solutions is
\begin{equation}
\begin{aligned}
	x^1 = & \left( (1,0,1,0,\dots,1,0,1), \quad \zero, \quad \one \right) \\
	x^2 = & \left( (1,1,0,1,\dots,0,1,0), \quad \zero, \quad \one \right) \\
	x^3 = & \left( (0,1,1,0,\dots,1,0,1), \quad \zero, \quad \one \right) \\
	      & \vdots \\
	x^{n_1} = & \left( (0,1,0,1,\dots,0,1,1), \quad \zero, \quad \one \right),
\end{aligned}
\end{equation}
which exists because of Condition~\ref{itm:ec-successive-2} and the choice
of $G_1$.
It fulfills~\eqref{eq:ec-facet} with equality because it uses
$\left\lceil\frac{k}{2}\right\rceil$ guards, and it is feasible, because
$\overline W$ and $W \setminus \overline W$ are covered by construction and
guards not in $G_3$ do not interfere with the coverage of witnesses
in $W \setminus \overline W$.

The second set provides $n_2$ solutions by using the $i$-th unit vector as
$x_{G_2}$.
\begin{equation}
	x^i = \left( x', \quad e_i, \quad \one \right)
\end{equation}
As every successive pair of witnesses is covered by some guard in $G_1$, a
choice of $x'$ such that $x^i$ fulfills \eqref{eq:ec-facet} with
equality is always possible.

The third and last set of $n_3$ solutions is constructed by subtracting
$e_i$ from $\one$ in the vector $x_{G_3}$:
\begin{equation}
	x^i = \left( (1,0,1,0,\dots,1,0,1), \quad \zero, \quad \one - e_i \right).
\end{equation}
It fulfills \eqref{eq:ec-facet} with equality.
Setting one guard value to zero in $G_3$ is feasible because in a
full-dimensional $\conv(\agp(G,W))$, every witness is seen by at least two
guards, compare Lemma~\ref{lem:agp-fulldim}.

All in all, we have $n_1 + n_2 + n_3 = n$ feasible, affinely
independent solutions of $\agp(G,W)$ fulfilling
\eqref{eq:ec-facet} with equality, so
\eqref{eq:ec-facet} has dimension $n-1$ and is a facet, as
claimed.
\end{proof}

\section{Computational Experience}\label{sec:exp}
%\subsection{Setup}
\old{%xxx
\begin{figure}
        \centering
        %\subfigure[Solution without Cuts]{
                \def\svgwidth{.48\linewidth}
                \input{testcase_fractional.pdf_tex}
                \label{fig:testcase-fractional}
        %}
        %\subfigure[Solution with Cuts]{
                \def\svgwidth{.48\linewidth}
                \input{testcase_cuts.pdf_tex}
                \label{fig:testcase-cuts}
        %}
        \caption[Solution without and with Cuts]{%
        A optimal fractional solution of value 5 without (left) and an optimal integer solution of value 6
        with cutting planes (right).
        Large circles indicate guards, fill-in corresponds to fractional amount;
        small blue circles indicate unused guard positions, while green circles indicate witnesses.
        The color gradient in the polygon indicates coverage: gray means one, bluish is
        more than one and green would be less than one, which must not happen in
        a feasible solution.
        %The fractional solution shown in the left figure  uses five guards.
        Cutting planes enforce at least two guards in the left red area
        and at least three guards in the right red area.
        %This cuts off the fractional solution and enforces one with six guards,
        %\ie, one more than above, see Figure~\subref{fig:testcase-cuts}.
        %It is both integral and optimal.
        }
        \label{fig:testcase}
\end{figure}
}%yyy

A variety of experiments on benchmark polygons
demonstrates the usefulness of our cutting planes,
as well as the appropriateness of our separation heuristic of using only $k=3$
for the SC related facets from Section~\ref{sec:sc}.

We test our cutting planes in two variations of our algorithm, IP and LP mode, \ie, Algorithms~\ref{alg:ip} and~\ref{alg:lp} from Section~\ref{sec:algorithms}.
An in-depth presentation of the results is conducted in Sections~\ref{sec:ip} and~\ref{sec:lp}.

\begin{figure}
	\centering
%	\subfigure[von Koch]{
	\subfigure{
		\def\svgwidth{.15\linewidth}
		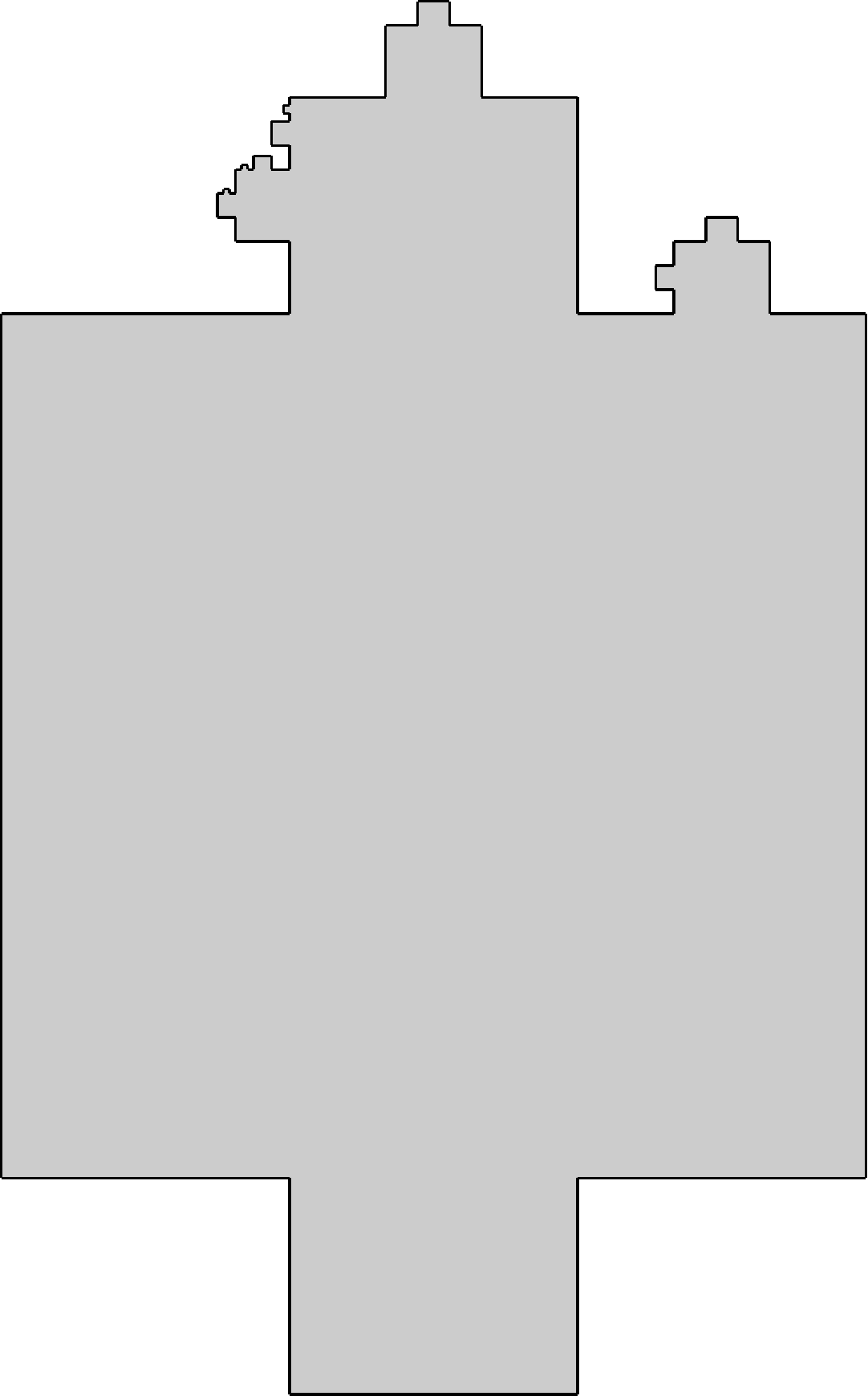
		\label{fig:testpolygon-koch}
	}
%	\subfigure[Orthogonal]{
	\subfigure{
		\def\svgwidth{.15\linewidth}
		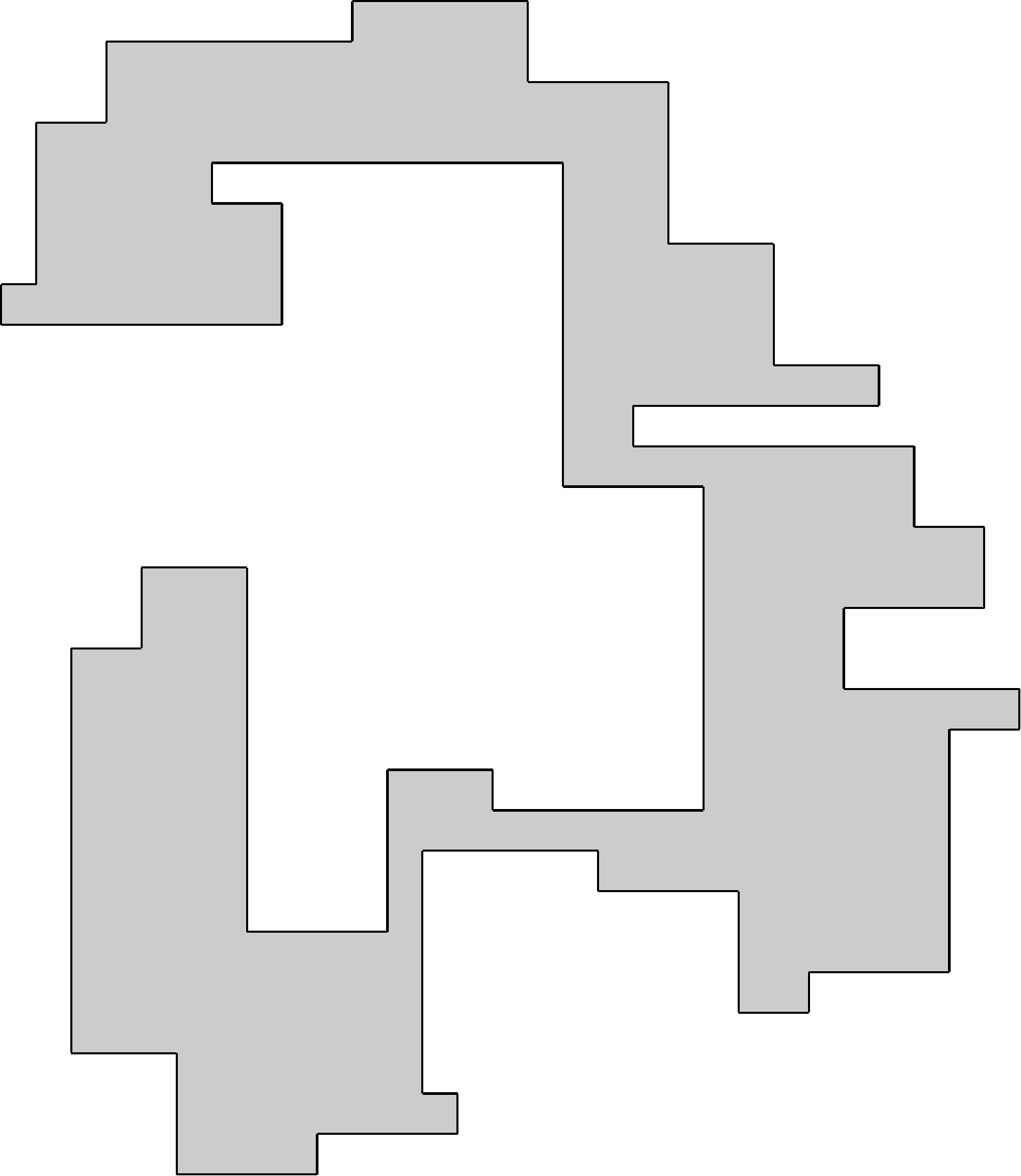
		\label{fig:testpolygon-orthogonal}
	}
%	\subfigure[Simple]{
	\subfigure{
		\def\svgwidth{.4\linewidth}
		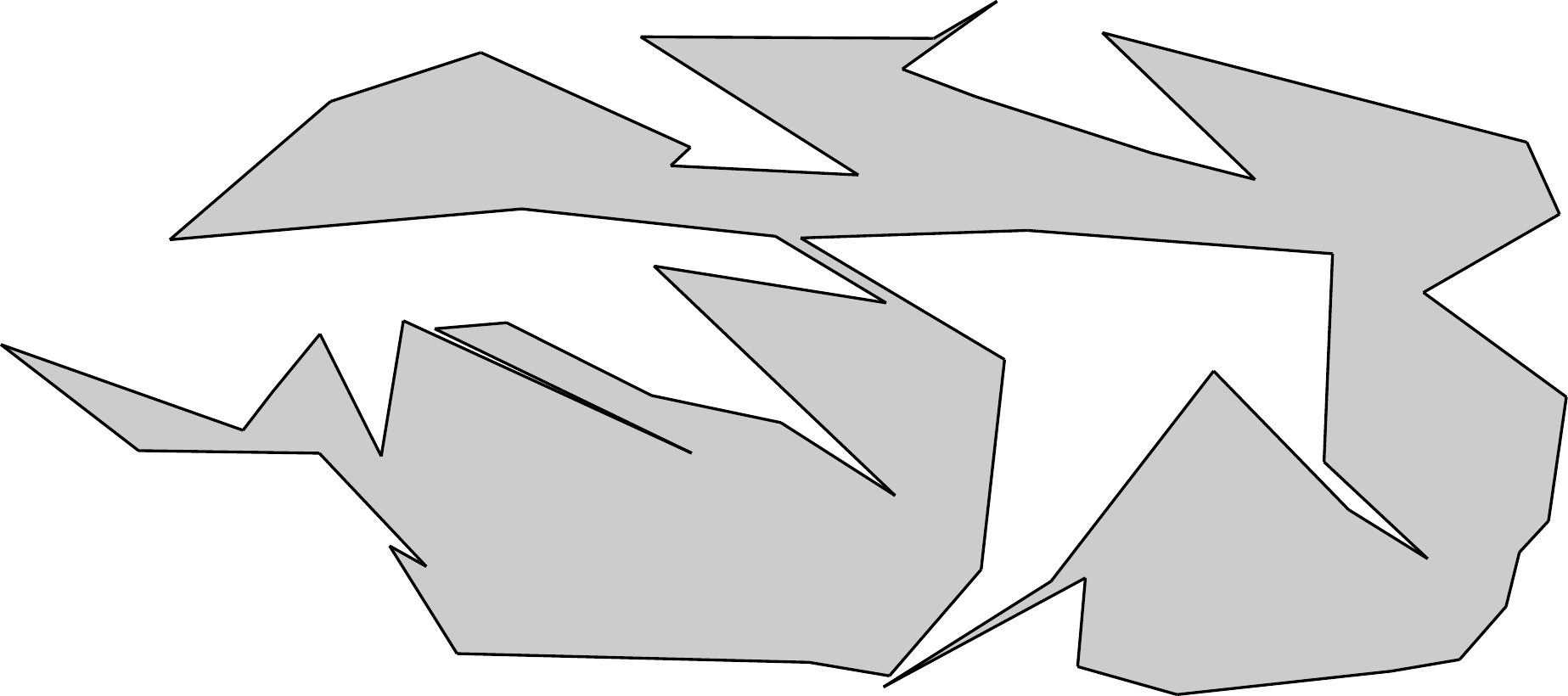
		\label{fig:testpolygon-simple}
	}
%	\subfigure[Spike]{
	\subfigure{
		\def\svgwidth{.15\linewidth}
		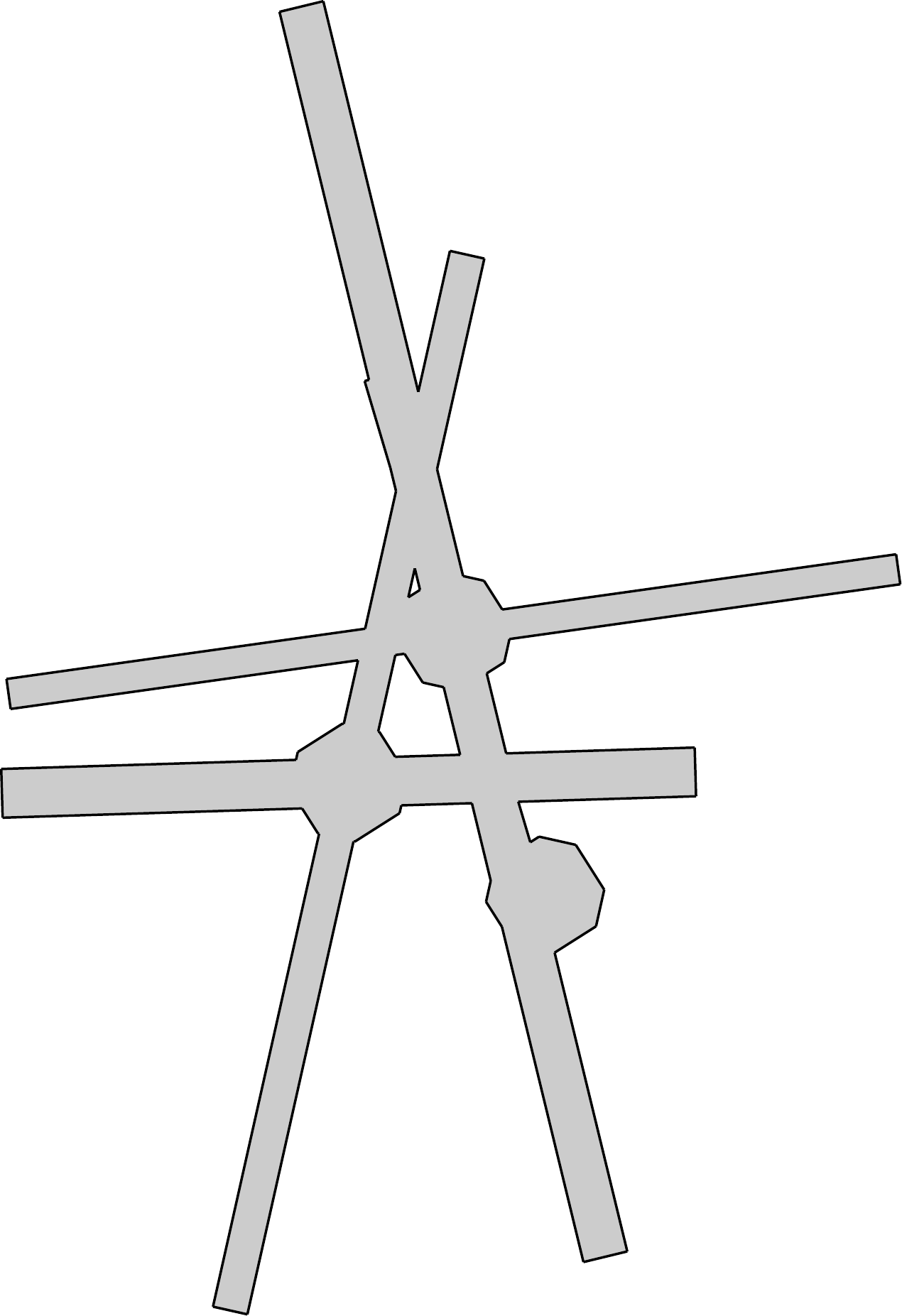
		\label{fig:testpolygon-spike}
	}
	\caption[Test Polygon Classes]{%
	Small \emph{von Koch}, \emph{Orthogonal}, \emph{Simple} and \emph{Spike} test polygons.
%	Random test polygons, here with approximately 60 vertices each.
	%All of them are produced by randomized algorithms.
%	The \emph{von Koch} type, compare Figure~\subref{fig:testpolygon-koch}, is
%	inspired by the Koch curve named after Helge von~Koch.
%	Another type of test polygon, the \emph{Orthogonal} type in
%	Figure~\subref{fig:testpolygon-orthogonal} roughly looks like the
%	floor plan of a building.
%	While the aforementioned polygons are orthogonal, the remaining two
%	are not.
%	The \emph{Spike} type in Figure~\subref{fig:testpolygon-spike} is used
%	to enforce the separation of guard locations within the intersection of
%	its long corridors.
%	Most of these instances contain holes.
%	Figure~\subref{fig:testpolygon-simple} shows a \emph{Simple} type
%	polygon that is connected and non-orthogonal, but contains no holes.
	}
	\label{fig:testpolygon}
%\vspace*{-5mm}
\end{figure}

Just as in~\cite{kbfs-esbgagp-12}, we employed
four different classes of benchmark polygons.
\begin{enumerate}
\item
	Random \emph{von Koch} polygons are inspired by fractal Koch curves,
	see Figure~\ref{fig:testpolygon}, left.

\item
	Random floorplan-like \emph{Orthogonal} polygons as in
	Figure~\ref{fig:testpolygon}, second polygon.

\item
	Random non-orthogonal \emph{Simple} polygons as in
	Figure~\ref{fig:testpolygon}, third polygon.

\item
	Random \emph{Spike} polygons (mostly with holes) %that are shaped in a way that
	%enforces the separation of specific guard locations.
	as in Figure~\ref{fig:testpolygon}, fourth polygon.

\end{enumerate}
Each polygon class was evaluated for different sizes
$n \in \{ 60, 200, 500, 1000 \}$, where $n$ is the approximate number
of vertices in a polygon.

Different combinations of cut separators were also employed.
The EC-related cuts from Section~\ref{sec:ec} are referred to as
\emph{EC cuts}, while
the \acs{SC}-related cuts of Section~\ref{sec:sc} that
rely on separating a maximum of $3 \leq k$ witnesses are denoted by \emph{SC$k$} cuts.
Note that for $3 \leq m \leq k$, \acs{SC}$k$ cuts also include all \acs{SC}$m$ cuts.

Whenever our algorithm separates cuts, it applies all configured cut
separators and we test the following combinations:
no cut separation at all, \acs{SC}3 cuts only, \acs{SC}4 cuts only,
\acs{EC} cuts only, and \acs{SC}3 and \acs{EC} cuts at the same time.

In total, we have two modes, five combinations of separators, four classes of polygons, and four
polygon sizes; for each combination, we tested 10 different polygons.
The experiments were run on \SI{3.0}{GHz} Intel dual core PCs with
\SI{2}{GB} of memory, running 32 bit Debian 6.0.5 with Linux 2.6.32-686.
Our algorithms were not parallelized, used version 4.0
of the ``Computational Geometry Algorithms Library'' (CGAL)
and CPLEX 12.1. Each test run had a time limit of \SI{600}{s}.

%\subsection{Results}
%\label{sec:experiments-results-gaps}

%We use the following lines to introduce a little notation.
%Then we describe the impact of the separator combinations on the gap
%development over time.
%The results for \acs{LP} and \acs{IP} mode are presented in dedicated
%subsections.

In the remaining part of this section, we refer to quartiles by \ac{Q0}, \ac{Q1}, \ac{Q2}, \ac{Q3} and \ac{Q4}.
\ac{Q1} is the first quartile, which is between the lowest 25\% and the rest
of the values.
\ac{Q2} is the second quartile or the median value and \ac{Q3}, the third
quartile, splits the upper 25\% from the lower 75\%.
For the sake of simplicity, the minimum and the maximum are denoted by \ac{Q0} and \ac{Q4},
respectively.

%The gaps are expressed as relative gaps in percent,
%a value allowing comparison between gaps when the absolute bounds of
%multiple test runs differ.
%The relative gap percentage between an upper bound $u$ and a lower bound $l$
%is $100 \cdot \left(\frac{u}{l} - 1\right)$.

%We choose one polygon size for each class of input polygon:
%The diagrams of smaller instances do not provide much information,
%because the gaps are constant most of the time.
%In this abstract,

\subsection{IP Mode}
\label{sec:ip}

The IP mode, Algorithm~\ref{alg:ip}, is a variation of the one introduced
in~\cite{kbfs-esbgagp-12}, which always determines binary solutions at the
expense of not necessarily terminating due to the integrality gap.
Our experiments confirm that the integrality gap is drastically reduced by our cutting planes.

\begin{figure}
	\subfigure[No Cuts]{
		\includegraphics[clip,width=.45\linewidth]%
				{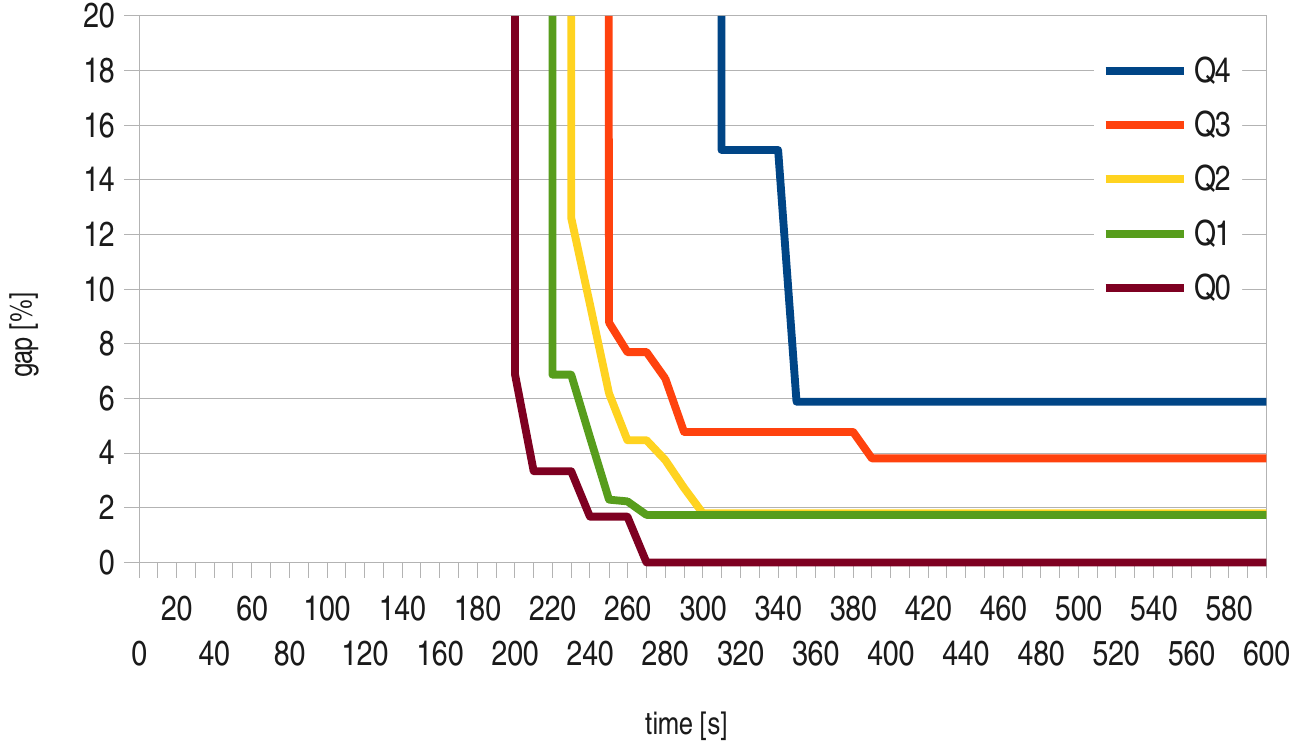}
		\label{fig:ip-ko1000}
	}
	\subfigure[\acs{EC}]{
		\includegraphics[clip,width=.45\linewidth]%
				{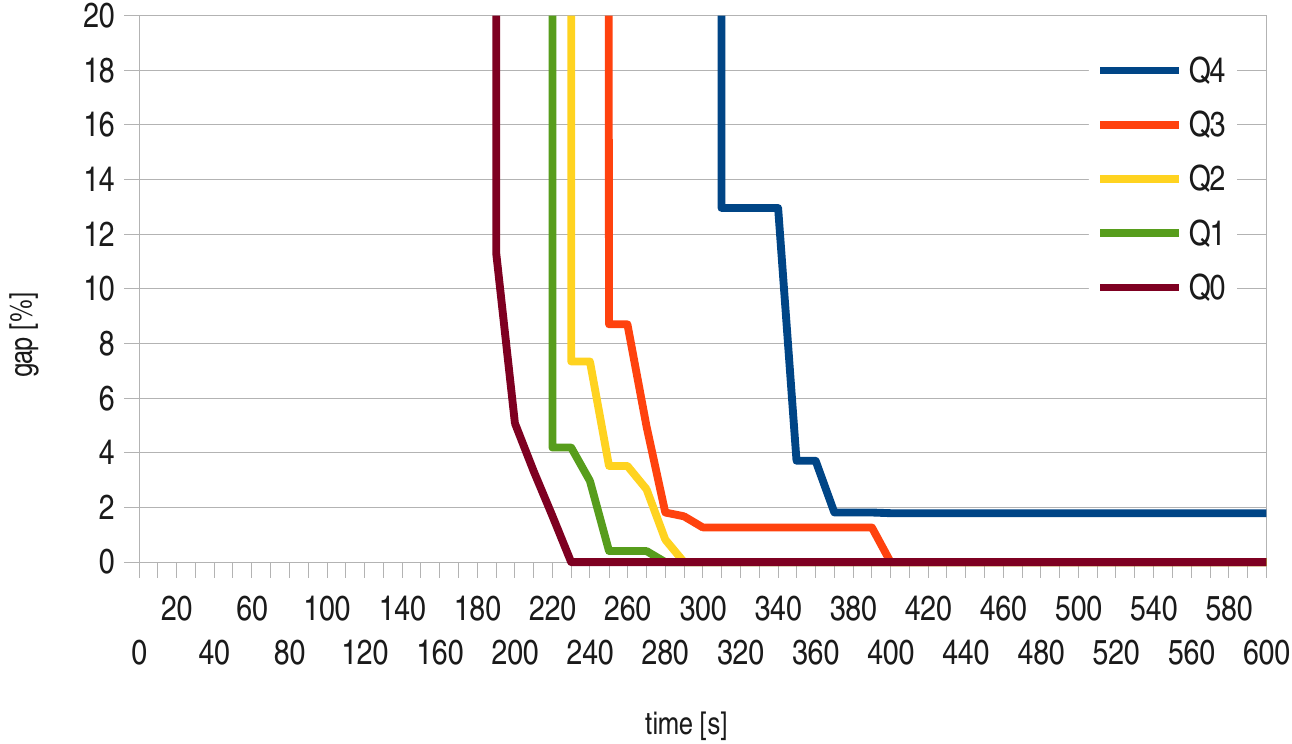}
		\label{fig:ipoc-ko1000}
	}
	\subfigure[\acs{SC}3]{
		\includegraphics[clip,width=.45\linewidth]%
				{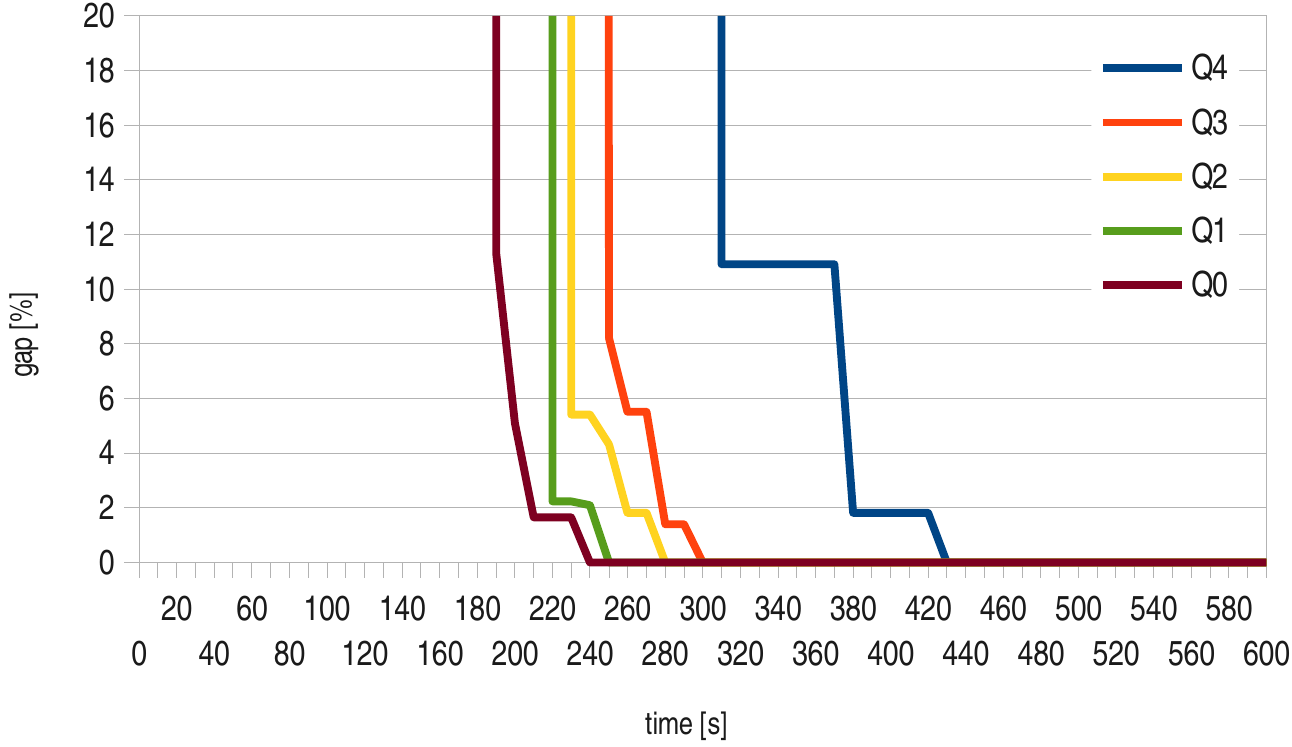}
		\label{fig:ipsc3-ko1000}
	}
	\subfigure[\acs{SC}3 and \acs{EC}]{
		\includegraphics[clip,width=.45\linewidth]%
				{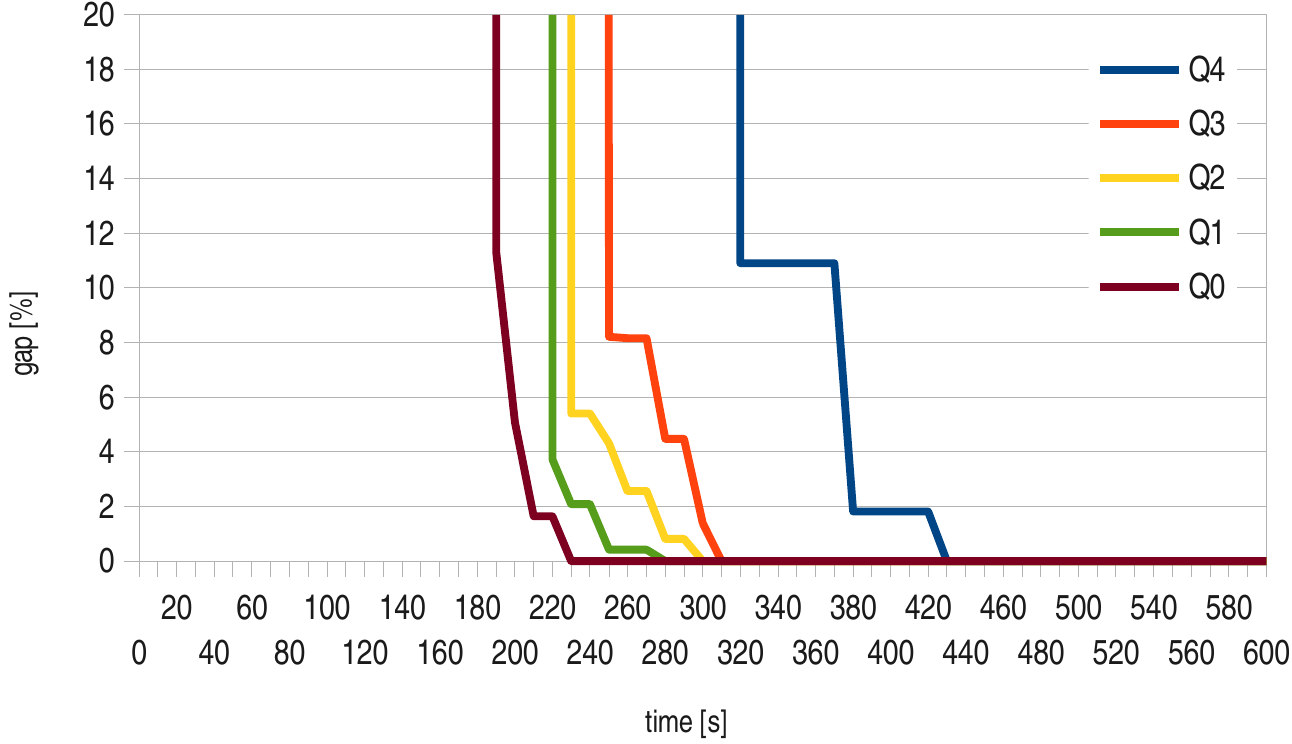}
		\label{fig:ipsc3oc-ko1000}
	}
	\subfigure[\acs{SC}4]{
		\includegraphics[clip,width=.45\linewidth]%
				{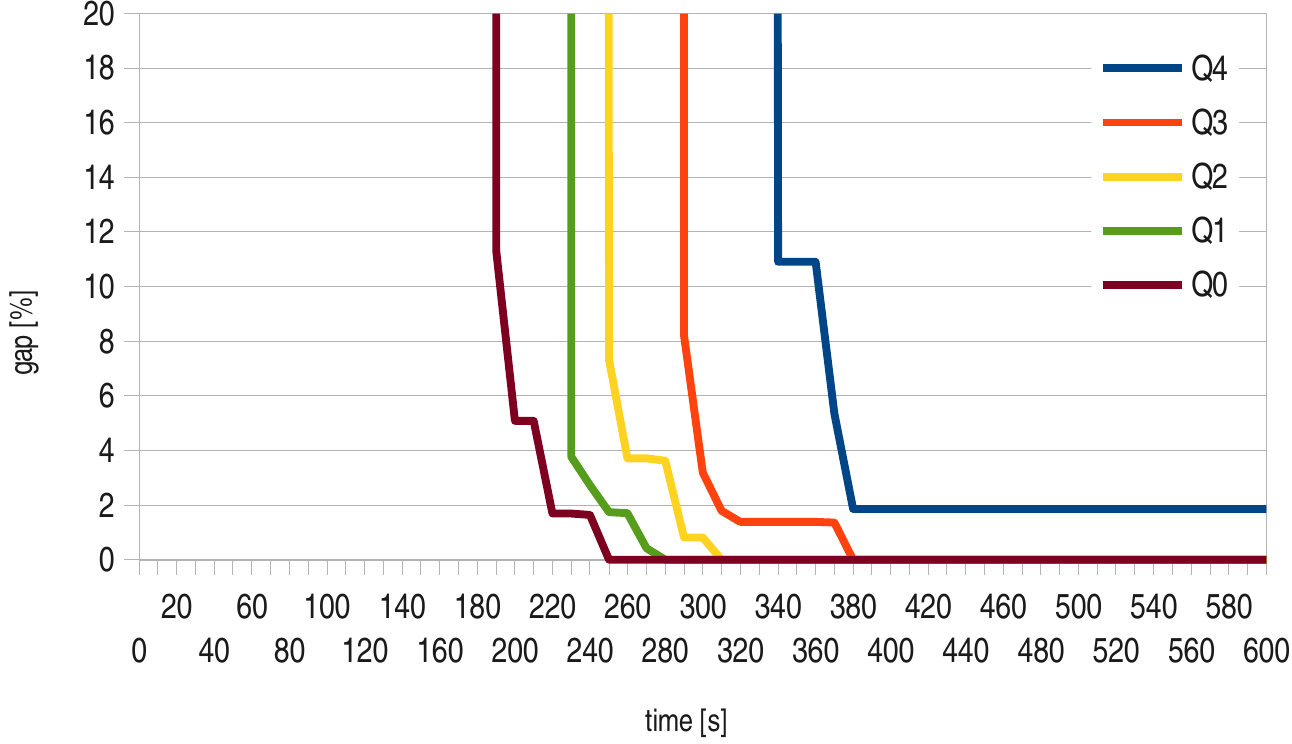}
		\label{fig:ipsc4-ko1000}
	}
	\caption[Relative Gap over Time in \acs{IP}-Mode: von Koch]{%
	Relative gap over time in \acs{IP}-mode for 1000-vertex
	\emph{von Koch}-type polygons.
	}
	\label{fig:ip-ko1000-all}
%\vspace*{-5mm}
\end{figure}

In Figure~\ref{fig:ip-ko1000-all} we present the relative gap over time for the five tested cut
separator selections for the \emph{von Koch}-type polygons with 1000
vertices.
Figure~\ref{fig:ip-ko1000}
shows the relative gap over time without cut
separation.
After about \SI{400}{s}, gaps are fixed between~0\% and~6\%,
the median gap being~2\%.
When applying the \acs{EC} separator (Figure~\ref{fig:ipoc-ko1000}),
75\% of the gaps drop to zero and the largest gap is 2\%.
%a strong improvement over using no separation.
Using the \acs{SC}3 separator (Figure~\ref{fig:ipsc3-ko1000})
yields an even better result in terms of both speed
and relative gap.
All gaps are closed, many of them earlier than with the
\acs{EC} separator.
Combining both, see Figure~\ref{fig:ipsc3oc-ko1000}, yields a result
comparable to using only \acs{SC}3.
Moving to the \acs{SC}4 separator (Figure~\ref{fig:ipsc4-ko1000}) yields a weaker performance:
computation times go up, and not all gaps reach 0\% within the allotted time,
because separation takes longer without improving the gap.
This illustrates the practical consequences of Theorem~\ref{thm:fcp-star}.
%illustrated in Figure~\ref{fig:sc012} and motivated before Lemma~\ref{lem:2-composition}.

\old{
Let $W = \{w_1,\dots,w_k\}$ be a set of witnesses inducing the
\acs{SC}-related constraint $\alpha_W$ and $G = \{g_1,\dots,g_k\}$ guards in
$\V(W)$, such that $\V(g_i) \cap W = W \setminus \{w_i\}$.
This is the case for a minimal $W$ in \acs{SC}-related constraints by
Lemma~\ref{lem:ag012-ck}.
Then, by Theorem~\ref{thm:fcp-star} and for $k \geq 4$, there are two
possibilities, corresponding to the two cases illustrated in Figure~\ref{fig:sc012}:
Either $\V(W)$ is star-shaped, and an additional guard $g^*$ in the kernel will
already yield an optimal solution, except for some highly complicated scenarios;
or we get dual separation, \ie, an additional witness position $w$ that already
separates the current fractional solution.
}

\old{
\begin{enumerate}
\item
If some $w \in \V(W)$ sees less then $k-1$ guards in $G$, \ie,
$|\V(w) \cap G| \leq k-2$, then
$\V(W)$ is no full circulant polygon by Definition~\ref{def:fcp}, and
the optimal solution of the underlying SC
may assign $\frac{1}{k-1}$ to each guard position, covers $W$ but not $w$.
Such a situation causes dual separation to identify $w$ or an equivalent
point, which in turn yields a cut constraint,
so $\alpha_W$ is not needed to separate the current fractional solution.

\item
Otherwise we know by Theorem~\ref{thm:fcp-star} that $\V(W)$ is star-shaped,
and a $1$-guard may be placed within its kernel,
covering $\V(W)$. There are some relatively complicated situations in which this
is not yet optimal, but these occur too rarely and are too time-consuming to detect
in order to warrant employing these cuts.
\todo{Are these complicated cases are mentioned earlier in the final version?
If so, just use a reference}
\end{enumerate}
}

\begin{figure}
	\subfigure[No Cuts]{
		\includegraphics[clip,width=.45\linewidth]%
				{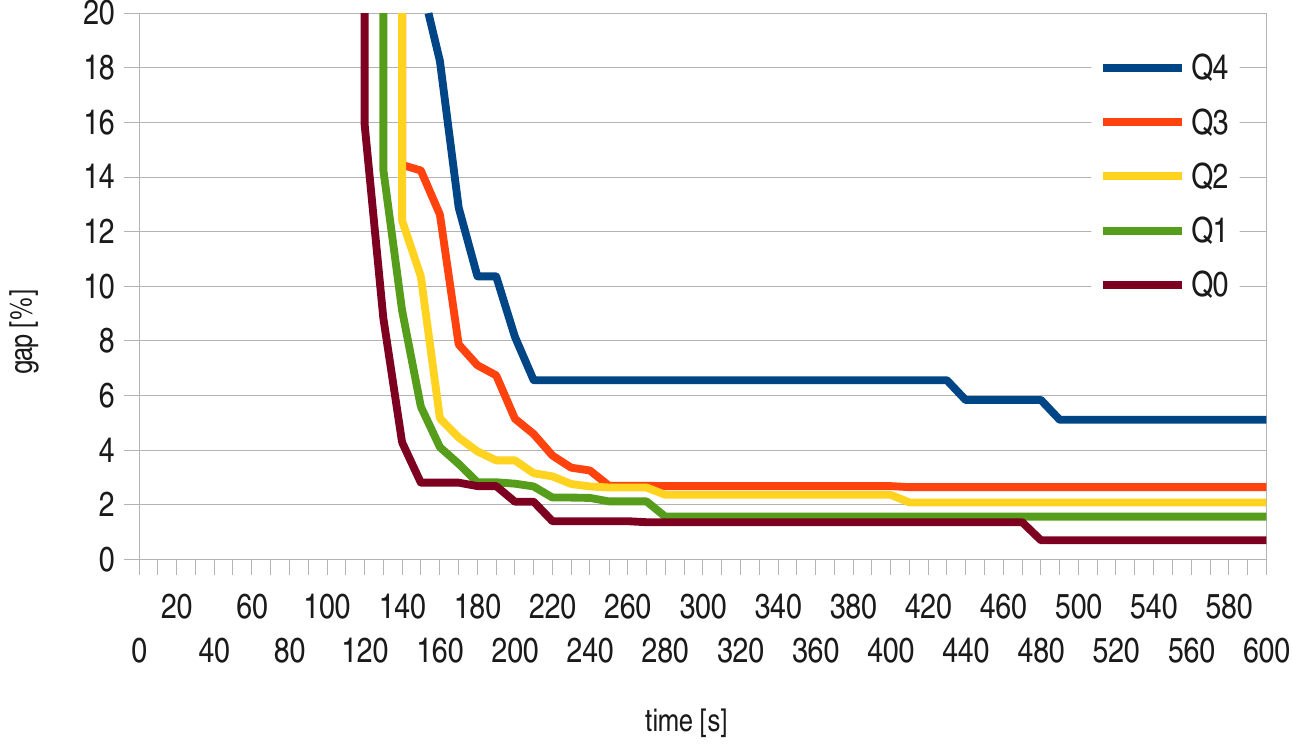}
		\label{fig:ip-or1000}
	}
	\subfigure[\acs{EC}]{
		\includegraphics[clip,width=.45\linewidth]%
				{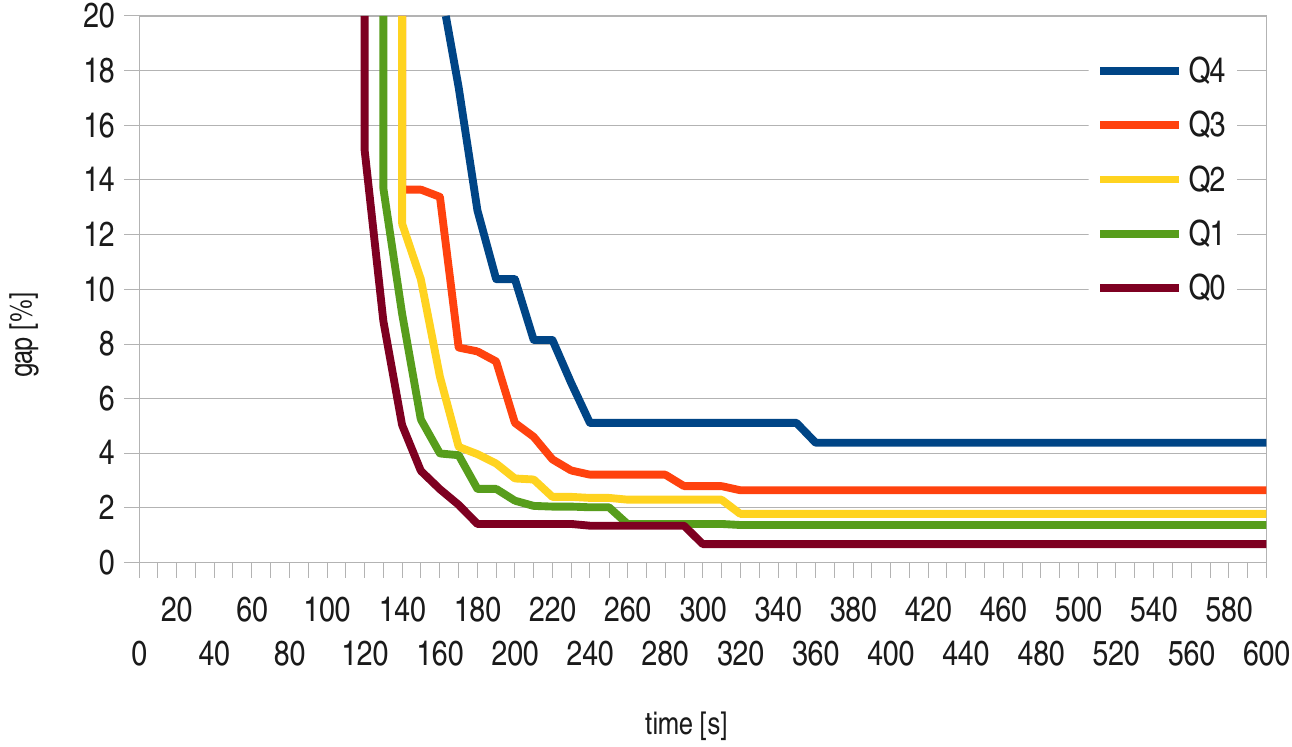}
		\label{fig:ipoc-or1000}
	}
	\subfigure[\acs{SC}3]{
		\includegraphics[clip,width=.45\linewidth]%
				{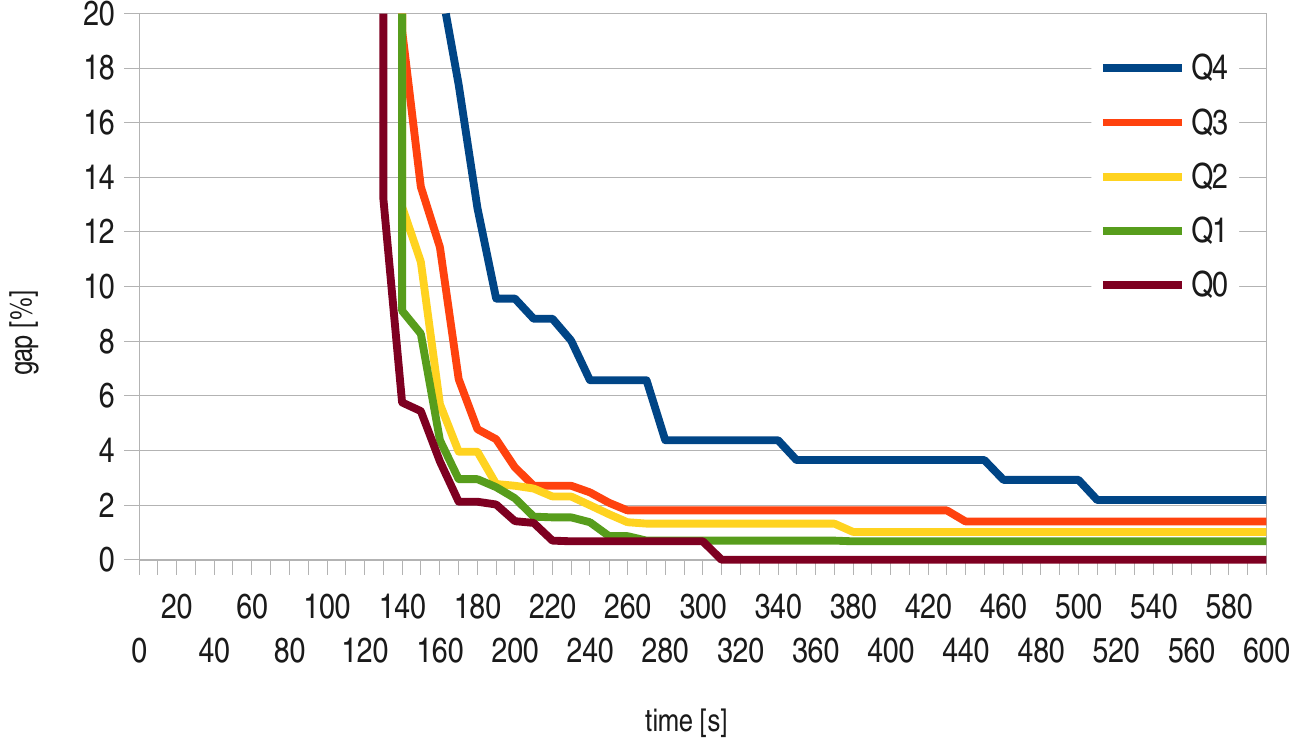}
		\label{fig:ipsc3-or1000}
	}
	\subfigure[\acs{SC}3 and \acs{EC}]{
		\includegraphics[clip,width=.45\linewidth]%
				{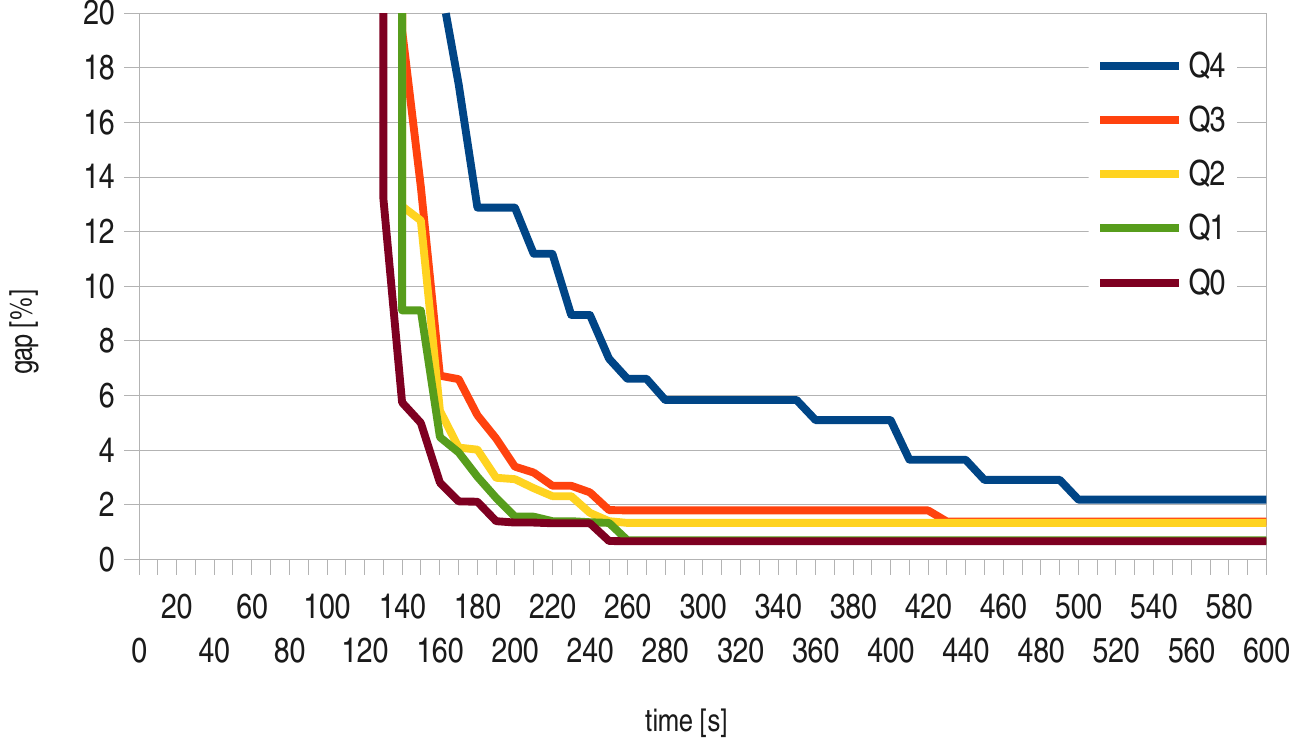}
		\label{fig:ipsc3oc-or1000}
	}
	\subfigure[\acs{SC}4]{
		\includegraphics[clip,width=.45\linewidth]%
				{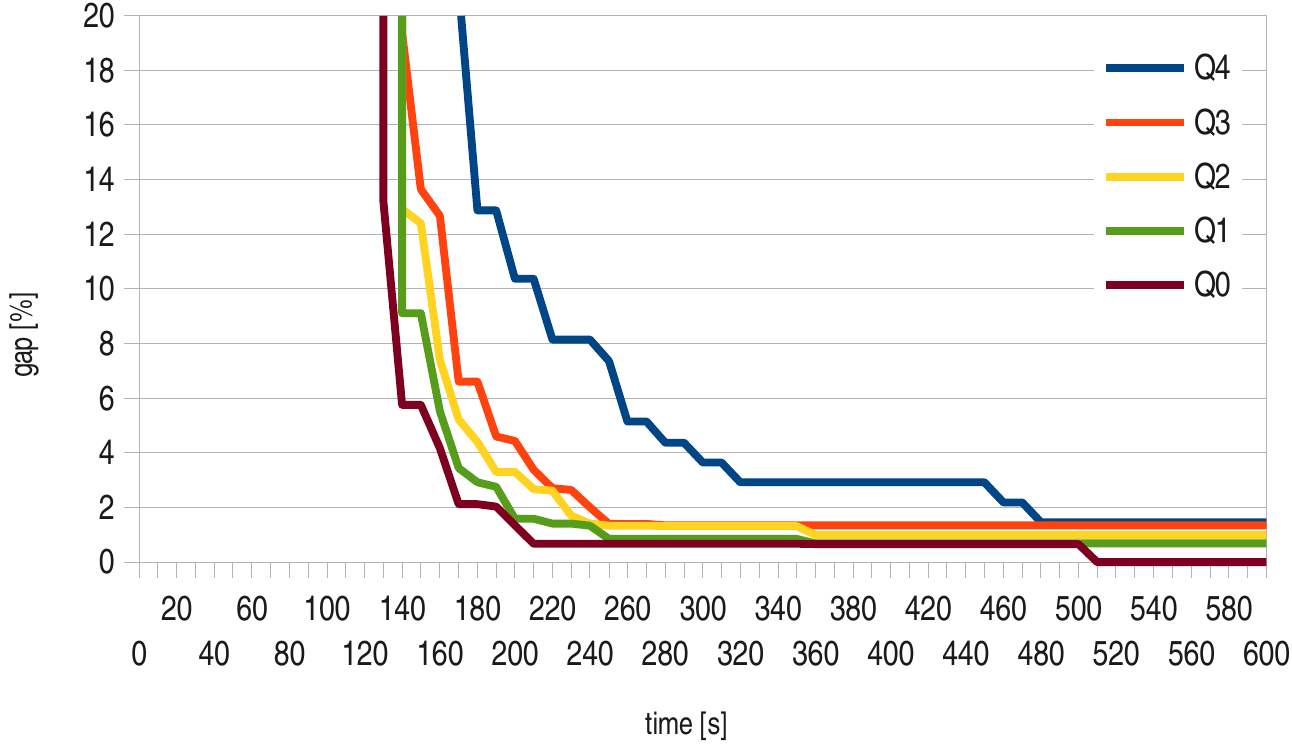}
		\label{fig:ipsc4-or1000}
	}
	\caption[Relative Gap over Time in \acs{IP}-Mode: Orthogonal]{%
	Relative gap over time in \acs{IP}-mode for the 1000-vertex
	\emph{Orthogonal}-type polygons.
	}
	\label{fig:ip-or1000-all}
\end{figure}

\begin{figure}
	\subfigure[No Cuts]{
		\includegraphics[clip,width=.45\linewidth]%
				{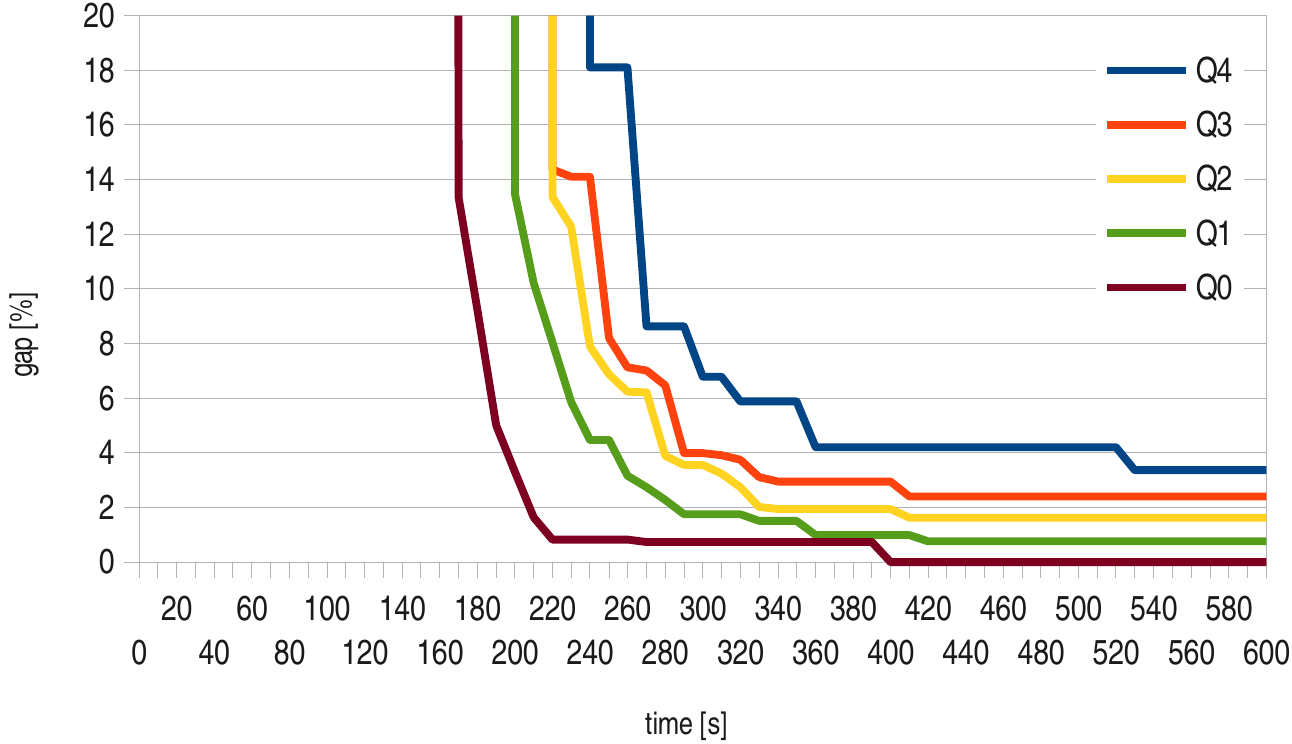}
		\label{fig:ip-si1000}
	}
	\subfigure[\acs{EC}]{
		\includegraphics[clip,width=.45\linewidth]%
				{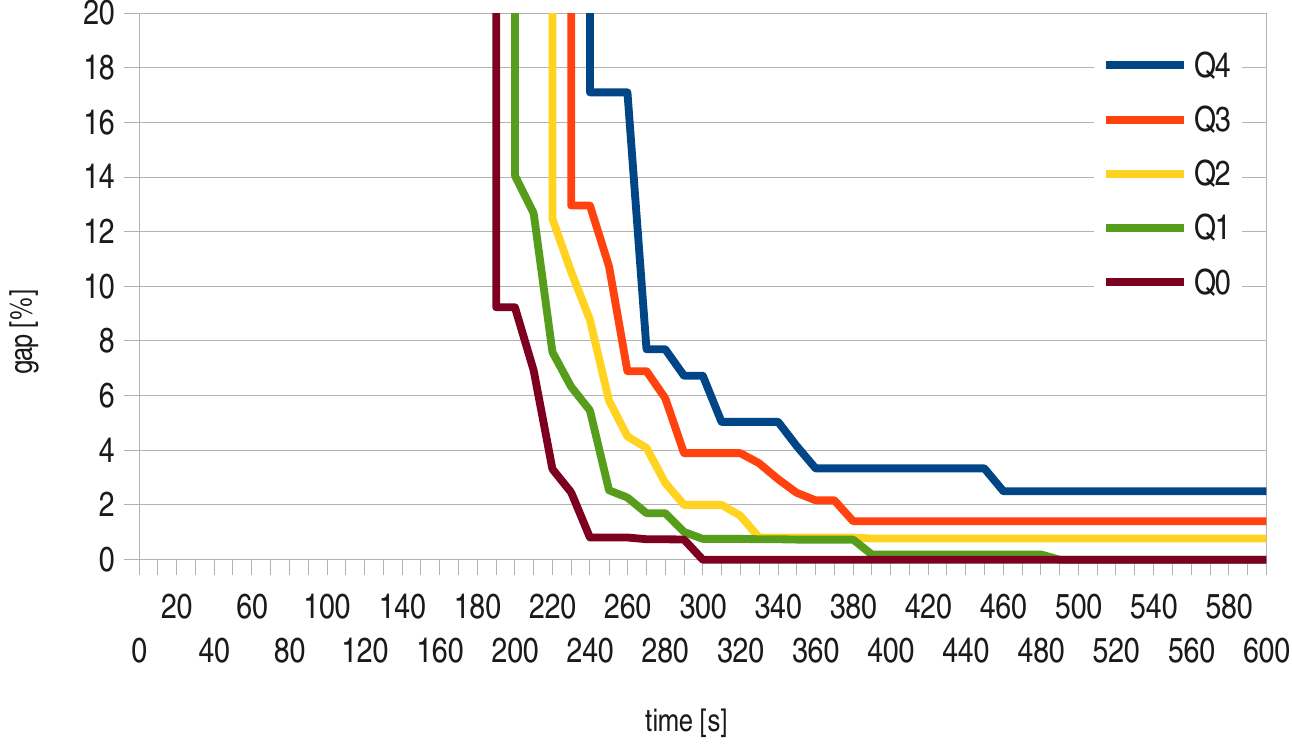}
		\label{fig:ipoc-si1000}
	}
	\subfigure[\acs{SC}3]{
		\includegraphics[clip,width=.45\linewidth]%
				{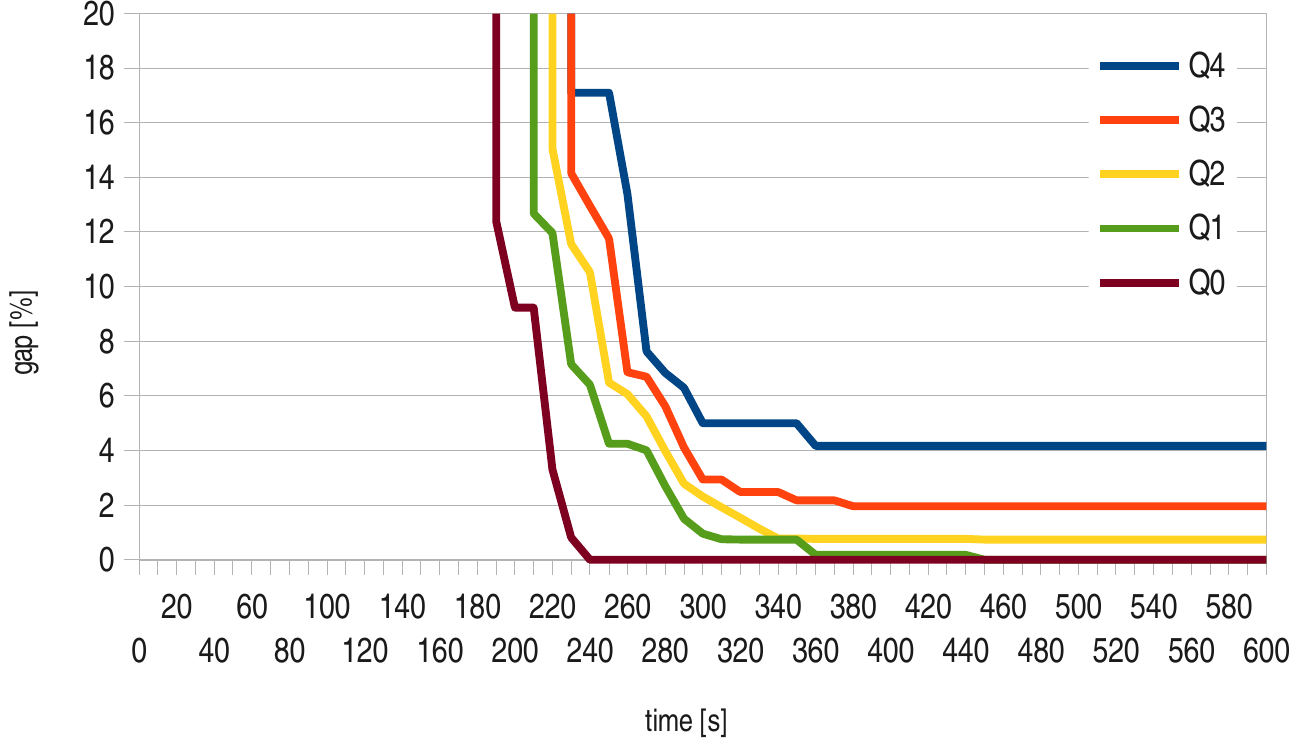}
		\label{fig:ipsc3-si1000}
	}
	\subfigure[\acs{SC}3 and \acs{EC}]{
		\includegraphics[clip,width=.45\linewidth]%
				{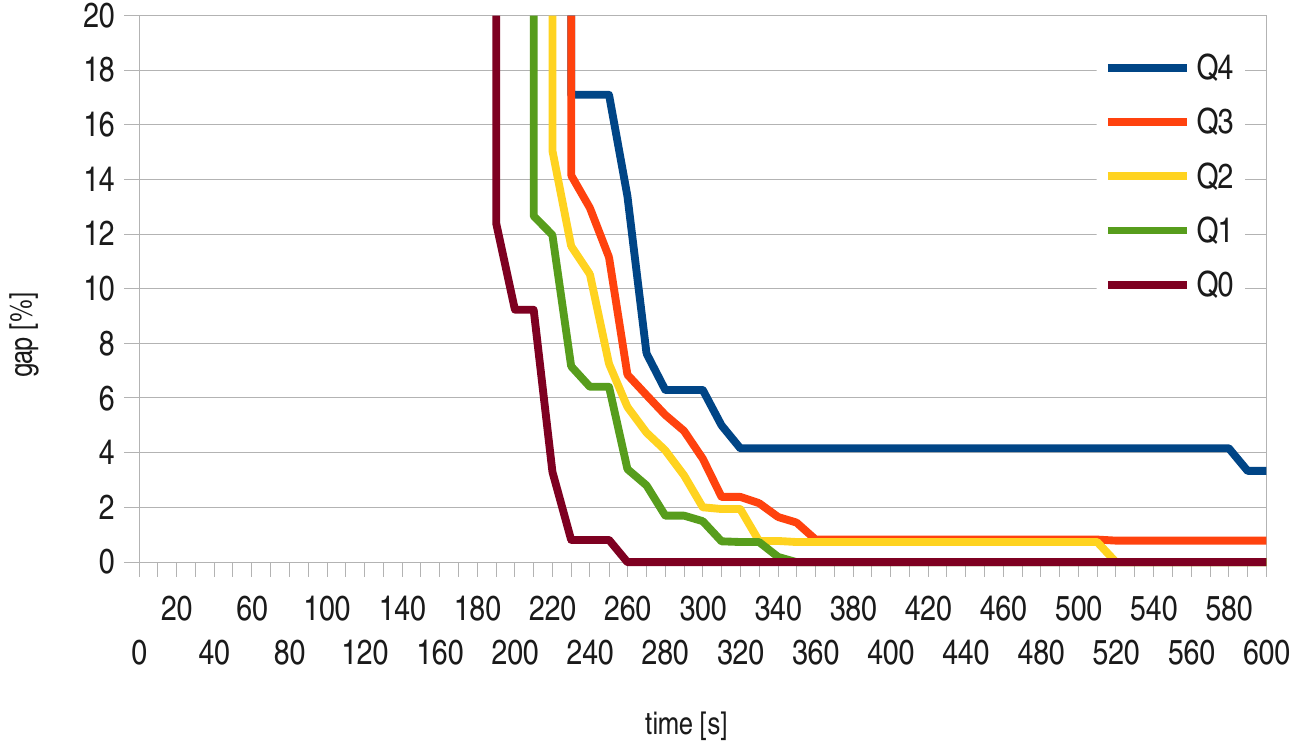}
		\label{fig:ipsc3oc-si1000}
	}
	\subfigure[\acs{SC}4]{
		\includegraphics[clip,width=.45\linewidth]%
				{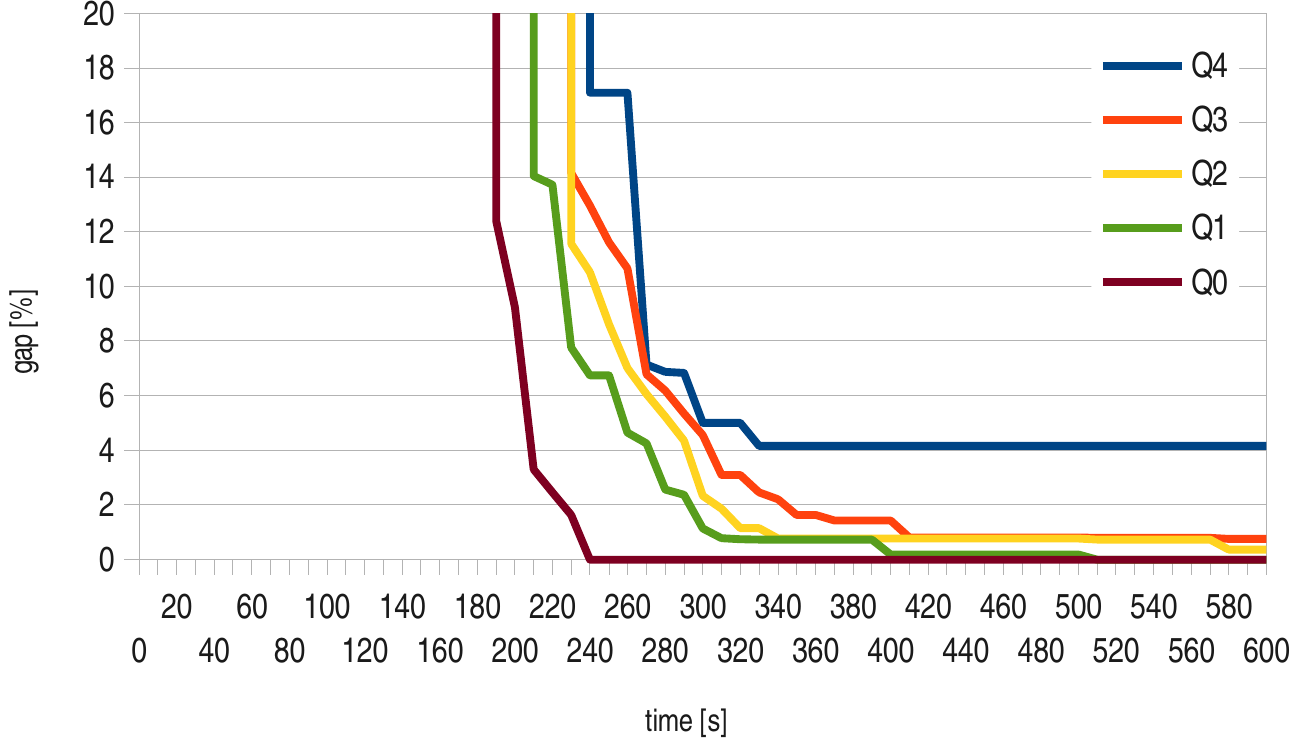}
		\label{fig:ipsc4-si1000}
	}
	\caption[Relative Gap over Time in \acs{IP}-Mode: Simple]{%
	Relative gap over time in \acs{IP}-mode for the 1000-vertex
	\emph{Simple}-type polygons.
	}
	\label{fig:ip-si1000-all}
\end{figure}

\begin{figure}
	\subfigure[No Cuts]{
		\includegraphics[clip,width=.45\linewidth]%
				{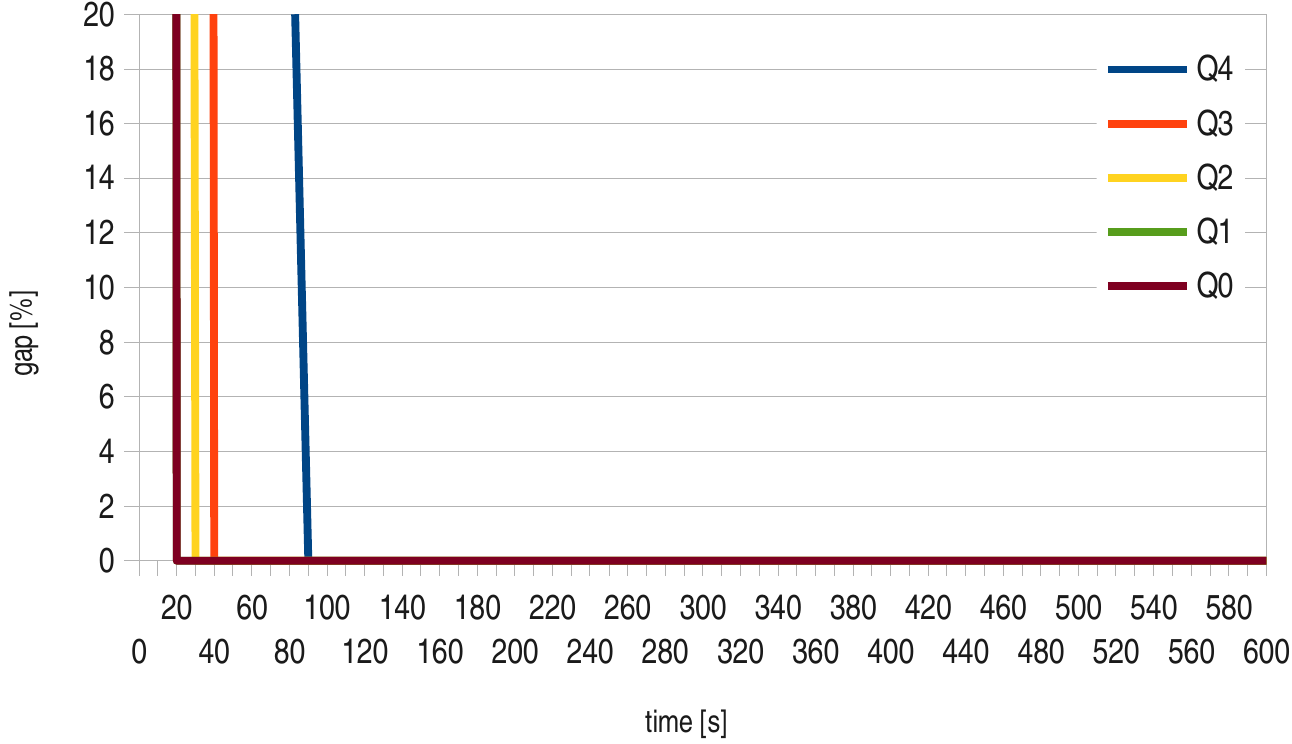}
		\label{fig:ip-sp200}
	}
	\subfigure[\acs{EC}]{
		\includegraphics[clip,width=.45\linewidth]%
				{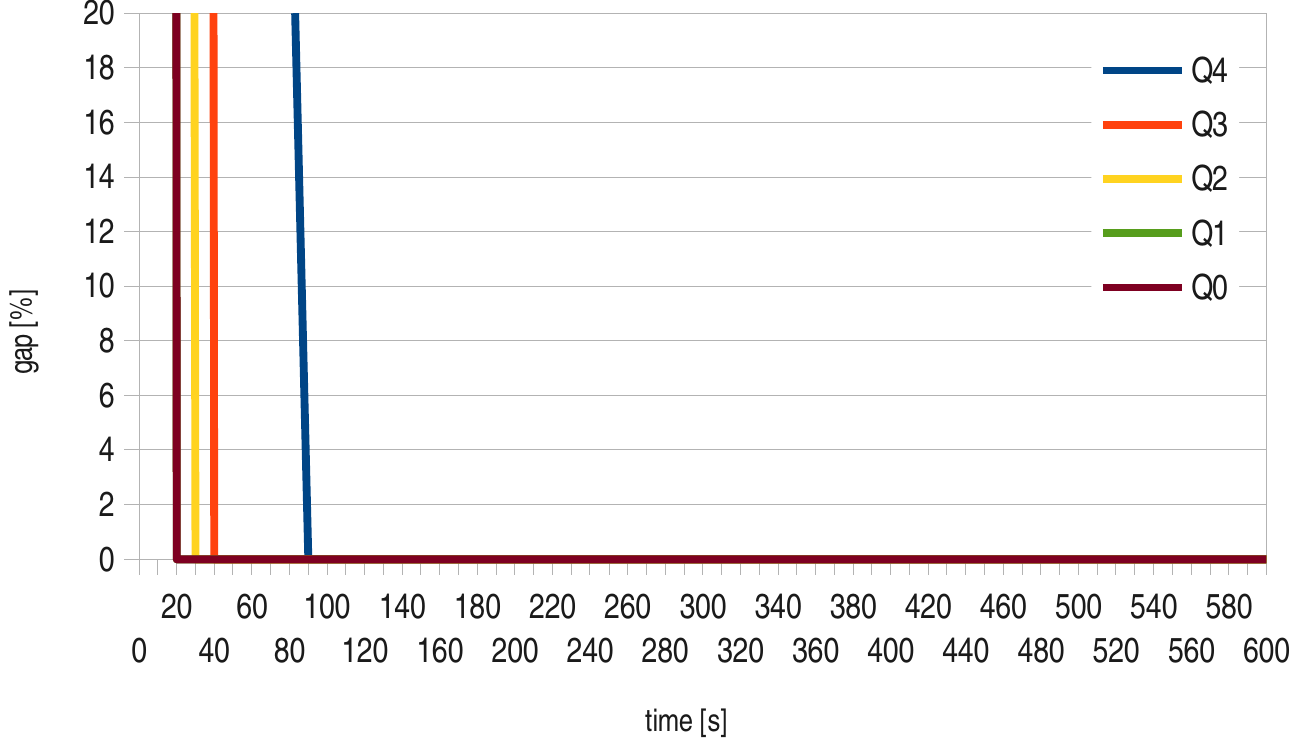}
		\label{fig:ipoc-sp200}
	}
	\subfigure[\acs{SC}3]{
		\includegraphics[clip,width=.45\linewidth]%
				{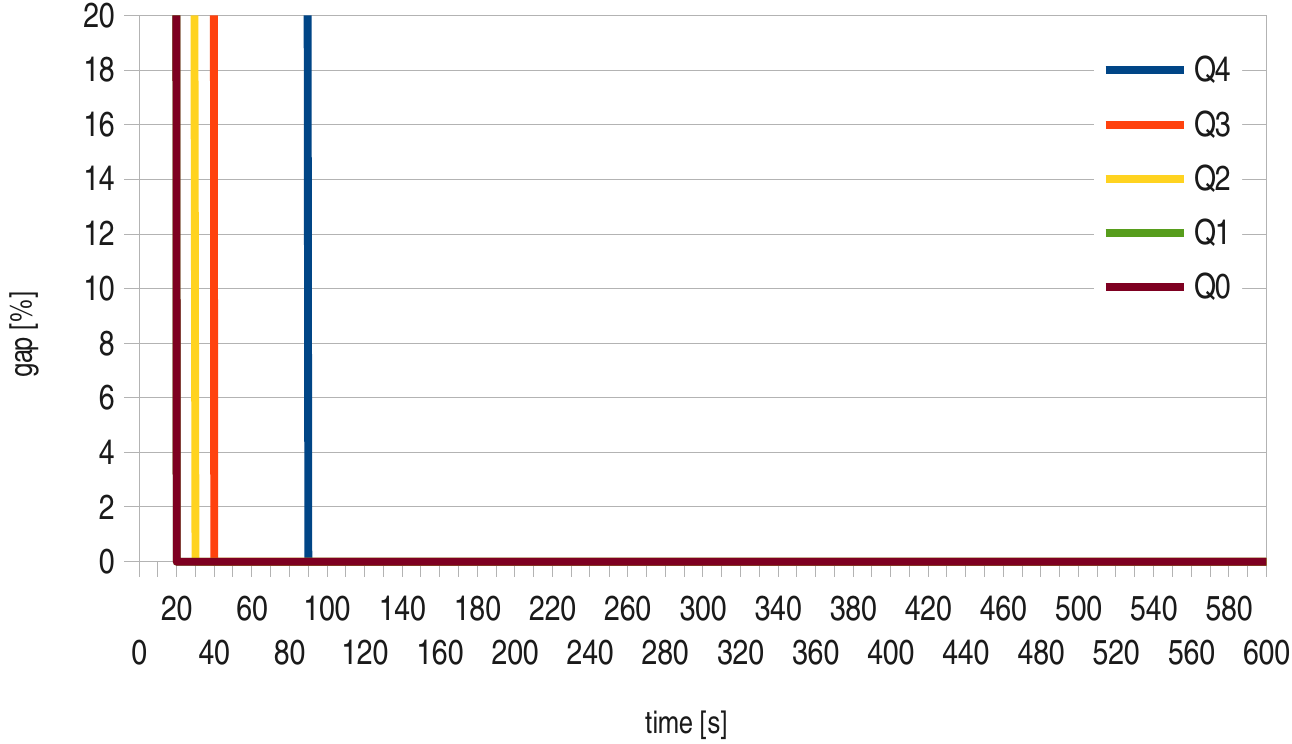}
		\label{fig:ipsc3-sp200}
	}
	\subfigure[\acs{SC}3 and \acs{EC}]{
		\includegraphics[clip,width=.45\linewidth]%
				{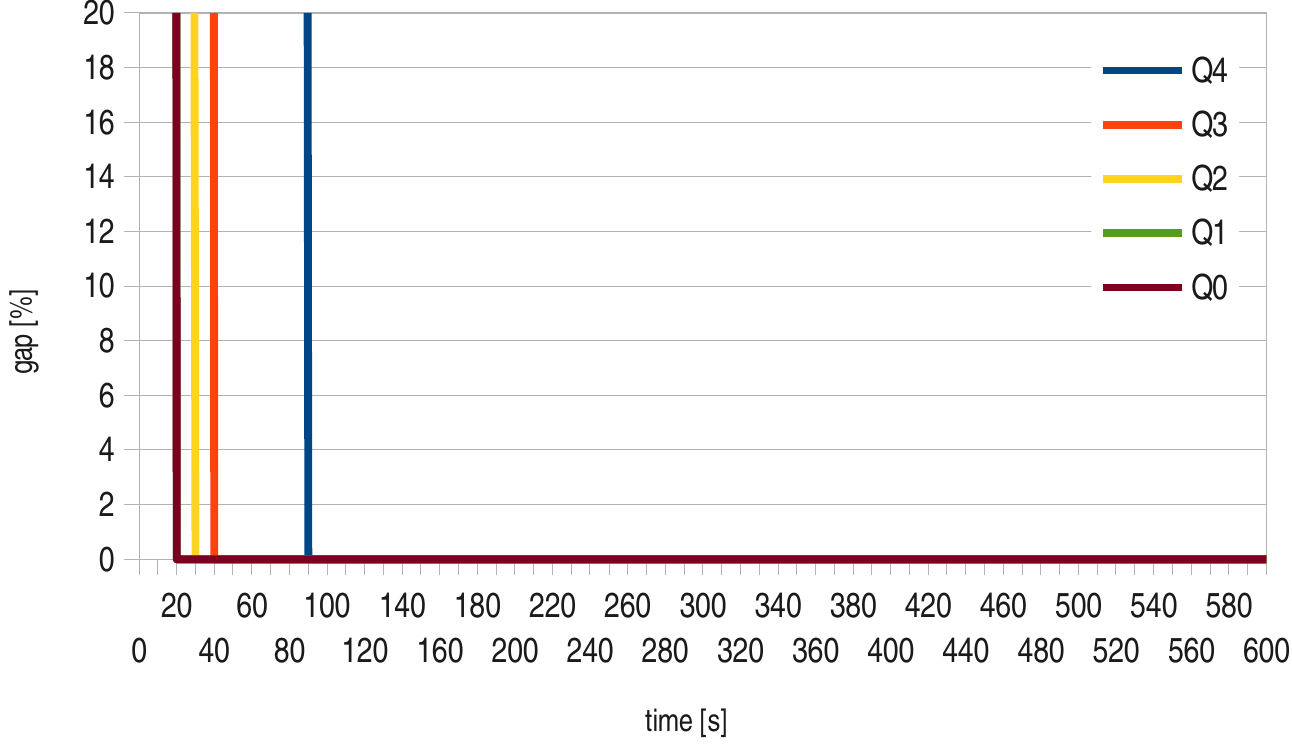}
		\label{fig:ipsc3oc-sp200}
	}
	\subfigure[\acs{SC}4]{
		\includegraphics[clip,width=.45\linewidth]%
				{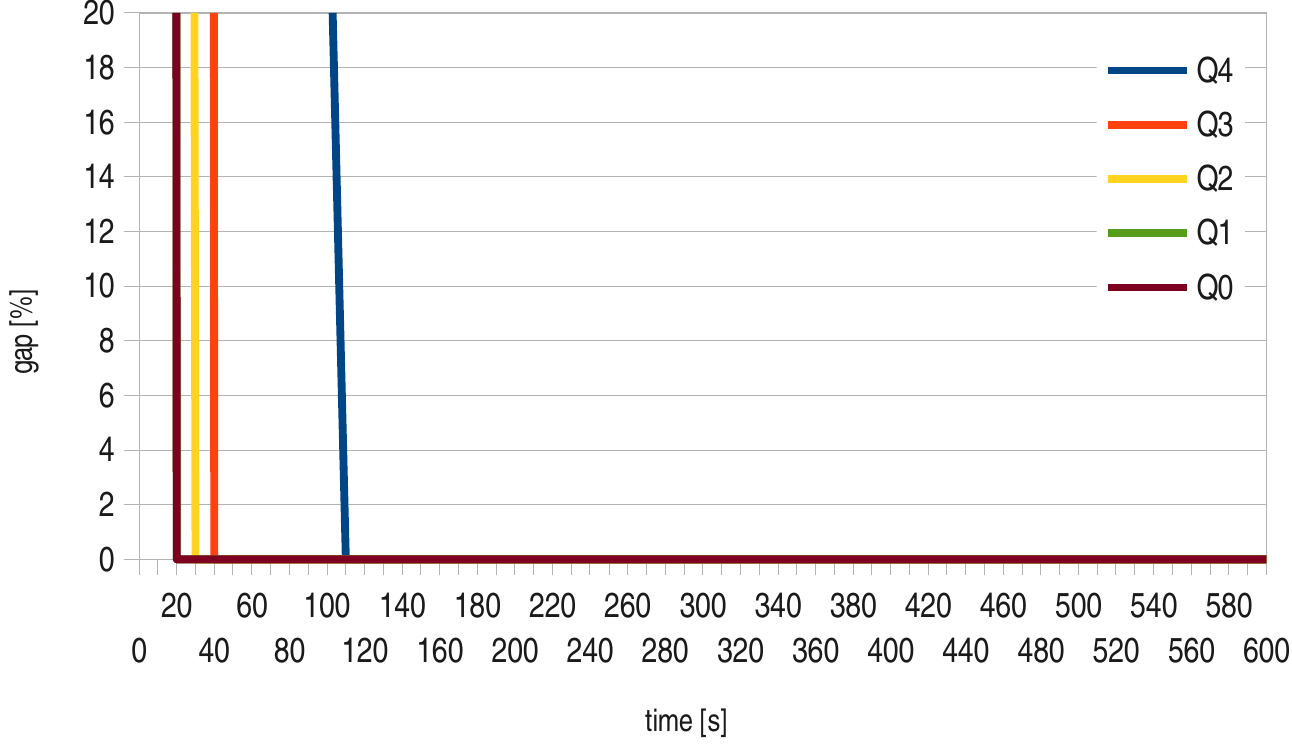}
		\label{fig:ipsc4-sp200}
	}
	\caption[Relative Gap over Time in \acs{IP}-Mode: Spike]{%
	Relative gap over time in \acs{IP}-mode for the 200-vertex
	\emph{Spike}-type polygons.
	}
	\label{fig:ip-sp200-all}
\end{figure}

The remaining test cases, \ie, the remaining polygon classes, confirm our interpretation.
We briefly summarize only the deviating observations.

For the \emph{Orthogonal}-type polygons with 1000 vertices in
Figure~\ref{fig:ip-or1000-all}, \acs{EC} and \acs{SC}3 separation yield an improvement
over using no separation:
The maximum relative gap drops and some gaps reach their \SI{600}{s} levels
earlier.
Joint application of \acs{SC}3 and \acs{EC} provides the best results.
\acs{SC}4 and \acs{SC}3 separation only differ in the extreme cases
of the minimum and the maximum relative gap.

The 1000-vertex \emph{Simple} polygons, see Figure~\ref{fig:ip-si1000-all},
allow a slight improvement of the relative gap with \acs{EC} as well as
\acs{SC}3 separation; joint application yields the best results.
A difference to the other experiments is that the \acs{SC}4 separator
performs slightly better than the \acs{SC}3 separator -- an isolated
observation.

Our separators have no measurable impact on the \emph{Spike}-type polygons,
see Figure~\ref{fig:ip-sp200-all}.
Larger instances of this type of polygon take much time when solving the
first couple of \acp{IP}, which means our separators are triggered late in
the \SI{600}{s} time limit -- or not at all.

\subsection{LP Mode}
\label{sec:lp}

Analogously to the IP mode, we test the separators in LP mode, \ie, Algorithm~\ref{alg:lp}.
The difference to Algorithm~\ref{alg:ip} is that in the primal phase, it solves the LP $\agr(G,W,A)$ instead of the IP.
If a solution of $\agr(G,W,A)$ is feasible for $\agr(G,P,A)$
and if it happens to be binary, it is an upper bound,
otherwise it is discarded and the primal phase is continued.

The challenge of the LP mode is to find a binary solution at all,
because the algorithm might stick to fractional optimal solutions
that are not handled by any cut separator.
Instances unsolved because of this are considered to have an infinite
gap; they result in diagrams in which only the lower quartiles are
visible.

\begin{figure}
	\subfigure[No Cuts]{
		\includegraphics[clip,width=.45\linewidth]%
				{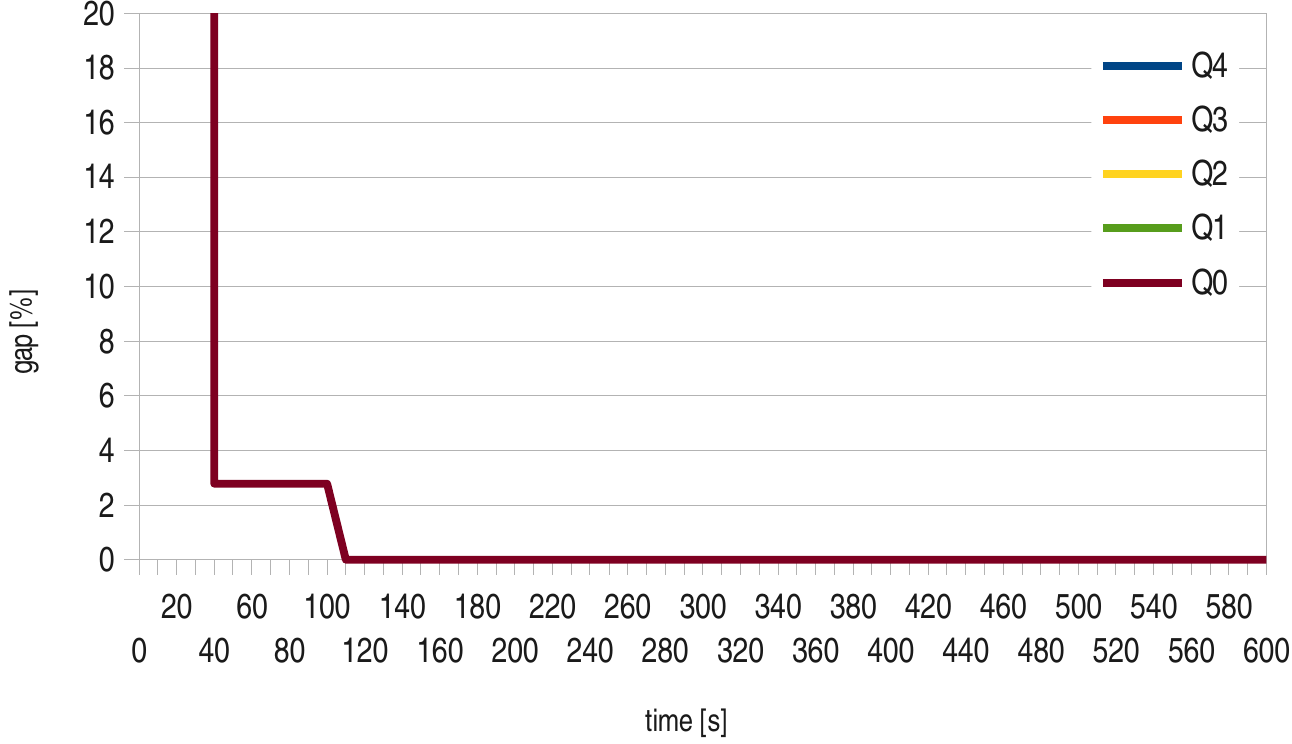}
		\label{fig:lp-ko500}
	}
	\subfigure[\acs{EC}]{
		\includegraphics[clip,width=.45\linewidth]%
				{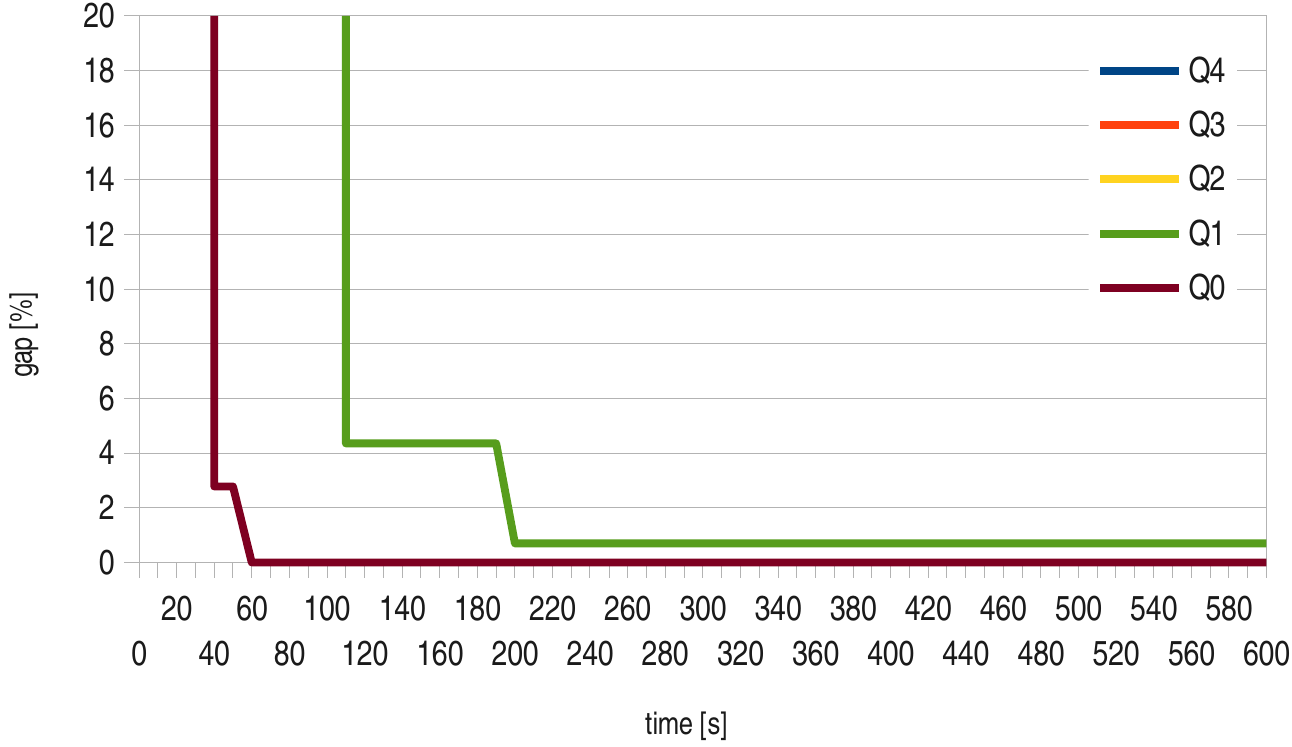}
		\label{fig:lpoc-ko500}
	}
	\subfigure[\acs{SC}3]{
		\includegraphics[clip,width=.45\linewidth]%
				{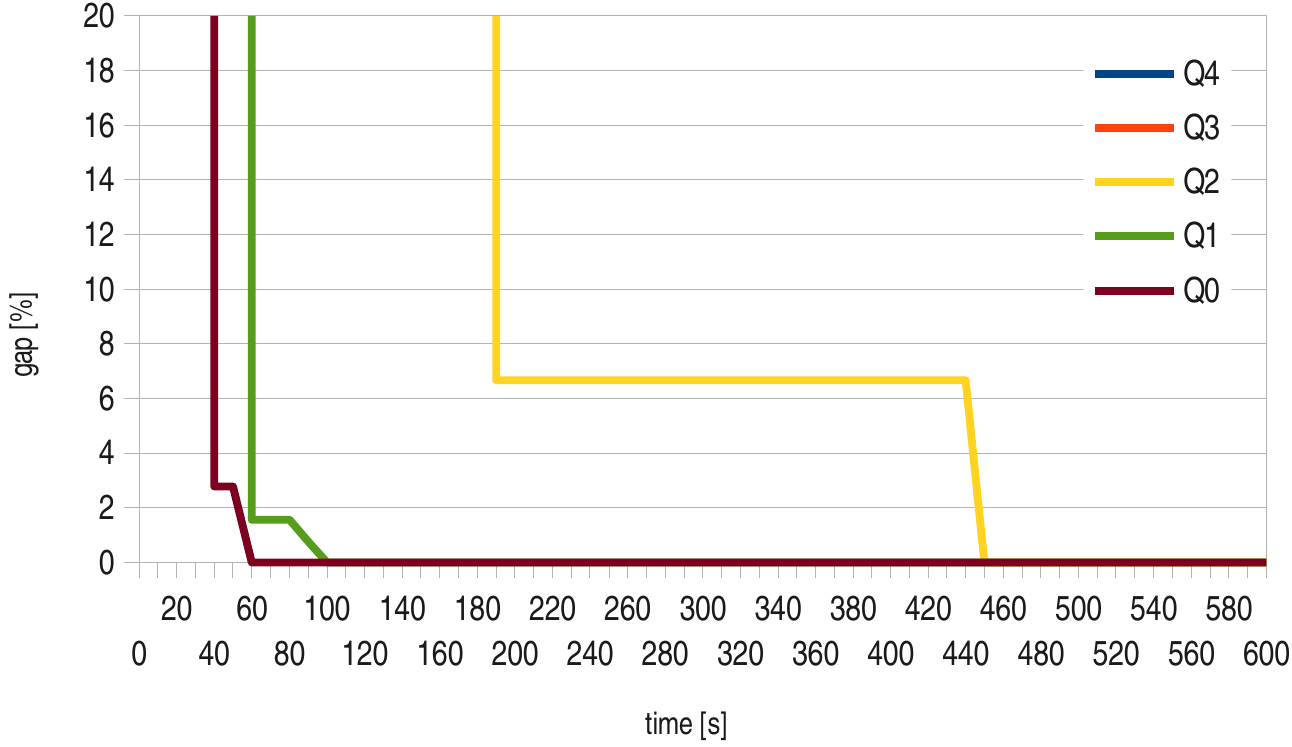}
		\label{fig:lpsc3-ko500}
	}
	\subfigure[\acs{SC}3 and \acs{EC}]{
		\includegraphics[clip,width=.45\linewidth]%
				{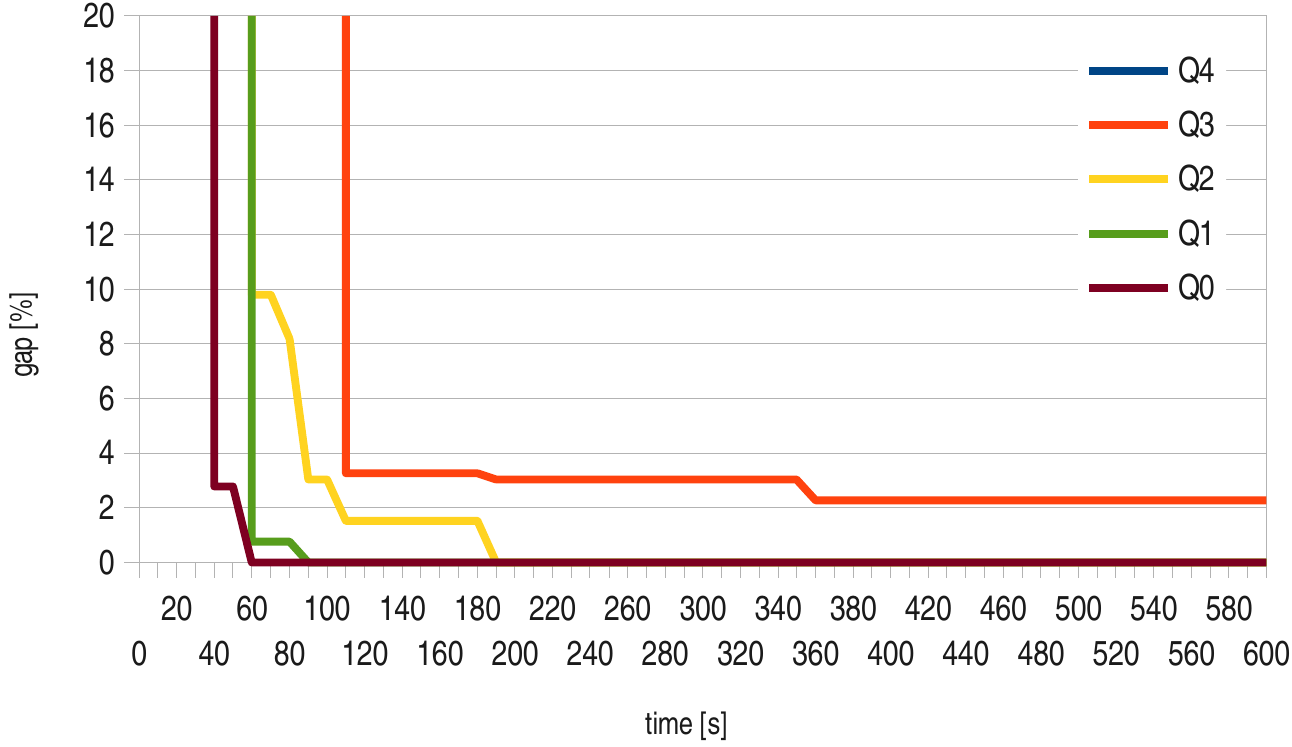}
		\label{fig:lpsc3oc-ko500}
	}
	\subfigure[\acs{SC}4]{
		\includegraphics[clip,width=.45\linewidth]%
				{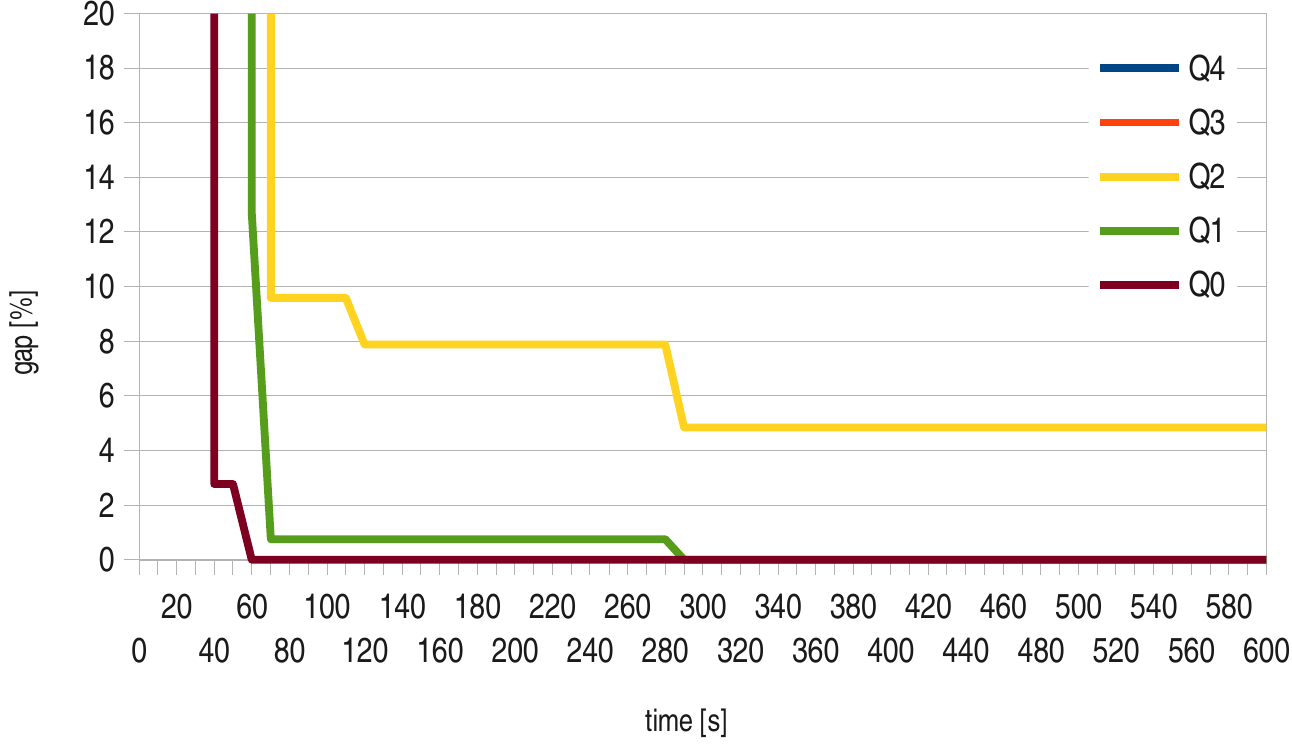}
		\label{fig:lpsc4-ko500}
	}
	\caption[Relative Gap over Time in \acs{LP}-Mode: von Koch]{%
	Relative gap over time in \acs{LP}-mode for the 500-vertex
	\emph{von Koch}-type polygons.
	}
	\label{fig:lp-ko500-all}
\end{figure}

\begin{figure}
	\subfigure[No Cuts]{
		\includegraphics[clip,width=.45\linewidth]%
				{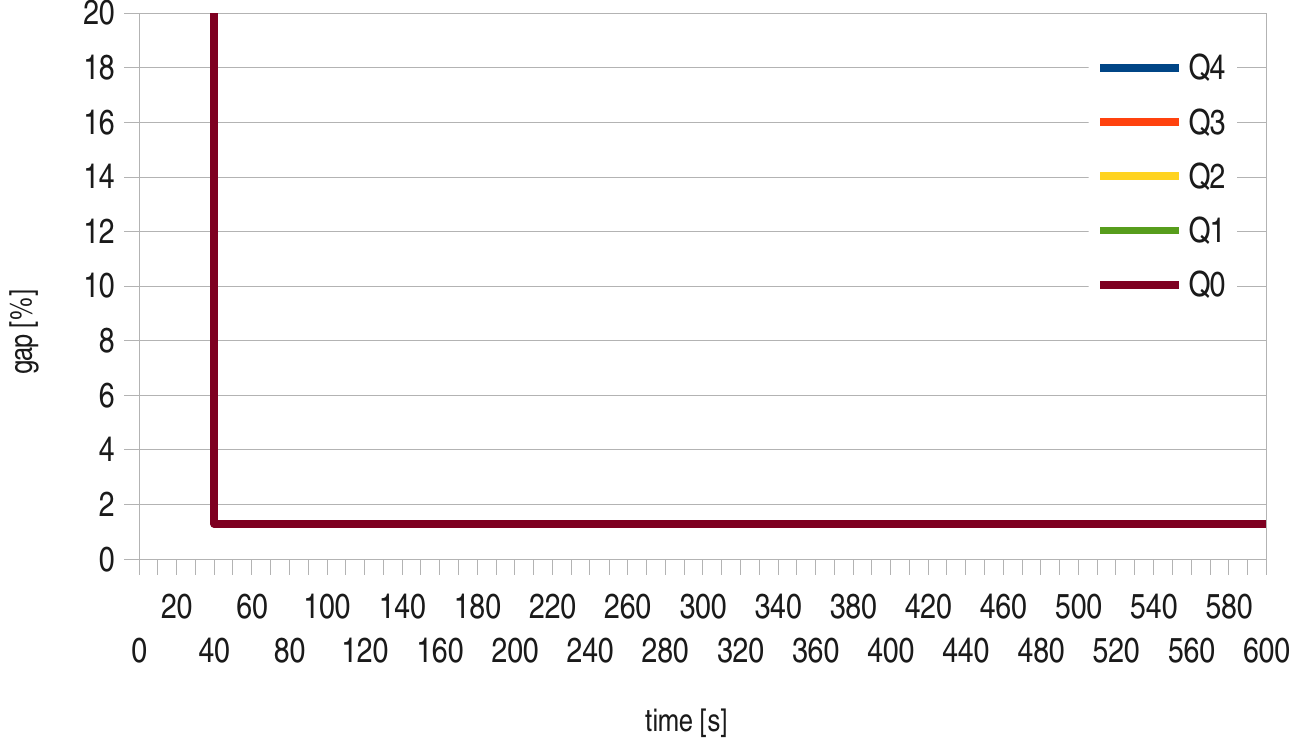}
		\label{fig:lp-or500}
	}
	\subfigure[\acs{EC}]{
		\includegraphics[clip,width=.45\linewidth]%
				{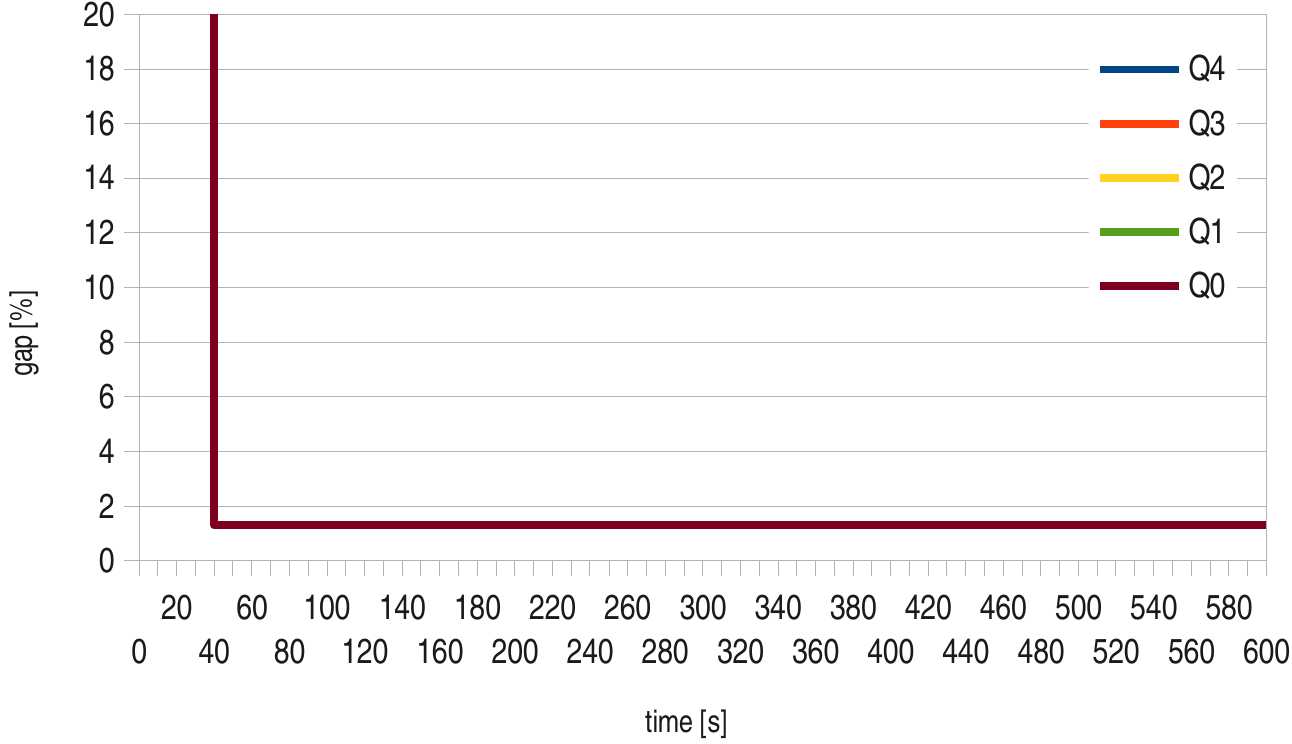}
		\label{fig:lpoc-or500}
	}
	\subfigure[\acs{SC}3]{
		\includegraphics[clip,width=.45\linewidth]%
				{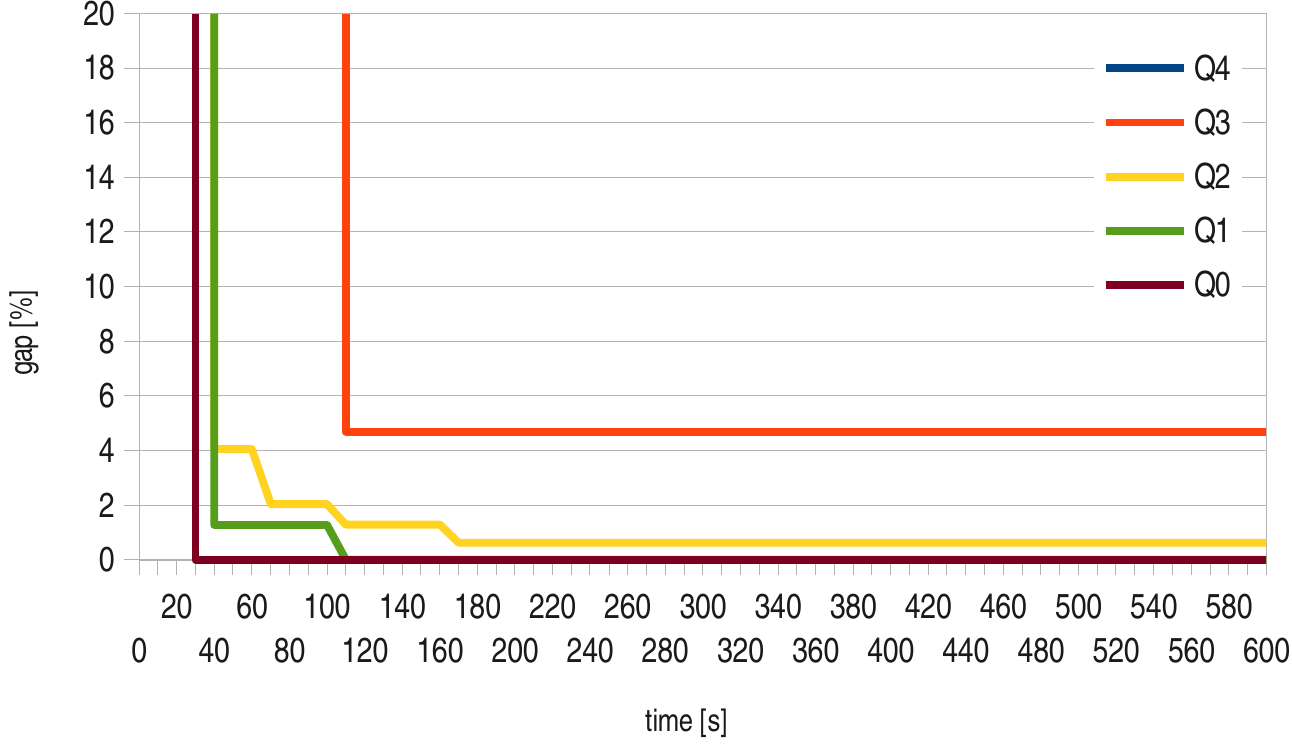}
		\label{fig:lpsc3-or500}
	}
	\subfigure[\acs{SC}3 and \acs{EC}]{
		\includegraphics[clip,width=.45\linewidth]%
				{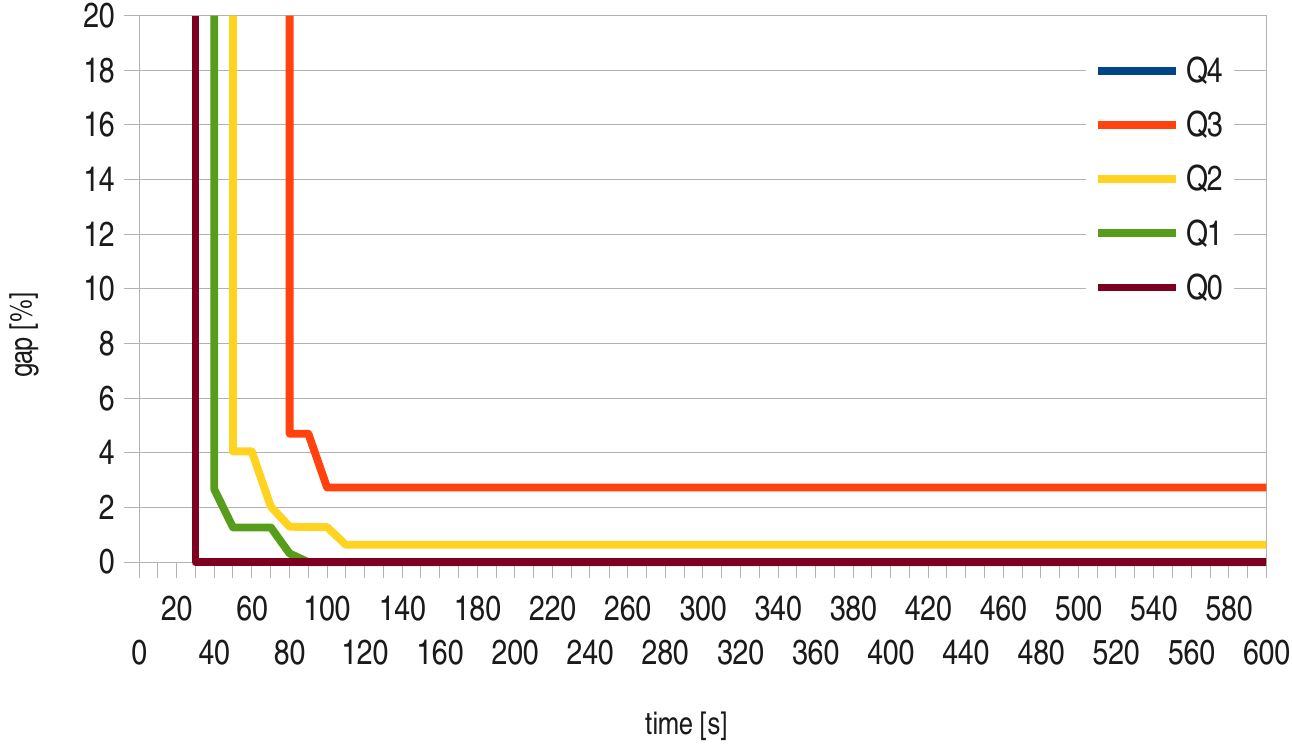}
		\label{fig:lpsc3oc-or500}
	}
	\subfigure[\acs{SC}4]{
		\includegraphics[clip,width=.45\linewidth]%
				{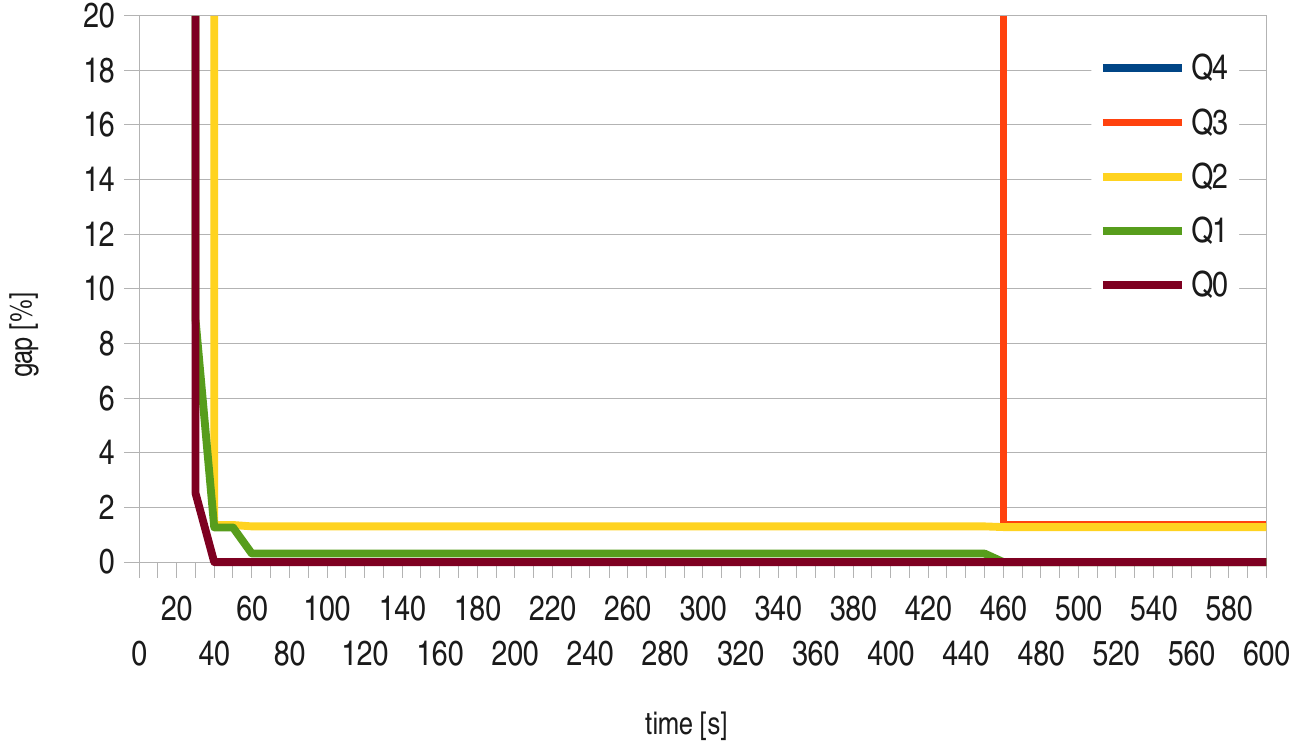}
		\label{fig:lpsc4-or500}
	}
	\caption[Relative Gap over Time in \acs{LP}-Mode: Orthogonal]{%
	Relative gap over time in \acs{LP}-mode for the 500-vertex
	\emph{Orthogonal}-type polygons.
	}
	\label{fig:lp-or500-all}
\end{figure}

\begin{figure}
	\subfigure[No Cuts]{
		\includegraphics[clip,width=.45\linewidth]%
				{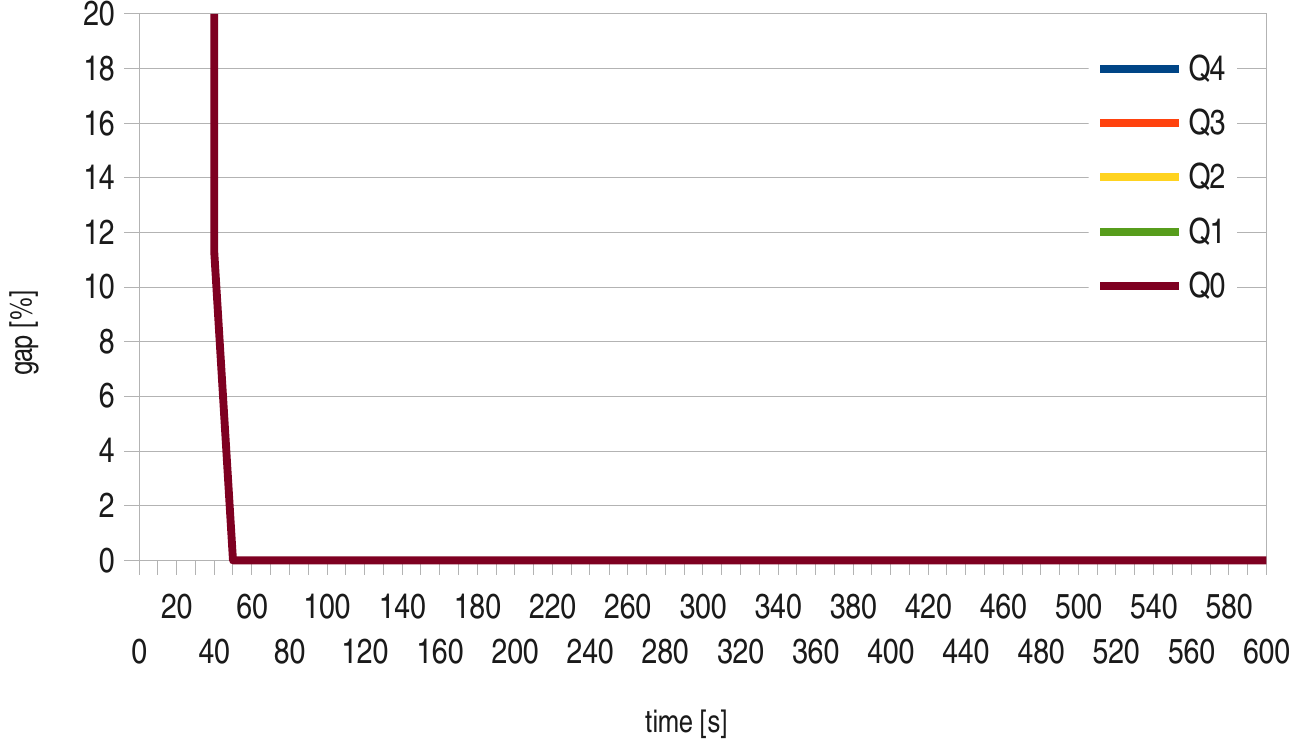}
		\label{fig:lp-si500}
	}
	\subfigure[\acs{EC}]{
		\includegraphics[clip,width=.45\linewidth]%
				{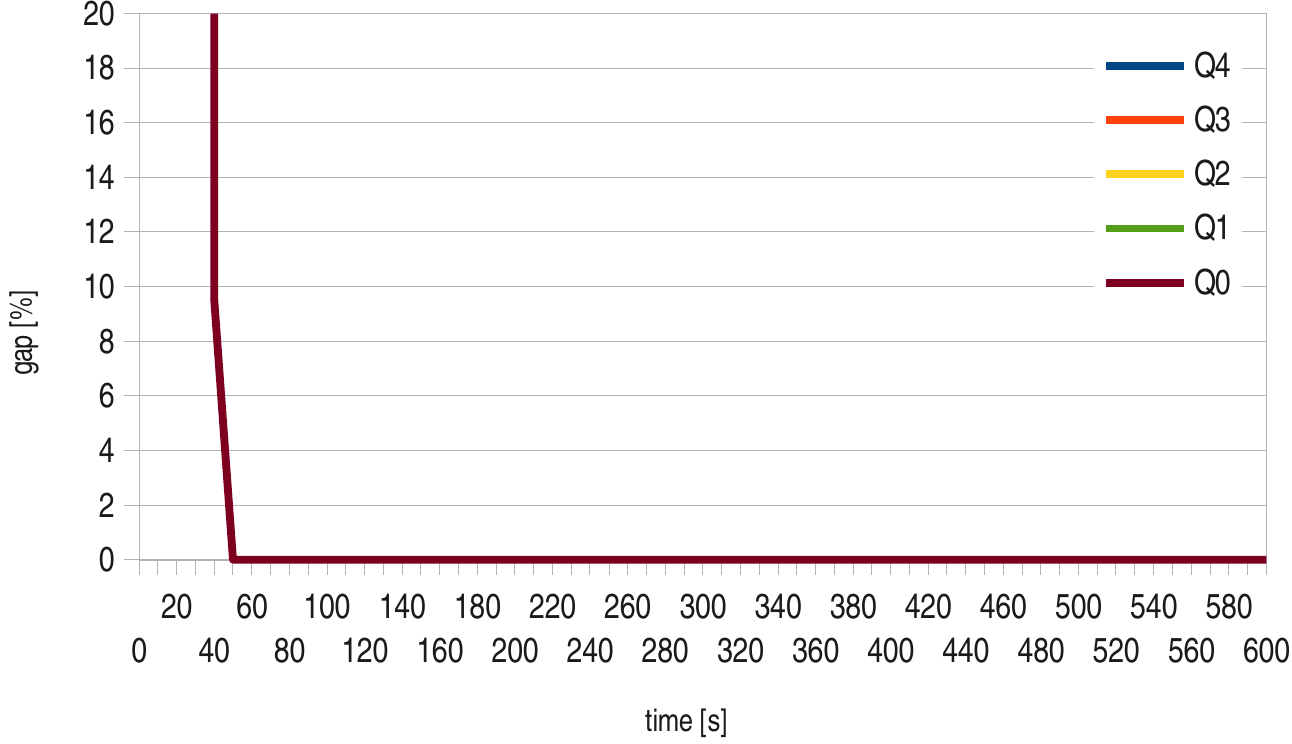}
		\label{fig:lpoc-si500}
	}
	\subfigure[\acs{SC}3]{
		\includegraphics[clip,width=.45\linewidth]%
				{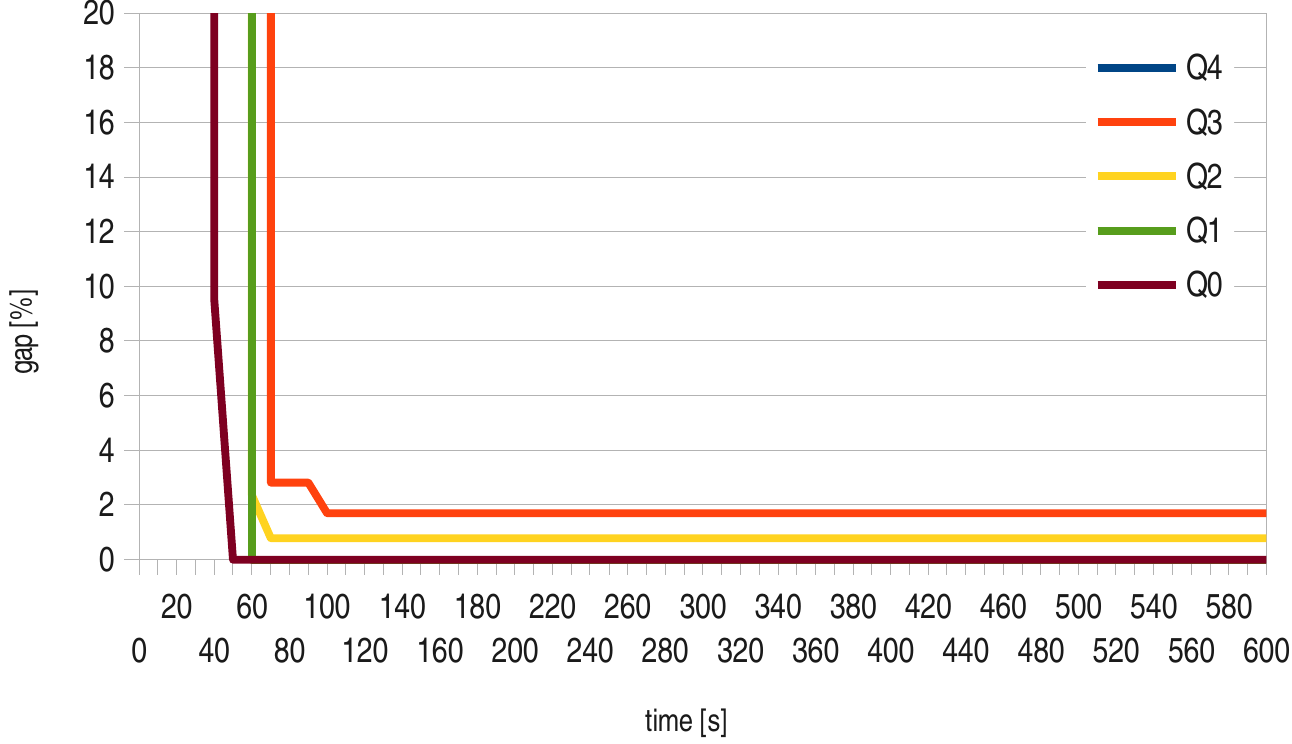}
		\label{fig:lpsc3-si500}
	}
	\subfigure[\acs{SC}3 and \acs{EC}]{
		\includegraphics[clip,width=.45\linewidth]%
				{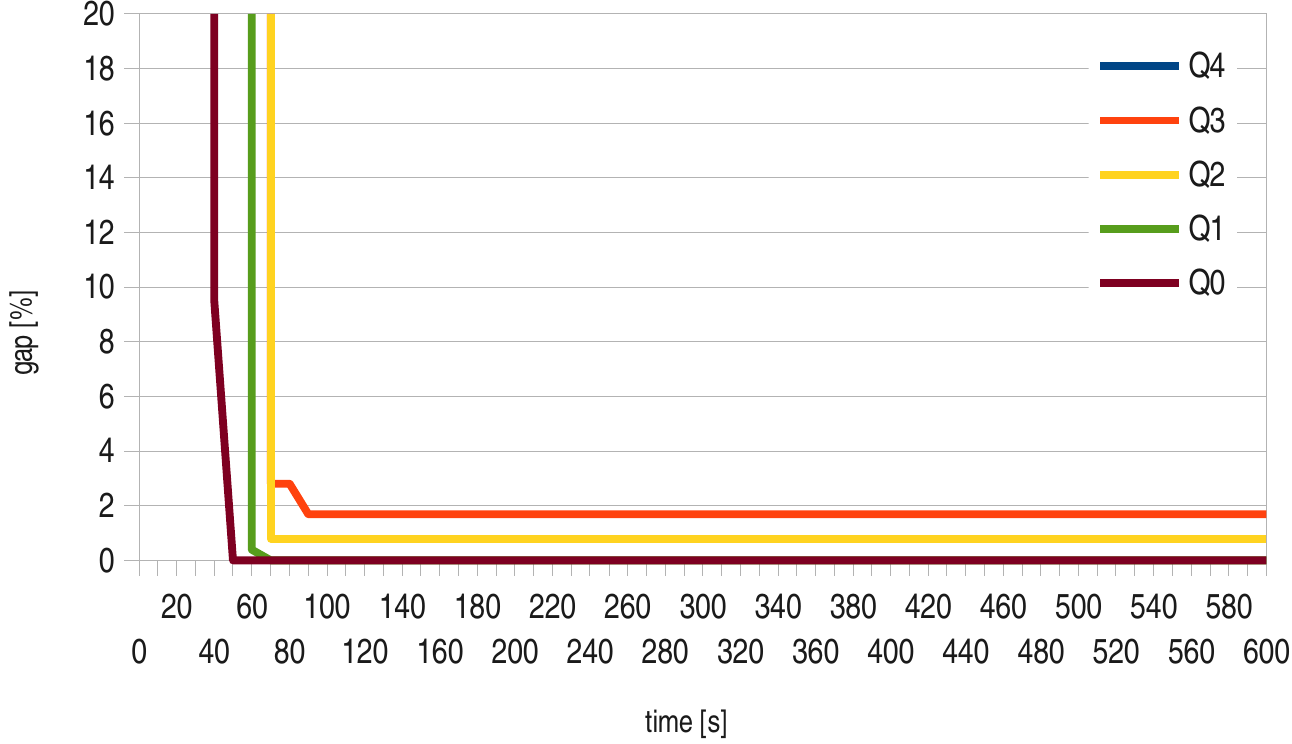}
		\label{fig:lpsc3oc-si500}
	}
	\subfigure[\acs{SC}4]{
		\includegraphics[clip,width=.45\linewidth]%
				{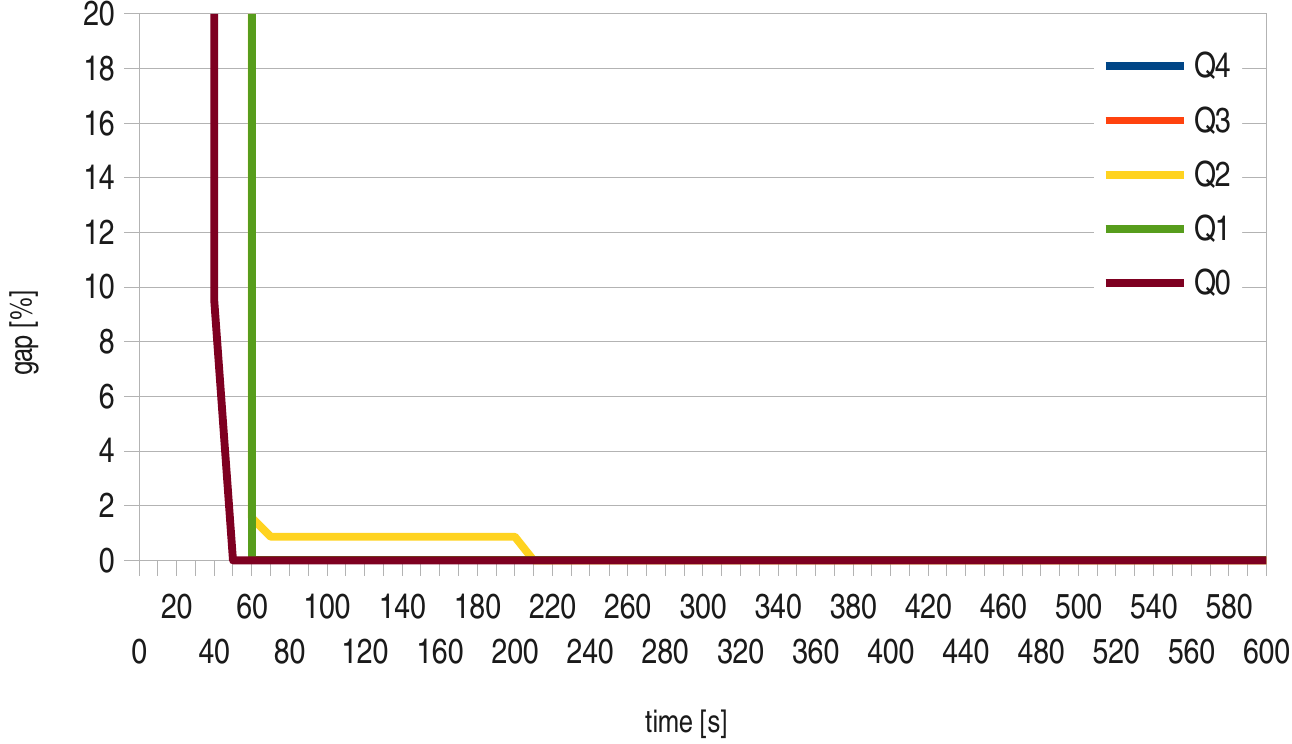}
		\label{fig:lpsc4-si500}
	}
	\caption[Relative Gap over Time in \acs{LP}-Mode: Simple]{%
	Relative gap over time in \acs{LP}-mode for the 500-vertex
	\emph{Simple}-type polygons.
	}
	\label{fig:lp-si500-all}
\end{figure}

\begin{figure}
	\subfigure[No Cuts]{
		\includegraphics[clip,width=.45\linewidth]%
				{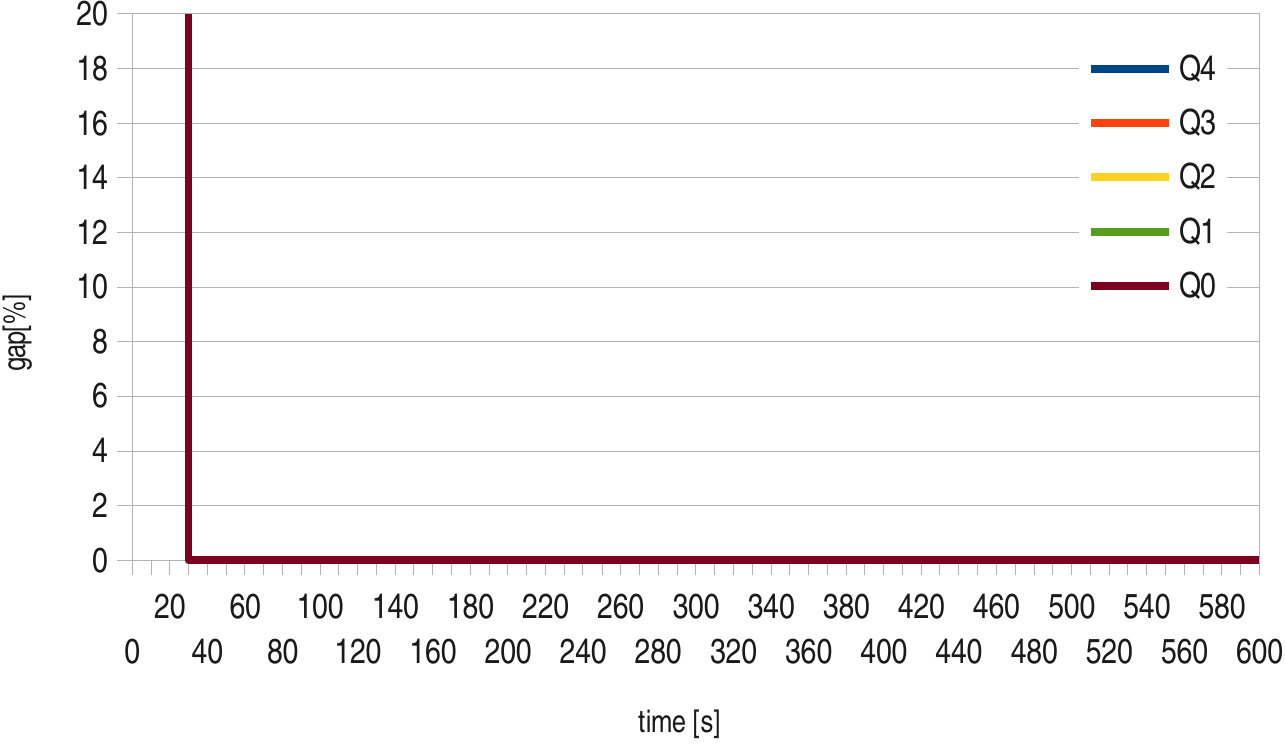}
		\label{fig:lp-sp200}
	}
	\subfigure[\acs{EC}]{
		\includegraphics[clip,width=.45\linewidth]%
				{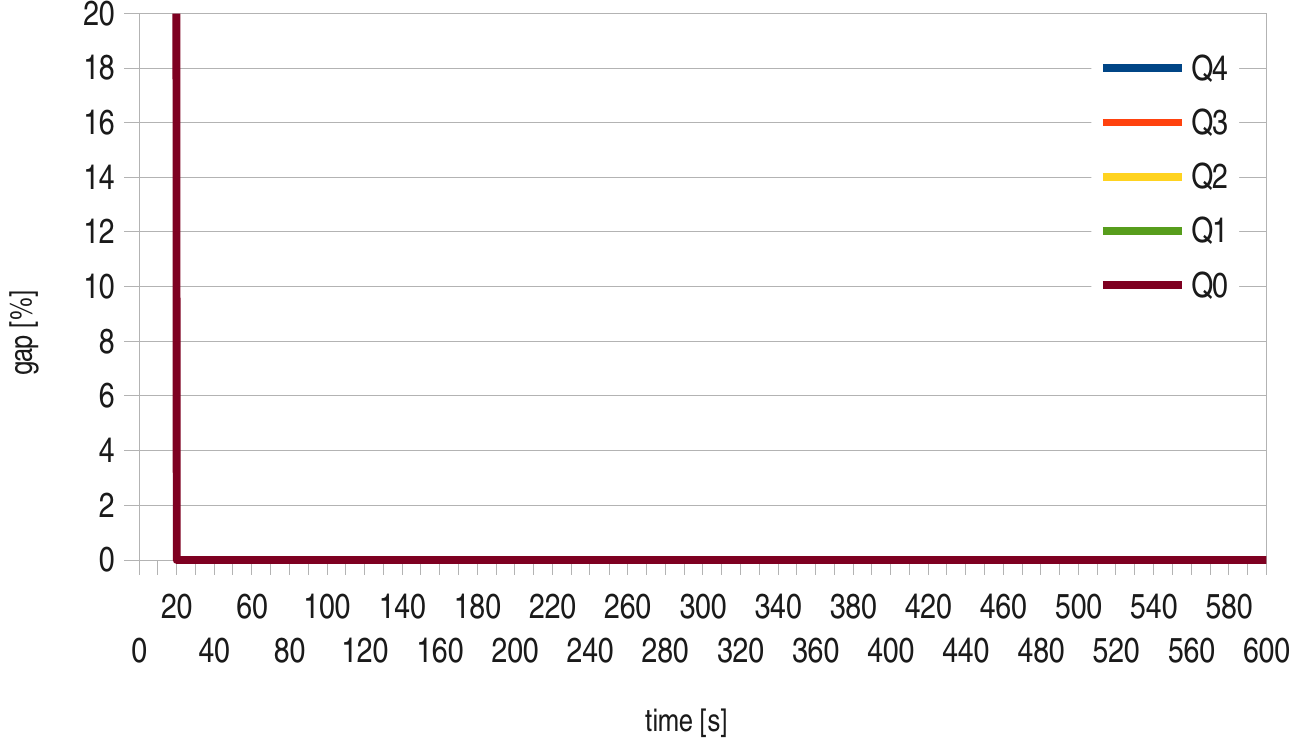}
		\label{fig:lpoc-sp200}
	}
	\subfigure[\acs{SC}3]{
		\includegraphics[clip,width=.45\linewidth]%
				{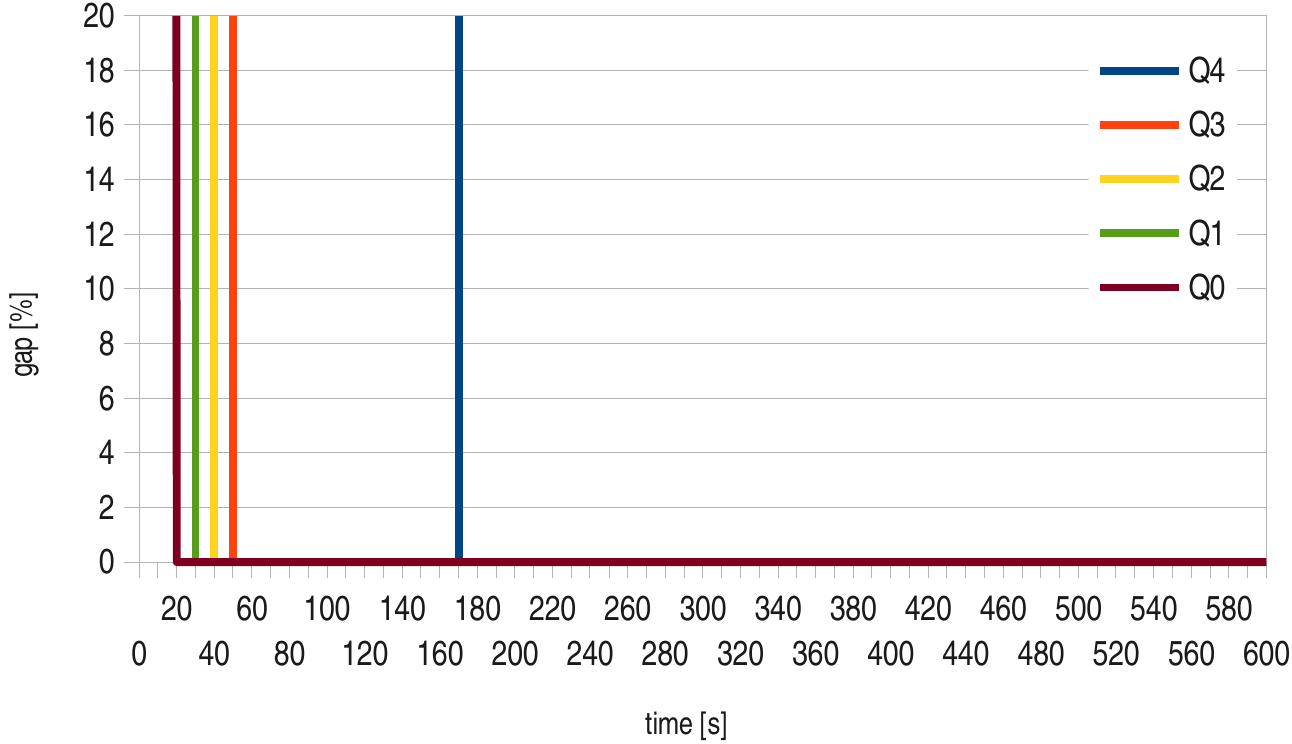}
		\label{fig:lpsc3-sp200}
	}
	\subfigure[\acs{SC}3 and \acs{EC}]{
		\includegraphics[clip,width=.45\linewidth]%
				{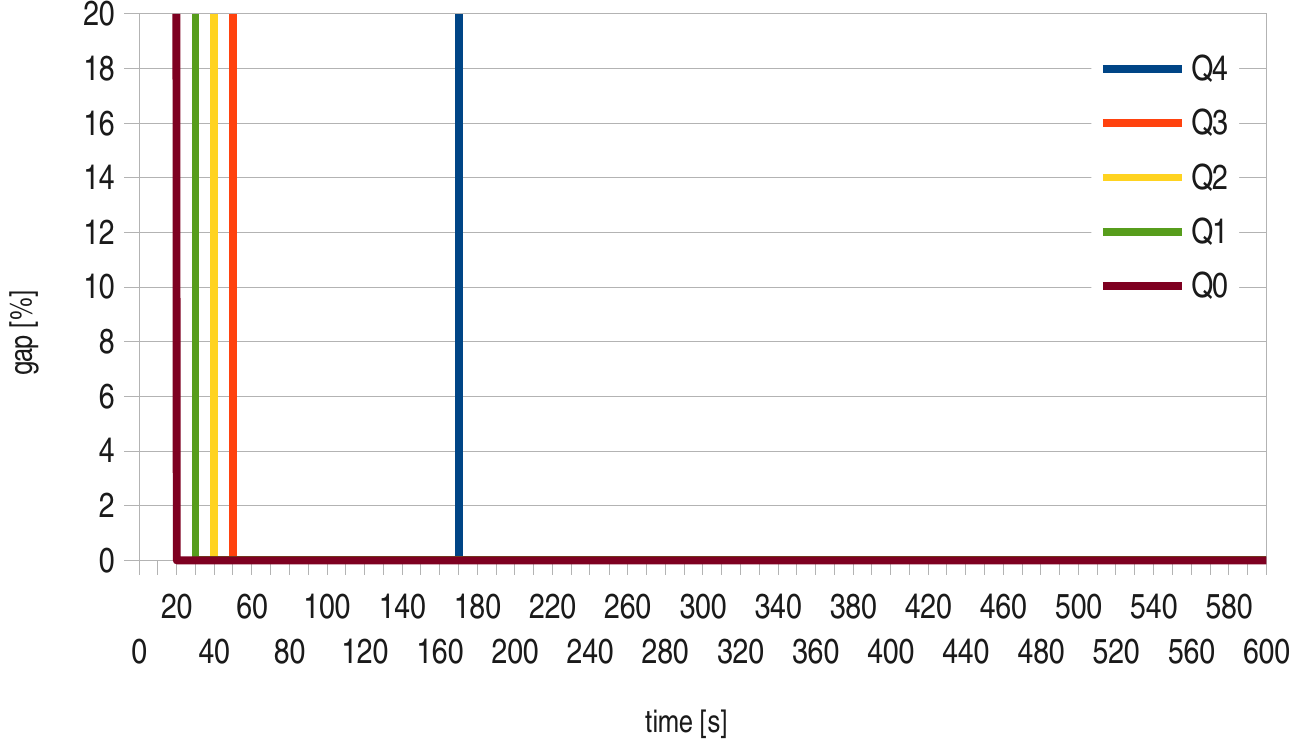}
		\label{fig:lpsc3oc-sp200}
	}
	\subfigure[\acs{SC}4]{
		\includegraphics[clip,width=.45\linewidth]%
				{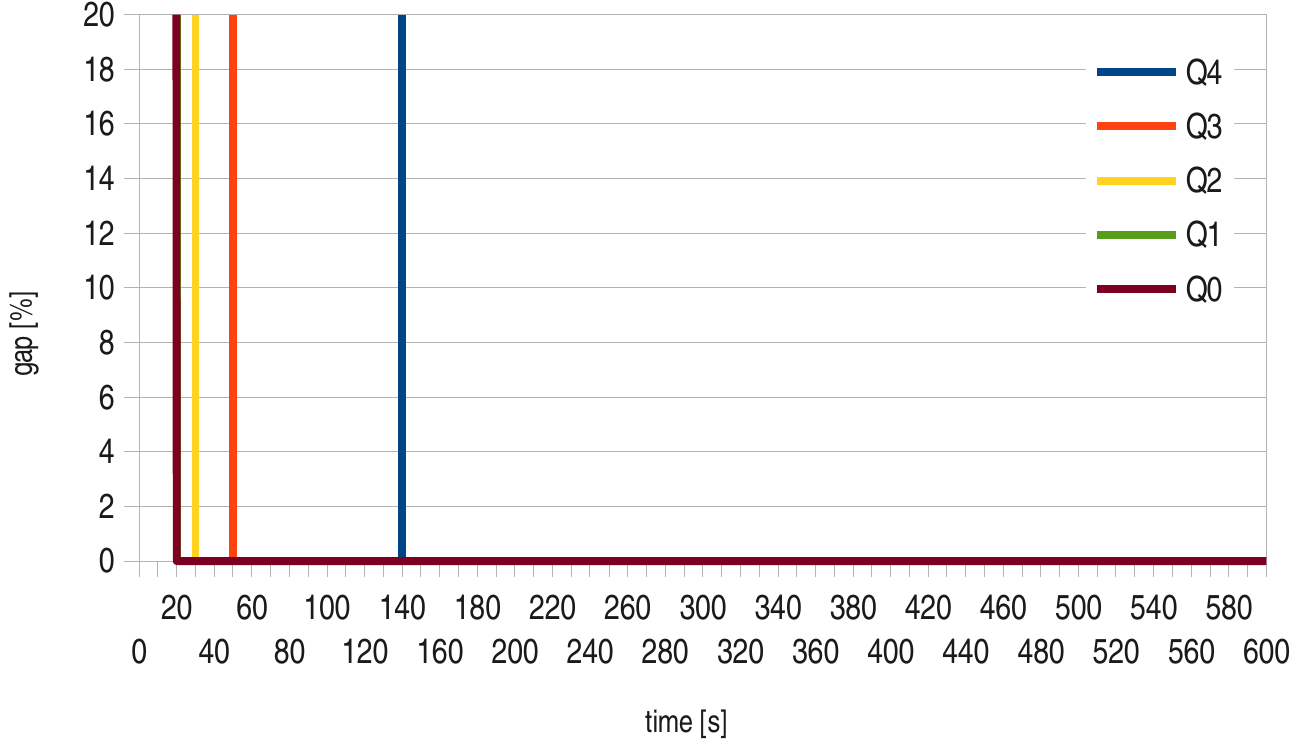}
		\label{fig:lpsc4-sp200}
	}
	\caption[Relative Gap over Time in \acs{LP}-Mode: Spike]{%
	Relative gap over time in \acs{LP}-mode for the 200-vertex
	\emph{Spike}-type polygons.
	}
	\label{fig:lp-sp200-all}
\end{figure}

Figures~\ref{fig:lp-ko500-all},~\ref{fig:lp-or500-all},~\ref{fig:lp-si500-all}
and~\ref{fig:lp-sp200-all} show
the relative gap over time diagrams for the 500-vertex \emph{von Koch},
\emph{Orthogonal} and \emph{Simple}, as well as the 200-vertex
\emph{Spike} polygons.

For the \emph{von Koch} polygons in Figure~\ref{fig:lp-ko500-all}, the \acs{EC}
separator provides a slight improvement, the \acs{SC}3 separator a stronger
improvement; the best result is obtained when using both of them.
\acs{SC}4 separation is weaker than \acs{SC}3 separation.

The EC separator does not improve the situation for the \emph{Orthogonal}
polygons, Figure~\ref{fig:lp-or500-all}, but the SC3 separator boosts solution percentage beyond 75\%.
Joint application of both separators results in even smaller gaps.
SC4 cut separation yields mixed results:
The median relative gap is larger, while the \ac{Q3} gap is smaller but
takes approximately \SI{350}{s} longer to reach its value.
In addition, the \ac{Q1} curve takes much longer to drop to zero as well.

The situation for the \emph{Simple} polygons in Figure~\ref{fig:lp-si500-all}
is as follows.
\acs{EC} separation has no impact on the results, when applied alone as
well as when used jointly.
As above, \acs{SC}3 separation is better than \acs{SC}4 separation,
it solves more instances.

\acs{EC} cuts have no impact on the \emph{Spike} polygons, Figure~\ref{fig:lp-sp200-all}, but \acs{SC}3
helps solving all of them instead of less than 25\%.
In this case, the \acs{SC}4 separator is approximately \SI{10}{s} faster in
the \ac{Q1}, \ac{Q2} and \ac{Q3} quartiles and \SI{30}{s} for the maximum.

We present Table~\ref{tab:lp} that summarizes the solution percentage after
\SI{600}{s} in LP mode as well as the median relative gap.

\begin{table}
	\centering
	\begin{tabular}{|r|c|c|c|c|c|}
		\hline
		Vertices & no cuts & \acs{EC} & \acs{SC}3 & \acs{EC} and \acs{SC}3 & \acs{SC}4 \\
		\hline
		\multicolumn{6}{c}{von Koch} \\
		\hline
		  60 &  90\%, 0.0\% &  90\%, 0.0\% & 100\%, 0.0\% & 100\%, 0.0\% & 100\%, 0.0\% \\
		 200 &  30\%, 0.0\% &  20\%, 0.0\% &  90\%, 0.0\% &  90\%, 0.0\% &  80\%, 0.0\% \\
		 500 &  20\%, 1.4\% &  50\%, 0.0\% &  60\%, 0.0\% &  90\%, 0.0\% &  70\%, 0.0\% \\
		1000 &   0\%, n/a   &  20\%, 0.0\% &  60\%, 0.0\% &  40\%, 0.0\% &  30\%, 0.0\% \\
		\hline
		\multicolumn{6}{c}{Orthogonal} \\
		\hline
		  60 &  80\%, 0.0\% &  80\%, 0.0\% & 100\%, 0.0\% & 100\%, 0.0\% & 100\%, 0.0\% \\
		 200 &  40\%, 7.5\% &  60\%, 1.7\% &  90\%, 0.0\% &  90\%, 0.0\% &  90\%, 0.0\% \\
		 500 &  10\%, 1.3\% &  10\%, 1.3\% &  80\%, 0.0\% &  90\%, 0.0\% &  80\%, 1.3\% \\
		1000 &   0\%, n/a   &   0\%, n/a   &  40\%, 2.0\% &  40\%, 2.1\% &  40\%, 2.0\% \\
		\hline
		\multicolumn{6}{c}{Simple} \\
		\hline
		  60 & 100\%, 0.0\% & 100\%, 0.0\% & 100\%, 0.0\% & 100\%, 0.0\% & 100\%, 0.0\% \\
		 200 &  40\%, 3.5\% &  50\%, 0.0\% & 100\%, 0.0\% &  90\%, 0.0\% & 100\%, 0.0\% \\
		 500 &  10\%, 0.0\% &  10\%, 0.0\% &  80\%, 0.0\% &  80\%, 0.0\% &  70\%, 0.0\% \\
		1000 &   0\%, n/a   &   0\%, n/a   &  50\%, 0.7\% &  50\%, 0.0\% &  50\%, 0.7\% \\
		\hline
		\multicolumn{6}{c}{Spike} \\
		\hline
		  60 & 100\%, 0.0\% & 100\%, 0.0\% & 100\%, 0.0\% & 100\%, 0.0\% & 100\%, 0.0\% \\
		 200 &  10\%, 0.0\% &  20\%, 0.0\% & 100\%, 0.0\% & 100\%, 0.0\% & 100\%, 0.0\% \\
		 500 &   0\%, n/a   &   0\%, n/a   &  20\%, 0.0\% &  20\%, 0.0\% &  10\%, 0.0\% \\
		1000 &   0\%, n/a   &   0\%, n/a   &   0\%, n/a   &   0\%, n/a   &   0\%, n/a   \\
		\hline
	\end{tabular}
	\caption[Percentage solved and Gap in \acs{LP}-Mode]{%
	After \SI{600}{s}, for each polygon/size combination and for every
	tested cut separator combination, this table shows
	for how many percent of the polygons a binary solution was found as well
	as their median relative gap.
	}
	\label{tab:lp}
\end{table}

\section{Conclusion}\label{sec:concl}

In this paper, we have shown how we can exploit both geometric properties and
polyhedral methods of mathematical programming to solve a classical and natural,
but highly challenging problem from computational geometry.

We have shown how to integrate cutting planes into linear programming formulations
of the Art Gallery Problem (AGP), a linear program with a potentially infinite
number of both variables and constraints.
Additionally, we provided three trivial and two non-trivial facets of the AGP
polytope based on {\sc Set Cover} and {\sc Edge Cover}, including a complete list of all AGP facets with coefficients in
$\{0,1,2\}$.

Furthermore, we have exploited the underlying geometric properties of the AGP to
identify a subset of one of our facet classes, that
\begin{enumerate}
\item can be separated in polynomial time, although the general separation problem
is \NP-complete.
\item is theoretically justified by showing that geometry behind the cutting planes
is star-shaped for the cases excluded in separation.
\item is justified by experimental data.
\end{enumerate}

This promises to pave the way for a range of practical AGP applications that
have to deal with additional real-life aspects.
We are optimistic that our basic approach can also be used for other geometric
optimization problems.

%\begin{acknowledgements}
%Kunst
%\end{acknowledgements}

\bibliographystyle{spmpsci}      % mathematics and physical sciences
%\bibliography{refs}   % name your BibTeX data base
%\iffalse

\end{document}